\documentclass[10pt,english]{article}
\usepackage[T1]{fontenc}
\usepackage[utf8]{inputenc}
\usepackage{geometry}
\usepackage{color}
\usepackage{mathtools}
\usepackage{amsmath}
\usepackage{amssymb}
\usepackage{stmaryrd}
\usepackage{graphicx}
\usepackage{bbm}
\usepackage[english]{babel}
\usepackage{amsfonts}
\usepackage{mathrsfs}
\usepackage{amsthm}
\usepackage{bm}
\usepackage{dsfont}
\usepackage{amsfonts}
\usepackage{epsfig}
\usepackage{lscape}
\usepackage{footnote}
\usepackage{amsthm}
\usepackage{wasysym}\usepackage{amsfonts}
\usepackage[english]{babel}
\usepackage{amsfonts,color}
 \usepackage{amsmath,amssymb}
 \usepackage{epsfig}
\usepackage{lscape}
\usepackage{amssymb}
\usepackage{graphicx}
 \usepackage{footnote}
\usepackage[utf8]{inputenc}	
\usepackage{mathtools}
\usepackage{amsthm,amssymb,wasysym}
\usepackage{amsfonts}
\usepackage[english]{babel}
\usepackage{tikz}
\usepackage{empheq}
\usepackage{framed}
\usepackage{bbm}
 \newtheorem{theorem}{Theorem}[section]
 \newtheorem{lemma}[theorem]{Lemma}
 \newtheorem{remark}[theorem]{Remark}
 \newtheorem{proposition}[theorem]{Proposition}

 \newtheorem{definition}[theorem]{Definition}

\DeclareMathOperator{\diag}{diag}
\DeclareMathOperator{\sym}{sym}
\DeclareMathOperator{\skw}{skew}
\DeclareMathOperator{\tr}{tr}

\DeclareMathOperator{\axl}{axl}
\DeclareMathOperator{\anti}{anti}
\DeclareMathOperator{\dev}{dev}
\DeclareMathOperator{\sL}{\mathfrak{sl}}
\DeclareMathOperator{\so}{\mathfrak{so}}
\DeclareMathOperator{\gl}{\mathfrak{gl}}

\DeclareMathOperator{\Curl}{Curl\,}

\newcommand{\Sym}{ {\rm{Sym}} }

\newcommand{\R}{\mathbb{R}}
\newcommand{\id}{{\boldsymbol{\mathbbm{1}}}}

\newcommand{\jump}[1]{ \llbracket #1 \rrbracket }

\def\skew{\text{skew}}
\def\dv{\textrm{div}}

\def\dd{\displaystyle}

\def\dv{\textrm{div}}
\def\Div{\textrm{Div\,}}
\def\curl{\textrm{curl\,}}

\makeatletter
\let\@fnsymbol\@arabic
\allowdisplaybreaks[1]

\begin{document}

\title{A new view on  boundary conditions in the Grioli-Koiter-Mindlin-Toupin
indeterminate couple stress model}

\author{\normalsize{ Angela Madeo\footnote{Angela Madeo, \ \  Laboratoire de G\'{e}nie Civil et Ing\'{e}nierie Environnementale,
Universit\'{e} de Lyon-INSA, B\^{a}timent Coulomb, 69621 Villeurbanne
Cedex, France; and International Center M\&MOCS ``Mathematics and Mechanics
of Complex Systems", Palazzo Caetani,
Cisterna di Latina, Italy,
 email: angela.madeo@insa-lyon.fr}\quad and \quad Ionel-Dumitrel Ghiba\thanks{{Corresponding author:} Ionel-Dumitrel Ghiba, \ \ \ \ Lehrstuhl f\"{u}r Nichtlineare Analysis und Modellierung, Fakult\"{a}t f\"{u}r Mathematik,
Universit\"{a}t Duisburg-Essen, Thea-Leymann Str. 9, 45127 Essen, Germany;  Alexandru Ioan Cuza University of Ia\c si, Department of Mathematics,  Blvd.
Carol I, no. 11, 700506 Ia\c si,
Romania;  Octav Mayer Institute of Mathematics of the
Romanian Academy, Ia\c si Branch,  700505 Ia\c si; and Institute of Solid Mechanics, Romanian Academy, 010141 Bucharest, Romania, email: { dumitrel.ghiba@uaic.ro, dumitrel.ghiba@uni-due.de}}  \quad
and \quad
Patrizio Neff\thanks{Patrizio Neff,  \ \ Head of Lehrstuhl f\"{u}r Nichtlineare Analysis und Modellierung, Fakult\"{a}t f\"{u}r
Mathematik, Universit\"{a}t Duisburg-Essen,  Thea-Leymann Str. 9, 45127 Essen, Germany, email: patrizio.neff@uni-due.de}\quad
and 
\quad Ingo M\"unch\thanks{Ingo M\"unch, Institute for Structural Analysis, Karlsruhe Institute of Technology, Kaiserstr. 12, 76131 Karlsruhe,
Germany, email: ingo.muench@kit.edu}}
}
\maketitle
\begin{abstract}
In this paper we consider the Grioli-Koiter-Mindlin-Toupin linear isotropic indeterminate
couple stress model. Our main aim is to show that, up to now, the boundary
conditions have not been completely understood  for this model. As it turns out, and
to our own surprise, restricting the well known boundary conditions stemming from the strain gradient or second gradient models  to the particular case of the indeterminate couple stress model, does not always reduce to the Grioli-Koiter-Mindlin-Toupin set of accepted boundary
conditions. We present, therefore, a proof of the fact that when specific ``mixed'' kinematical and traction boundary conditions are assigned on the boundary, no ``a priori'' equivalence can be established between Mindlin's and our approach. \\
 \vspace*{0.25cm}
 \\
 \textbf{{Key words:}} generalized continua, second
gradient elasticity, strain gradient elasticity, couple stresses, hyperstresses, indeterminate couple stress model, consistency of mixed  boundary conditions.
\\
\vspace*{0.25cm}
\\
{\bf{AMS 2010 subject classification:} } 74A30, 74A35.
\end{abstract}
\newpage{}

\tableofcontents{}

\newpage{}

\section{Introduction}

Higher gradient elasticity models are nowadays increasingly used to
describe mechanical systems with underlying micro- or nano-structures (see e.g. among many others \cite{Auffray,SciarraCoussy,AlibertSeppFdI,FdIGuarascio,IsolaSciarraVidoliPRSA,FibrousComposites,AltenbachEremeyev,eremeyev2014equilibrium,RinaldiPlacidi,PlacidiDamage,Forest98,Forest02b,Forest03,Lazar06b,Aifantis98}) or to regularize certain ill-posed problems with higher gradient contributions (see e.g. \cite{Mielke09b,Neff_Chelminski07_disloc,ebobisse2015existence,ebobisse2010existence,GhibaNeffExistence,NeffGhibaMadeoLazar}).  Such higher gradient models, together with the more general class of micromorphic models \cite{NeffGhibaMicroModel,Neff_micromorphic_rse_05,Neff_Forest_jel05,Neff_Jeong_ZAMP08,Neff_Jeong_Conformal_ZAMM08,Jeong_Neff_ZAMM08,Neff_Svendsen08,Barbagalo}, have been also proved to be a useful tool for the description of micro-structured materials showing exotic behaviours in the dynamic regime (see e.g. \cite{MadeoNeffGhibaW,MadeoNeffGhibaWZAMM,FdIPlacidiMadeo,PlacidiRosiMadeo,Eremeyev5,rigatos1991stabilizing,gourgiotis2015torsional,gourgiotis2008approach,zisis2014some}).

One
among such higher order models which was introduced at the very beginning is the so called \textit{indeterminate couple
stress model} in which the higher gradient contributions only enter
through gradients on the continuum rotation. We place ourselves in
the context of the linear elastic, isotropic model by choosing a specific form of the quadratic  free elastic energy density.

The question of boundary conditions in higher
gradient elasticity models has been a subject of constant attention.
The matter is that in a higher gradient model, it is not possible
to independently vary the test function and its gradient. Some sort
of split into tangential and normal parts is usually performed  (see e.g. \cite{FdISepp,EdgeIsolaSepp,NthGrad}). This is
well known in general higher gradient models. The boundary conditions in the general case of gradient elasticity and strain gradient elasticity have been settled in  the paper by  Bleustein \cite{bleustein1967note}, see also \cite{maugin2014continuum,madeo2012second,MadeoZAMM2014,dell2014origins}. However, as it turns out, the  boundary conditions obtained by  Tiersten and Bleustein in \cite{TierstenBleustein} with respect to the special case of the indeterminate couple stress model are  not the only possible ones in the framework of the indeterminate couple-stress model.

While the strain gradient framework necessitates to work with a third order hyperstress tensor, the indeterminate couple stress model is apparently simpler: it restricts the form of the curvature energy and allows to work with a second order couple-stress tensor work-conjugate to gradients of rotation. For this apparent simplification  the indeterminate couple stress
model has been  heavily investigated and is still being heavily used as well. A first answer  {as regards boundary conditions}
has been given by Mindlin and Tiersten as well as Koiter \cite{Mindlin62,Koiter64}
who established (correctly) that only 5 geometric and 5 traction boundary
conditions can be prescribed. Their format of boundary conditions
has become the commonly accepted one
for the couple stress model \cite{Mindlin62,Toupin62,Koiter64,anthoine2000effect,Yang02,Park07}, all
these  papers using the same  set of (incomplete) boundary conditions. It seems, to us,
however, that the state of the art in general strain gradient theories
\cite{Mindlin62,Toupin62,Koiter64,Neff_Jeong_IJSS09,park2008variational,anthoine2000effect,MauginVirtualPowers,maugin1980method,javili2013geometrically,Neff_Forest_jel05}
is much more advanced as far as boundary conditions are concerned.

This paper has been motivated by our reading of \cite{Dargush,hadjesfandiari2013fundamental,hadjesfandiari2013skew,hadjesfandiari2011couple,hadjesfandiari2010polar,hadjesfandiari2014evo},
in which the form of traction boundary conditions in the indeterminate couple stress model,
together with an apparently plausible physical postulate lead to unacceptable
conclusions, see \cite{NeffGhibagegenHD}. Therefore, there had to be an underlying problem which
we believe to have tracked down to the hitherto accepted format of
boundary conditions.

\medskip

The main result of this paper,  {consisting in setting up a "stronlgy independent"  set of  boundary conditions} for the couple stress model, has been announced in \cite{MadeoGhibaNeffMunchCRM}. This contribution is now structured as follows: after a subsection fixing the notations used throughout the paper, we outline  some related models in isotropic second gradient elasticity and we give a brief digression concerning differential geometry. In Section \ref{directsection}, we present the equilibrium equations and the constitutive equations of the indeterminate couple stress model as they have been derived in the literature. We also present the classical "weakly independent" boundary conditions proposed by Mindlin and Tiersten \cite{Mindlin62} and the main  {arguments} of their proposal. Since we remark that these boundary conditions are not the only possible ones, in Section \ref{fromgradientbc} we obtain the novel  set of boundary conditions in the indeterminate couple stress model. To this main aim of our paper, we follow two different paths. On the one hand, we consider the indeterminate couple stress model as a special case of the second gradient elasticity model and we derive the "strongly independent"  boundary conditions which follow naturally as restriction of such general framework, see Subsection \ref{sub:The-full-approach}. This kind of approach involves the third order hyperstress tensor as a reminiscence of the second gradient elasticity approach. However, in Subsection \ref{sub:identification-1}, we prove that the equilibrium equations and the boundary conditions may all be rewritten in terms of  Mindlin's second order couple stress tensor. On the other hand, following the line of Mindlin's  argument in combination with some calculations specific to second gradient elasticity model, we set up a "direct approach" which leads  to the same set of boundary conditions with those coming from second gradient elasticity, see Section \ref{sectionincomplete3}. However, these  boundary conditions do not always coincide with those proposed by Mindlin and Tiersten and which are accepted and used until now in the literature. We explain this fact in Subsection \ref{newsubsectionmadeo}  where we explicitly show that, if an "a priori" equivalence can be found in most cases between our approach and Mindlin's one, this is not the case when considering "mixed" boundary conditions, simultaneously assigning the force and the curl of displacement (Mindlin) or the force and the normal derivative of displacement (our approach) on the same portion of the boundary. 

In the appendix, we give some explicit or alternative calculations which are used in the main text and we answer again to the question: what are the missing steps in Mindlin and Tiersten's approach? (see Appendix  \ref{sectionincomplete2}). We do so by using all the arguments provided throughout the paper and pointed out in different circumstances.  We end our paper by some concluding diagrams summarizing our findings.

\subsection{Notational agreements}

In this paper, we denote by $\R^{3\times3}$
the set of real $3\times3$ second order tensors, which will be written with capital
letters.  We denote respectively by $\cdot\:$, $:$ and $\left.\langle  \cdot,\cdot\right.\rangle $
a simple and double contraction and the scalar product between two
tensors of any suitable order%
\footnote{For example, $(A\cdot v)_{i}=A_{ij}v_{j}$, $(A\cdot B)_{ik}=A_{ij}B_{jk}$,
$A:B=A_{ij}B_{ji}$, $(C\cdot B)_{ijk}=C_{ijp}B_{pk}$, $(C:B)_{i}=C_{ijp}B_{pj}$,
$\left.\langle  v,w\right.\rangle =v\cdot w=v_{i}w_{i}$, $\left.\langle  A,B\right.\rangle =A_{ij}B_{ij}$
etc.%
}. Everywhere we adopt the Einstein convention of sum over repeated
indices if not differently specified.  The standard
Euclidean scalar product on $\R^{3\times3}$ is given by $\langle{X},{Y}\rangle_{\R^{3\times3}}=\tr({X\cdot Y^{T}})$,
and thus the Frobenius tensor norm is $\|{X}\|^{2}=\langle{X},{X}\rangle_{\R^{3\times3}}$.
In the following we omit the index $\R^{3},\R^{3\times3}$. The identity
tensor on $\R^{3\times3}$ will be denoted by $\id$, so that $\tr({X})=\langle{X},{\id}\rangle$.
We adopt the usual abbreviations of Lie-algebra theory, i.e., $\so(3):=\{X\in\mathbb{R}^{3\times3}\;|X^{T}=-X\}$
is the Lie-algebra of skew symmetric tensors and $\sL(3):=\{X\in\mathbb{R}^{3\times3}\;|\tr({X})=0\}$
is the Lie-algebra of traceless tensors. For all $X\in\mathbb{R}^{3\times3}$
we set $\sym X=\frac{1}{2}(X^{T}+X)\in\Sym$, $\skw X=\frac{1}{2}(X-X^{T})\in\so(3)$
and the deviatoric part $\dev X=X-\frac{1}{3}\;\tr(X)\,\id\in\sL(3)$
and we have the \emph{orthogonal Cartan-decomposition of the Lie-algebra}
$\gl(3)$
\begin{align}
\gl(3) & =\{\sL(3)\cap\Sym(3)\}\oplus\so(3)\oplus\mathbb{R}\!\cdot\!\id,\qquad\quad X=\dev\sym X+\skw X+\frac{1}{3}\tr(X)\,\id\,.
\end{align}
Throughout this paper (when we do not specify else) Latin subscripts
take the values $1,2,3$. Typical conventions for differential operations
are implied such as comma followed by a subscript to denote the partial
derivative with respect to the corresponding cartesian coordinate.
Here, for
\begin{align}
\overline{A}=\left(\begin{array}{ccc}
0 & -a_{3} & a_{2}\\
a_{3} & 0 & -a_{1}\\
-a_{2} & a_{1} & 0
\end{array}\right)\in\so(3)
\end{align}
we consider the operators $\axl:\so(3)\rightarrow\mathbb{R}^{3}$
and $\anti:\mathbb{R}^{3}\rightarrow\so(3)$ which verify the following identities
\begin{align}
\axl(\overline{A}):=\left(a_{1},a_{2},a_{3}\right)^{T},\quad\quad\overline{A}.\, v=(\axl\overline{A})\times v,\qquad\qquad(\anti(v))_{ij}=-\epsilon_{ijk}v_{k},\quad\quad\forall\, v\in\mathbb{R}^{3},\label{eq:AxlAnti_indices}\\
\notag(\axl\overline{A})_{k}=-\frac{1}{2}\,\epsilon_{ijk}\overline{A}_{ij}=\frac{1}{2}\,\epsilon_{kij}\overline{A}_{ji}\,,\quad\qquad\overline{A}_{ij}=-\epsilon_{ijk}(\axl\overline{A})_{k}=:\anti(\axl\overline{A})_{ij},\quad(a\times b)_{i}=\epsilon_{ijk}a_{j}b_{k},
\end{align}
where $\epsilon_{ijk}$ is the totally antisymmetric Levi-Civita third order permutation
tensor.

We consider a body which occupies a bounded open set $\Omega$ of
the three-dimensional Euclidian space $\R^{3}$ and assume that its
boundary $\partial\Omega$ is a  smooth surface of
class $C^{2}$. An elastic material fills the domain $\Omega\subset\R^{3}$
and we refer the motion of the body to rectangular axes $Ox_{i}$.
\begin{figure}[h!]
\centering %
\begin{minipage}[c]{0.7\linewidth}%
\includegraphics[scale=0.5]{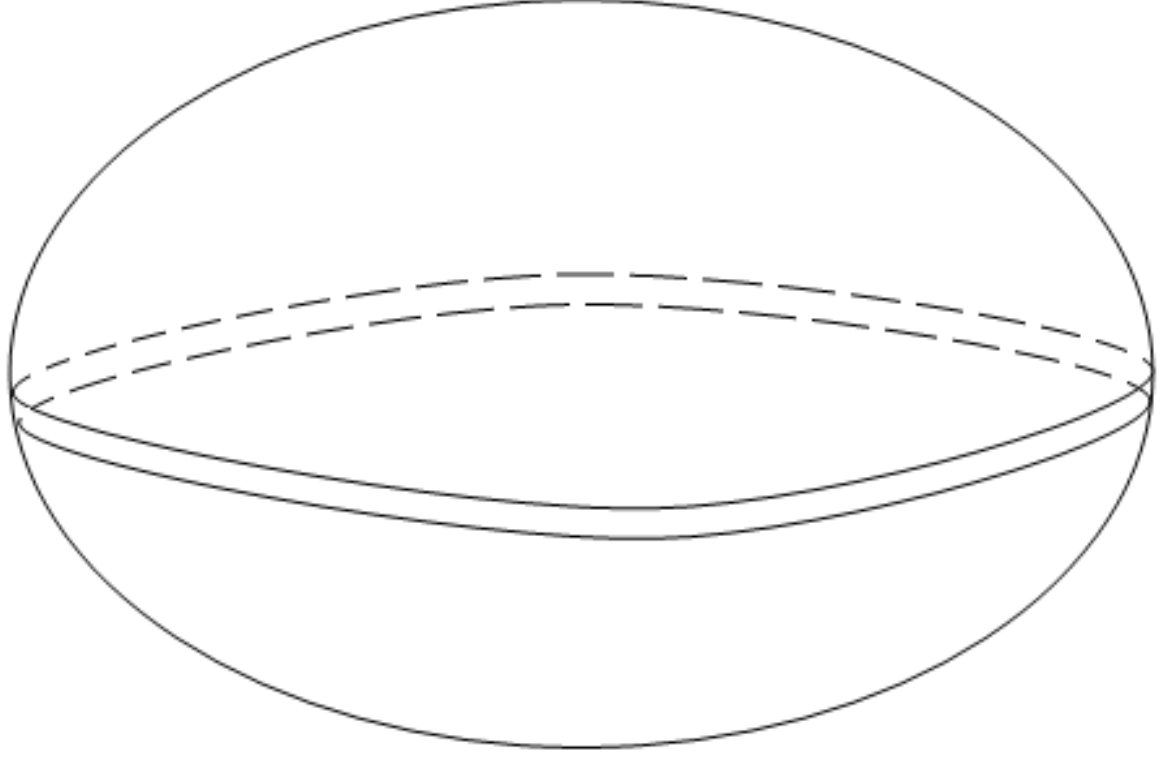} \centering \put(-100,80){$\partial\Omega\setminus\overline{\Gamma}$}
\put(-40,85){$\Omega$} \put(-90,10){$\Gamma$} \put(-90,40){$\partial\Gamma$}
\put(-105,40){\vector(0,-1){20}} \put(-120,20){$\nu^{-}$}
\put(-70,34){\vector(0,1){20}} \put(-65,50){$\nu^{+}$} \protect\protect\protect\protect\caption{{\footnotesize{}{}{}{}{The domain $\Omega\subset\mathbb{R}^{3}$
together with the part $\Gamma\subset\partial\Omega$, where Dirichlet
boundary conditions are prescribed. We need to represent the boundary
conditions on a disjoint union of $\partial\Omega=(\partial\Omega\setminus\overline{\Gamma})\cup\Gamma\cup\partial\Gamma$,
where $\Gamma$ is a open subset of $\partial\Omega$.}}}

{\footnotesize{}{}{}{}\label{capture-bc} }%
\end{minipage}
\end{figure}

With reference to Fig.\ref{capture-bc}, $n$ is the outward unit normal to $\partial \Omega$, $\Gamma$ is an open subset of the boundary $\partial \Omega$,  $\nu^{-}$ is a vector tangential to the surface $\partial\Omega\setminus\overline{\Gamma}$
and which is orthogonal to its boundary $\partial(\partial\Omega\setminus\overline{\Gamma})$,
$\tau^{-}=n\times\nu^{-}$ is the tangent to the curve $\partial(\partial\Omega\setminus\overline{\Gamma})$
with respect to the orientation on $\partial\Omega\setminus\overline{\Gamma}$ given by the outward unit normal $n$ to this surface.
Similarly, $\nu^{+}$ is a vector tangential to the surface $\Gamma$
and which is orthogonal to its boundary $\partial\Gamma$, $\tau^{+}=n\times\nu^{+}$
is the tangent to the curve $\partial\Gamma$ with respect to the
orientation on $\Gamma$.

In the following, given any vector field $a$ defined on the boundary
$\partial\Omega$ we will also set
\begin{equation}
\jump{\left.\langle  a,\nu\right.\rangle }:=\left.\langle  a^{+},\nu^{+}\right.\rangle +\left.\langle  a^{-},\nu^{-}\right.\rangle =\left.\langle  a^{+},\nu\right.\rangle -\left.\langle  a^{-},\nu\right.\rangle =\left.\langle  a^{+}-a^{-},\nu\right.\rangle ,\label{eq:Jump}
\end{equation}
which defines a measure of the jump of $a$ through the line $\partial\Gamma$, where $\nu:=\nu^+=-\nu^-$ and
\[
[\cdot]^{-}:=\hspace*{0cm}\dd\lim\limits _{\footnotesize{\begin{array}{c}
x\in\partial\Omega\setminus\overline{\Gamma}\\
\ x\rightarrow\partial\Gamma
\end{array}}}\hspace*{0cm}[\cdot], \qquad
[\cdot]^{+}:=\hspace*{-0.2cm}\dd\lim\limits _{\footnotesize{\begin{array}{c}
x\in\Gamma\\
\ x\rightarrow\partial\Gamma
\end{array}}}\hspace*{-0.2cm}[\cdot].
\]
We are assuming here that $\partial\Omega$ is a smooth
surface. Hence, there are no geometric singularities of the boundary. The jump
$\jump{\cdot}$ arises only as a consequence of possible discontinuities
which follows from the prescribed boundary conditions on $\Gamma$
and ${\partial\Omega\setminus\overline{\Gamma}}$. Nevertheless, if one would like to explicitly consider continua with non-smooth boundaries, the jump conditions to be imposed at the edges of the boundary would be formally the same to those that we will present in the remainder of this paper, with the precision that the jump would in this case indicate a true jump across a geometrical discontinuity of the surface.

The usual Lebesgue spaces of square integrable functions, vector or
tensor fields on $\Omega$ with values in $\mathbb{R}$, $\mathbb{R}^{3}$
or $\mathbb{R}^{3\times3}$, respectively will be denoted by $L^{2}(\Omega)$.
Moreover, we introduce the standard Sobolev spaces \cite{Adams75,Raviart79,Leis86}
\begin{align}
\begin{array}{ll}
{\rm H}^{1}(\Omega)=\{u\in L^{2}(\Omega)\,|\,{\rm grad}\, u\in L^{2}(\Omega)\}, & \qquad \|u\|_{{\rm H}^{1}(\Omega)}^{2}:=\|u\|_{L^{2}(\Omega)}^{2}+\|{\rm grad}\, u\|_{L^{2}(\Omega)}^{2}\,,\vspace{1.5mm}\\
{\rm H}({\rm curl};\Omega)=\{v\in L^{2}(\Omega)\,|\,{\rm curl}\, v\in L^{2}(\Omega)\}, &\qquad \|v\|_{{\rm H}({\rm curl};\Omega)}^{2}:=\|v\|_{L^{2}(\Omega)}^{2}+\|{\rm curl}\, v\|_{L^{2}(\Omega)}^{2}\,,\vspace{1.5mm}\\
{\rm H}({\rm div};\Omega)=\{v\in L^{2}(\Omega)\,|\,{\rm div}\, v\in L^{2}(\Omega)\}, &\qquad \|v\|_{{\rm H}({\rm div};\Omega)}^{2}:=\|v\|_{L^{2}(\Omega)}^{2}+\|{\rm div}\, v\|_{L^{2}(\Omega)}^{2}\,,
\end{array}
\end{align}
of functions $u$ or vector fields $v$, respectively.

For vector fields $v$ with components in ${\rm H}^{1}(\Omega)$,
i.e. $v=\left(v_{1},v_{2},v_{3}\right)^{T}\,,v_{i}\in{\rm H}^{1}(\Omega),$
we define \break $\nabla\, v=\left((\nabla\, v_{1})^{T},(\nabla\, v_{2})^{T},(\nabla\, v_{3})^{T}\right)^{T}$,
while for tensor fields $P$ with rows in ${\rm H}({\rm curl}\,;\Omega)$,
resp. ${\rm H}({\rm div}\,;\Omega)$, i.e. \break $P=\left(P_{1}^{T},P_{2}^{T},P_{3}^{T}\right)$,
$P_{i}\in{\rm H}({\rm curl}\,;\Omega)$ resp. $P_{i}\in{\rm H}({\rm div}\,;\Omega)$
we define ${\rm Curl}\, P=\left(({\rm curl}\, P_{1})^{T},({\rm curl}\, P_{2})^{T},({\rm curl}\, P_{3})^{T}\right)^{T},$
${\rm Div}\, P=\left({\rm div}\, P_{1},{\rm div}\, P_{2},{\rm div}\, P_{3}\right)^{T}.$
The corresponding Sobolev-spaces will be denoted by
\[
{\rm H}^{1}(\Omega),\qquad {\rm H}^{1}(\Div;\Omega),\qquad {\rm H}^{1}(\Curl;\Omega).
\]

\subsection{The indeterminate couple stress model}
In the indeterminate couple stress model we consider that the elastic energy
is given in the form
\begin{align}
W=W_{{\rm lin}}(\nabla u)+W_{{\rm curv}}(\nabla[\axl(\skw\nabla u)])=W_{{\rm lin}}(\nabla u)+\widetilde{W}_{{\rm curv}}(\nabla\curl u),
\end{align}
where
\begin{align}
W_{{\rm lin}}(\nabla u) & =\,\mu\,\|\sym\nabla u\|^{2}+\frac{\lambda}{2}\,[\tr(\sym\nabla u)]^{2}=\,\mu\,\|\dev\sym\nabla u\|^{2}+\frac{\kappa}{2}\,[\tr(\sym\nabla u)]^{2}.
\end{align}
Here, $\mu>0$ is the infinitesimal shear modulus,
$\kappa=\frac{2\mu+3\lambda}{3}>0$ is the infinitesimal bulk modulus with $\lambda$ the first Lam\'{e} constant.
In order to discuss the form of the curvature energy $W_{{\rm curv}}(\nabla[\axl(\skw\nabla u)])$,
let us recall some variants of the linear isotropic indeterminate couple stress models. 
Some parts of this classification have already been included in the paper \cite{NeffGhibaMadeoMunch} but  we include them also here for the sake of completeness:
\begin{itemize}
\item \textbf{the indeterminate couple stress model} (Grioli-Koiter-Mindlin-Toupin
model) \cite{Grioli60,Aero61,Koiter64,Mindlin62,Toupin64,Sokolowski72,grioli2003microstructures}
in which the higher derivatives (apparently) appear only through derivatives
of the infinitesimal continuum rotation $\curl u$. Hence, the curvature
energy has the equivalent forms
\begin{align}
W_{{\rm curv}}(\nabla[\axl(\skw\nabla u)]) & =\frac{\alpha_{1}}{4}\,\|\sym\nabla\curl\, u\|^{2}+\frac{\alpha_{2}}{4}\,\|\skw\nabla\curl\, u\|^{2}\qquad\notag\\
 & =\alpha_{1}\,\|\sym\nabla[\axl(\skw\nabla u)]\|^{2}+\alpha_{2}\,\|\skw\nabla[\axl(\skw\nabla u)]\|^{2}  \label{energyindet}\\
 & =\frac{\alpha_{1}}{4}\,\|\dev\sym\nabla\curl\, u\|^{2}+\frac{\alpha_{2}}{4}\,\|\skw\nabla\curl\, u\|^{2}.\qquad\ \notag
\end{align}
Here, we have used the identities
\begin{align}
\ \quad2\,\axl(\skew\nabla u) & =\curl u,\qquad\tr[\nabla[\axl(\skw\nabla u)]]=\frac{1}{2}\tr[\nabla[\curl u]]=\frac{1}{2}{\rm div}[\curl u]=0,
\end{align}
together with the fact that $\nabla \curl u$ is a trace-free second order tensor and hence so is $\sym \nabla \curl u$. This implies that $\dev \sym \nabla \curl u=\sym\nabla \curl u$.
Although this energy admits the equivalent forms \eqref{energyindet}$_{1}$ 
and \eqref{energyindet}$_{3}$, the equations and the boundary value problem
of the indeterminate couple stress model is usually formulated only
using the form \eqref{energyindet}$_{1}$ of the energy. Hence, we may individuate
one of the aims of the present paper in the fact that we want to formulate the boundary value
problem for \textbf{the indeterminate couple stress model using the
alternative form \eqref{energyindet}$_{2}$ of the energy} of the Grioli-Koiter-Mindlin-Toupin
model (see Section \ref{directsection}). We also remark that the spherical part of the couple stress
tensor is zero since $\tr(\nabla\curl u)={\rm div}(\curl u)=0$.
In order to prove the pointwise uniform positive definiteness it is assumed, following \cite{Koiter64}, that $\alpha_1>0, \alpha_2>0$   {(corresponds to $-1<\eta:=\frac{\alpha_1-\alpha_2}{\alpha_1+\alpha_2}<1$ in the notation of \cite{Koiter64}).}
Note that pointwise uniform positivity is often assumed \cite{Koiter64}
when deriving analytical solutions for simple boundary value problems
because it allows to invert the couple stress-curvature relation.
We have shown elsewhere \cite{NeffGhibaMadeoMunch} that  pointwise positive definiteness is
not necessary for well-posedness.
\item In this setting, {\bf Grioli} \cite{Grioli60,grioli2003microstructures} (see also Fleck \cite{Fleck93,Fleck97,Fleck01}) initially considered only the choice $\alpha_1=\alpha_2$. In fact, the energy originally proposed by Grioli \cite{Grioli60} is
{\begin{align}\label{KMTe0}
	W_{\rm curv}(D^2 u)&=\mu\,L_c^2\,\frac{\alpha_1}{4}\left[\, \, \|\nabla (\curl\, u)\|^2+\eta \,\tr[ (\nabla (\curl\, u))^2]\right]\notag\\
	&=\mu\,L_c^2\,\frac{\alpha_1}{4}[ \|\dev \sym \nabla [\axl (\skw \nabla u)]\|^2+\| \skw \nabla [\axl (\skw \nabla u)]\|^2\notag\\&\qquad \qquad +\eta\,\langle \nabla [\axl (\skw \nabla u)],(\nabla [\axl (\skw \nabla u)])^T\rangle]
	\\\notag
	&=\mu\,L_c^2\,\frac{\alpha_1}{4}\Big[ \|\dev \sym \nabla (\curl\, u)\|^2+\| \skw \nabla (\curl\, u)\|^2\\
	&\quad\qquad \qquad+\eta\,\langle \nabla (\curl\, u),(\nabla (\curl\, u))^T\rangle\Big]\notag
	\\
	&=\mu\,L_c^2\,\frac{\alpha_1}{4}\left[(1+\eta)\, \|\dev \sym \nabla (\curl\, u)\|^2+(1-\eta)\,\| \skw \nabla (\curl\, u)\|^2\right].\notag
	\end{align}}
\\
Mindlin \cite[p. 425]{Mindlin62}   {(with $\eta=0$)} explained the relations between  Toupin's constitutive equations \cite{Toupin62} and  Grioli's \cite{Grioli60} constitutive equations and concluded that the obtained equations in the linearized theory are identical, since the extra constitutive parameter $\eta$ of Grioli's model does not explicitly appear in the equations of motion  but enters only the boundary conditions,   {since
	$\nabla \axl(\skw \nabla u)=[\Curl(\sym \nabla u)]^T$,  ${\Div}\Curl(\cdot)=0$, and 
	$${\rm Div}\{\anti {\rm Div}[\nabla \axl(\skw \nabla u)]^T\}={\rm Div}\{\anti {\rm Div}[\Curl(\sym \nabla u)]\}={\rm Div}\{\anti(0)\}=0.$$ } 
\\
The same extra constitutive coefficient appears in  Mindlin and Eshel \cite{Mindlin68} and Grioli's version \eqref{KMTe0}.

\item \textbf{the modified - symmetric couple stress model - the conformal
model}. On the other hand, in the conformal case \cite{Neff_Jeong_IJSS09,Neff_Paris_Maugin09}
one may consider that $\alpha_{2}=0$, which makes the second order couple stress
tensor $\widetilde{m}$ symmetric and trace free \cite{dahler1963theory}.
This conformal curvature case has been considered by Neff in \cite{Neff_Jeong_IJSS09},
the curvature energy having the form
\begin{align}
\widetilde{W}_{{\rm curv}}(\nabla\curl\, u) & =\frac{\alpha_{1}}{4}\,\,\|\sym\nabla\curl\, u\|^{2}.
\end{align}
Indeed, there are two major reasons uncovered in \cite{Neff_Jeong_IJSS09}
for using the modified couple stress model. First, in order to avoid
singular stiffening behaviour for smaller and smaller samples in bending
\cite{Neff_Jeong_bounded_stiffness09} one has to take $\alpha_{2}=0$.
Second, based on a homogenization procedure invoking an intuitively
appealing natural ``micro-randomness\textquotedbl{} assumption (a
strong statement of microstructural isotropy) requires conformal invariance,
which is again equivalent to $\alpha_{2}=0$. Such a model is still
well-posed \cite{Neff_JeongMMS08} leading to existence and uniqueness
results with only one additional material length scale parameter,
while it is \textbf{not} pointwise uniformly positive definite. The initial motivation of Yang et al. \cite{Yang02} for using the modified couple stress model is based on incomplete arguments \cite{MunchYang}, even if their conclusions concerning a symmetric couple stress tensor may be kept in some particular phenomenological cases \cite{MunchYang}.
\end{itemize}

\subsection{\label{sub:Diff_Geometry}Brief digression concerning differential
geometry}

When dealing with higher order theories it is suitable to introduce (see also \cite{TheseSeppecher,dell2012beyond2} for details)
two second order tensors $T$ and $Q$ which are the two projectors
on the tangent plane and on the normal to the considered surface,
respectively. As it is well known from differential geometry, such
projectors actually allow to split a given vector or tensor field
in one part projected on the plane tangent to the considered surface
and one projected on the normal to such surface.
  We can
introduce the quoted projectors as
\[
T=\tau\otimes\tau+\nu\otimes\nu=\id-{n}\otimes{n},\qquad{Q}={n}\otimes{n}.
\]
It is easy to check that the following identities are verified by
the two introduced projectors%
\begin{equation}
T+Q=\id,\qquad T\cdot T=T,\qquad Q\cdot Q=Q,\qquad T\cdot Q=0,\qquad T=T^{T},\qquad Q=Q^{T}.\label{eq:Projectors_Properties}
\end{equation}
It is then straightforward that any vector $v$ can be represented in the
local basis $\left\{ \tau,{n},\nu\right\} $ as
\[
v=v\cdot(T+Q)=(T+Q)\cdot v=\underbrace{\langle v,\tau\rangle\,\tau+\langle v,\nu\rangle\,\nu}_{T\cdot v}+\underbrace{\langle v,n\rangle\, n}_{Q\cdot v}\,,
\]
or equivalently the components of $v$ can be written as
\[
v_{i}=v_{h}\, T_{hi}+v_{h}\, Q_{hi}=T_{ih}\, v_{h}+Q_{ih}\, v_{h}.
\]
Analogously, a second order tensor $ {B}$ can be represented in such
local coordinates as $ {B}=(T+Q)^{T}\cdot  {B}\cdot(T+Q)$ and this representation
can be also straightforwardly generalized to a tensor of generic order
$N$. To the sake of generality, from now on we introduce a global
orthonormal basis $\left\{ {e}^{1},{e}^{2},{e}^{3}\right\} $
in which all fields (included $\tau,\nu$ and $n$) will
be represented if not differently specified. This also implies that all the space differential
operations which will be performed in the following calculations are
all referred to the coordinates $\left(X_{1},X_{2},X_{3}\right)$
associated to such global basis.

From \cite[p.~58, ex.~7]{gurtin2010mechanics}, we have the following
variant of the surface divergence theorem \begin{proposition} Suppose
that $w$ satisfies $\langle w,n\rangle=0$ on the surface $S\subset \mathbb{R}^3$, then
\begin{align}
\int_{S}\langle T,\nabla w\rangle \,da=\int_{\partial S}\langle n\times w,\tau\rangle \,ds,
\end{align}
where $\tau$ is a vector tangent to $S$ and to $\partial S$. \end{proposition}

Taking now $w=T\cdot v$, for arbitrary $v\in\mathbb{R}^{3}$ and
using that
\[
\langle n\times\left(T\cdot v\right),\tau\rangle=\langle\tau\times n,T\cdot v\rangle=:\langle\nu,T\cdot v\rangle=\langle T\cdot\nu,v\rangle
\]
we obtain
\begin{align}
\int_{S}\langle T,\nabla\left(T\cdot v\right)\rangle \,da=\int_{\partial S}\langle n\times(T\cdot v),\tau\rangle \,ds=\int_{\partial S}\langle T\cdot\nu,v\rangle \,ds=\int_{\partial S}\langle\nu,v\rangle\, ds,
\end{align}
for $v\in\mathbb{R}^{3}$, since $T\cdot\nu=\nu\cdot T=\nu$. We explicitly
remark that the vector $\nu$ is tangent to the surface $S$ but orthogonal
to its boundary $\partial S$ (it points inward or outwards, depending
on the choice of the orientation of $n$ and $\text{\ensuremath{\tau}}$:
usually, $n$ is chosen to be the outward unit normal to $S$ and
the orientation of $\tau$ is chosen in order to have that the vector
$\nu$, orthogonal to the border of $S$, is pointing outward the
border of the surface $S$ itself). With such notations, when considering
a vector field $v$ defined in the vicinity of the considered surface,
the surface divergence theorem can be applied to the projection of
$v$ on the tangent plane to the surface as follows (see e.g. \cite{Spivak})
\begin{equation}
\int_{S}\mathrm{Div}^{S}\left(v\right)\,da:=\int_{S}\langle T,\nabla\left(T\cdot v\right)\rangle \,da=\int_{\partial S}\langle T\cdot v,{\nu}\rangle \,da=\int_{\partial S}\langle v,T\cdot{\nu}\rangle \,da=\int_{\partial S}\langle v,{\nu}\rangle \,ds,\label{eq:Surface_Div_Th}
\end{equation}
where, clearly, in the last equality the notation
\[
\mathrm{Div}^{S}\left({v}\right):=\langle\nabla\left(T\cdot v\right),T\rangle
\]
for the surface divergence of a generic vector ${v}$ has been used.
Equivalently, in index notation the surface divergence theorem reads
\begin{equation}
\int_{S}\left(T_{ij}v_{j}\right)_{,k}T_{ik}\,da=\int_{\partial S}v_{i}T_{ik}\nu_{k}\,ds=\int_{\partial S}v_{i}\,\nu_{i}\, ds.\label{eq:Surface_Div_Th-1}
\end{equation}
This definition of surface divergence can be extended to higher order
tensors, in particular, for a second order tensor $ {B}$, its surface
divergence is introduced as the vector of components $\left[\mathrm{Div}^{S}( {B})\right]_{i}=(T\cdot  {B})_{ij,k}T_{jk}$.

\begin{remark} We explicitly remark that if $S$ coincides with the
boundary $\partial\Omega$ of the considered body and $\Gamma$ is an open subset of $\partial \Omega$, then the surface
divergence theorem \eqref{eq:Surface_Div_Th} implies (see Fig. \ref{capture-bc}
and eq. \eqref{eq:Jump})
\begin{align}
\int_{\partial\Omega}\mathrm{Div}^{S}\left(v\right)\,da&:=\int_{\partial\Omega}\langle T,\nabla\left(T\cdot v\right)\rangle \, da=\int_{\partial\Omega\setminus\Gamma}\mathrm{Div}^{S}\left(v^{-}\right)\,da+\int_{\Gamma}\mathrm{Div}^{S}\left(v^{+}\right)\,da
\label{eq:Surface_Div_Th-2}\\
&\,\, =\int_{\partial\left(\partial\Omega\setminus\Gamma\right)}\langle v^{-},{\nu^{-}}\rangle \,ds+\int_{\partial\Gamma}\langle v^{+},{\nu^{+}}\rangle \, ds=\int_{\partial\Gamma}\left\jump{\langle \:\: v\text{ },\:\nu\right.\rangle}\, ds.\nonumber
\end{align}
 Clearly, we can equivalently write in index notation
\begin{equation}
\int_{\partial\Omega}\left(T_{ij}v_{j}\right)_{,k}T_{ik}\,da=\int_{\partial\Gamma}\jump{\: v_{j}\:\nu_{j}}\,ds.\label{eq:Surface_Div_Th-1-1}
\end{equation}
\end{remark}

\subsection{Weakly and strongly independent surface fields}\label{newsubsectionmadeo}

Since it will be useful in the following, we give in this subsection some definitions which will be used throughout the paper. In particular, we introduce the notions of "strongly" and "weakly independent" vector fields defined on suitable regular surfaces.
\begin{definition}{\rm [Weakly  independent surface fields]}
	Given two vector fields $u$ and $v$ defined on a suitable regular surface $S$, we say that they are "weakly independent" if we can arbitrarily assign $u$ and $v$, independent of each other, by choosing two vectors $\overline{u}$ and $\overline{v}$ and a function $f$ such that\begin{align}
	u=\overline{u} \qquad \qquad \text{and}\qquad\qquad v=f(\overline{u},\overline{v}).
	\end{align}
	By means of the notion of weak independence, we can arbitrarily fix the vectors $u$ and $v$ on the surface $S$, but a variation of the chosen $\overline{u}$ induces a variation on one part of $v$. Nevertheless, the vector $v$ can still be arbitrarily chosen thanks to the arbitrariness  of $\overline{v}$ and $f$.
\end{definition}

\begin{definition}{\rm [Strongly independent surface fields]}
	Given two vector fields $u$ and $v$ defined on a suitable regular surface $S$, we say that they are "strongly independent" if we can arbitrarily assign $u$ and $v$, independent of each other,   by choosing two vectors $\overline{u}$ and $\overline{v}$ such that\begin{align}
	u=\overline{u} \qquad \qquad \text{and}\qquad\qquad v=\overline{v}.
	\end{align}
	The notion of strong independence allows to arbitrarily fix the two vector fields $u$ and $v$ on the surface $S$ and a variation in the choice of the first vector $u=\overline{u}$ does not induce a variation in the vector $v$. 
\end{definition}

Considering the power of external actions for a second gradient continuum
$$ \mathcal{P}^{ext}=\int_{\partial \Omega} \langle t^{ext},\delta u \rangle + \int_{\partial \Omega} \langle g^{ext}, D^1(\delta u) \rangle,$$
where $ D^1(\delta u) $ is a suitable first order space differential operator, we say, by extension of the former definition, that $ t^{ext}$ and $ g^{ext}$ are strongly and weakly independent tractions whenever they are the conjugates of strongly or weakly independent surface vector fields $\delta u$ and $ D^1(\delta u) $.

\vskip0.5cm

An example of strong and weak independence which will be pertinent in the framework of the present paper is that which can be established between the displacement assigned on the surface and its $\mathrm{curl}$ or, alternatively, its normal derivative.
Indeed, if for convenience we choose an orthonormal local basis $\{\tau,\nu,n\}$ on $S$, where $\tau$ and $\nu$ are unit tangent vectors to $S$, while $n$ is the outward unit normal vector, we can recognize that

\begin{equation}\nabla u \cdot n= \left( \frac{\partial u_\tau}{\partial x_n}, \ \frac{\partial u_\nu}{\partial x_n}, \ \frac{\partial u_n}{\partial x_n} \right)^T  \quad \text{and} \quad  
\mathrm{curl}\,u=\left( \frac{\partial u_n}{\partial x_\nu}- \frac{\partial u_\nu}{\partial x_n}, \ \frac{\partial u_\tau}{\partial x_n}- \frac{\partial u_n}{\partial x_\tau}, \ \frac{\partial u_\nu}{\partial x_\tau}- \frac{\partial u_\tau}{\partial x_\nu} \right)^T , \label{Curl_NormDer}
\end{equation}
where we clearly indicated by $u_\tau$, $u_\nu$ and $u_n$ the components of the displacement field in such local basis and $x_\tau$, $x_\nu$ and $x_n$ the space coordinates of a generic material point in the same reference frame.
 
It is known that the fact of fixing the displacement field $u$ on the surface does not fix its normal derivatives $\partial u_\tau/\partial x_n, \, \partial u_\nu/\partial x_n, \,\partial u_n/\partial x_n$, while it fixes the tangential derivatives $\partial u_n/\partial x_\nu, \, \partial u_n/\partial x_\tau, \, \partial u_\nu/\partial x_\tau, \, \partial u_\tau/\partial x_\nu$. In this optic, we can say that $u$ and $\nabla u \cdot n$ are strongly independent surface vector fields, while $u$ and $\mathrm{curl}\, u$ are weakly independent surface vector fields. Indeed, even if when fixing the displacement $u$ on the surface its tangential derivatives result to be fixed, the arbitrariness of the normal derivatives still allows to globally choose $\mathrm{curl}\, u$ in an arbitrary way. 

Therefore (regarding the $\curl u$), it is impossible to fix as constant the $\curl u$ but varying only $u$. In contrast, it is possible to fix as constant $\nabla u.\, n$ and to vary $u$ arbitrarily. In summary in this example
\begin{center}
\begin{itemize}
	\item[] \qquad\qquad\qquad $u$ \qquad and \qquad $\nabla u.\, n$ \qquad are strongly independent,
	\item[] \qquad\qquad\qquad $u$ \qquad and \qquad $\curl u$ \qquad are (only) weakly independent.
\end{itemize}
\end{center}
\section{Bulk equations and (Mindlin's weakly independent) boundary conditions in the indeterminate couple stress model}\label{directsection}\setcounter{equation}{0}
In this section we re-propose the results presented by Mindlin and Tiersten but explicitly using a variational procedure. More particularly, we try to explicitly present some reasonings that can be made to obtain their bulk equations and boundary conditions. Nevertheless, how we will show in due detail in the remainder of this paper, such reasonings are not in complete agreement with other sets of boundary conditions which can be assigned when dealing with a couple stress theory, and which are equally legitimate.
\subsection{Equilibrium and constitutive equations of the indeterminate couple stress model}
Since we consider that the solution belongs to ${\rm H}^1(\Omega)$,  we take free variations $\delta u\in C^\infty(\overline{\Omega})$ of the energy
\begin{align}W(\nabla u,\nabla  \curl u)=W_{\rm lin}( \nabla u)+\widetilde{W}_{\rm curv}(\nabla  \curl u), \label{energy_couple_stress} \end{align} where
\begin{align}\label{gradeq11}
W_{\rm lin}(\nabla u)=&\,\mu\, \|\sym \nabla u\|^2+\frac{\lambda}{2}\, [\tr(\nabla u)]^2=\mu\, \|\dev \sym \nabla u\|^2+\frac{2\, \mu+3\,\lambda}{6}\, [\tr(\nabla u)]^2,\notag\\
\widetilde{W}_{\rm curv}(\nabla  \curl u)=&\, \alpha_1\, \|\dev\sym \nabla [\axl (\skw \nabla u)]\|^2+
\alpha_2\, \|\skw\nabla [\axl (\skw \nabla u)]\|^2,
\end{align}
 and we obtain the first variation of the action functional as
 \begin{align}\label{gradeq211}
\delta\mathcal{A}=\delta \int_\Omega -W(\nabla u,\nabla\axl(\skw\nabla u))\,dv=&-\int_\Omega\bigg[ 2\, \mu\,\langle\sym \nabla u, \sym \nabla \delta u \rangle+\lambda \tr(\nabla u)\,\tr( \nabla   \delta u)\notag\\&+2\,  \alpha_1\,[ \langle \dev\sym \nabla [\axl (\skw \nabla u)],\dev\sym \nabla [\axl (\skw \nabla \delta u)]\rangle \\&+
2\, \alpha_2\, \langle \skw\nabla [\axl (\skw \nabla u)],\skw\nabla [\axl (\skw \nabla \delta u)]\rangle] \bigg]\, dv. \notag
\end{align}
Or equivalently, applying  the classical divergence theorem
 \begin{gather}
\delta\mathcal{A}=\int_\Omega \langle \mathrm{Div} \left(2\,\mu\, \sym \nabla u+\lambda \tr (\nabla u) \mathds{1} \right), \delta u \rangle \,dv-\int_{\partial \Omega} \langle \left(2\mu\, \sym \nabla u+\lambda \tr (\nabla u) \mathds{1} \right)\cdot n,\delta u\rangle\,da \nonumber \\
-2\int_\Omega \bigg\langle\alpha_{1}\,\dev\sym[\nabla\axl(\skw\nabla u)]+\alpha_{2}\,\skw[\nabla\axl(\skw\nabla u)],\nabla\axl(\skw\nabla\delta u)]\bigg\rangle\,dv. \label{gradeq211}
\end{gather}
The classical divergence theorem can be applied again to the last term appearing in the previous equation. In particular, we can notice that the following chain of equalities holds
\begin{align}
 & \int_{\Omega}\bigg\langle\alpha_{1}\,\dev\sym[\nabla\axl(\skw\nabla u)]+\alpha_{2}\,\skw[\nabla\axl(\skw\nabla u)],\nabla\axl(\skw\nabla\delta u)]\bigg \rangle\,dv\qquad\notag\\
&=\int_{\Omega}\bigg\langle- {\rm Div}\bigg(\alpha_{1}\,\dev\sym[\nabla\axl(\skw\nabla u)]+\alpha_{2}\,\skw[\nabla\axl(\skw\nabla u)]\bigg),\axl(\skw\nabla\delta u)]\bigg\rangle\,dv\notag\\
  &\qquad+\int_{\partial\Omega}\bigg \langle\bigg[\alpha_{1}\,\dev\sym[\nabla\axl(\skw\nabla u)]+\alpha_{2}\,\skw[\nabla\axl(\skw\nabla u)]\bigg]\cdot n,\axl(\skw\nabla\delta u) \bigg \rangle \,da\notag\\
&=\frac{1}{2}\int_{\Omega}\bigg  \langle-\anti\{{\rm Div}(\alpha_{1}\,\dev\sym[\nabla\axl(\skw\nabla u)]+\alpha_{2}\,\skw[\nabla\axl(\skw\nabla u)])\},\skw\nabla\delta u\bigg \rangle\,dv\notag\\
 &\qquad +\int_{\partial\Omega}\bigg\langle\bigg[\alpha_{1}\,\dev\sym[\nabla\axl(\skw\nabla u)]+\alpha_{2}\,\skw[\nabla\axl(\skw\nabla u)]\bigg]\cdot n,\axl[\skw\nabla\delta u]\bigg \rangle \,da\notag\\
&=\frac{1}{2}\int_{\Omega}\bigg  \langle-\skw\anti\{{\rm Div}(\alpha_{1}\,\dev\sym[\nabla\axl(\skw\nabla u)]+\alpha_{2}\,\skw[\nabla\axl(\skw\nabla u)])\},\nabla\delta u\bigg\rangle\,dv\notag\\
 &\qquad +\int_{\partial\Omega}\bigg \langle\bigg[\alpha_{1}\,\dev\sym[\nabla\axl(\skw\nabla u)]+\alpha_{2}\,\skw[\nabla\axl(\skw\nabla u)]\bigg]\cdot n,\axl(\skw\nabla\delta u)\bigg\rangle \,da\notag\\
&=\frac{1}{2}\int_{\Omega}  \bigg\langle\Div[\skw\anti\{{\rm Div}(\alpha_{1}\,\dev\sym[\nabla\axl(\skw\nabla u)]+\alpha_{2}\,\skw[\nabla\axl(\skw\nabla u)])\}],\delta u\bigg\rangle\,dv\notag\\
 &\qquad -\frac{1}{2}\int_{\partial \Omega}\bigg\langle[\skw\anti\{{\rm Div}(\alpha_{1}\,\dev\sym[\nabla\axl(\skw\nabla u)]+\alpha_{2}\,\skw[\nabla\axl(\skw\nabla u)])\}]\cdot n,\delta u\bigg\rangle\,da\notag\\
 &\qquad +\int_{\partial\Omega}\bigg\langle\bigg[\alpha_{1}\,\dev\sym[\nabla\axl(\skw\nabla u)]+\alpha_{2}\,\skw[\nabla\axl(\skw\nabla u)]\bigg]\cdot n,\axl(\skw\nabla\delta u)\bigg\rangle \,da,\notag
\end{align}
where $n$ is the unit outward normal vector at the surface $\partial\Omega$.
Hence, recalling that $\skw\anti(\cdot)=\anti(\cdot)$, the first variation \eqref{gradeq211} of the action functional can be finally rewritten as
\begin{align}
\delta\mathcal{A}=\int_{\Omega}\langle\Div(\sigma-\widetilde{\tau}),\delta u\rangle\, dv \nonumber -\int_{\partial\Omega}\langle(\sigma-\widetilde{\tau})\cdot n,\delta u\rangle\, dv-\int_{\partial\Omega} & \langle\widetilde{m}\cdot n,\axl(\skw\nabla\delta u)\rangle \,da,\label{germaneq3110}
\end{align}
for all virtual displacements $\delta u\in C^{\infty}(\overline{\Omega})$, where
$\sigma$ is the symmetric local force-stress tensor
\begin{align}
\sigma=2\,\mu\,\sym\nabla u+\lambda\,\tr(\nabla u)\id\in{\rm Sym}(3),
\end{align} $\widetilde{\tau}$ represents the second order nonlocal force-stress tensor
(which here is automatically skew-symmetric)
\begin{align}
\widetilde{\tau} & ={\alpha_{1}}\,\anti{\rm Div}(\dev\sym[\nabla\axl(\skw\nabla u)])+{\alpha_{2}}\,\anti{\rm Div}(\skw[\nabla\axl(\skw\nabla u)]\\
 & =\frac{\alpha_{1}}{2}\anti{\rm Div}(\dev\sym[\nabla\curl u])+\frac{\alpha_{2}}{2}\,\anti{\rm Div}(\skw[\nabla\curl u])\notag\\
 & =\anti{\rm Div}\left[\frac{\alpha_{1}}{2}\dev\sym(\nabla\curl u)+\frac{\alpha_{2}}{2}\,\skw(\nabla\curl u)\right]\notag\\
 & =\frac{1}{2}\anti{\rm Div}\left[\widetilde{m}\right]\in\so(3),\notag
\end{align}
and
\begin{align}
\widetilde{m} & =2\,{\alpha_{1}}\dev\sym(\nabla\axl(\skw\nabla u))+2\,{\alpha_{2}}\,\skw(\nabla\axl(\skw\nabla u)) \notag\\
&=[{\alpha_{1}}\dev\sym(\nabla\curl u)+{\alpha_{2}}\,\skw(\nabla\curl u)] \label{miutilde}\\
 & =\left[\frac{\alpha_{1}+\alpha_{2}}{2}\ \nabla\curl\, u+\frac{\alpha_{1}-\alpha_{2}}{2}\ \left[\nabla\curl\, u\right]^{T}\right]  \notag
\end{align}
is the second order hyperstress tensor (couple stress) which may or may not be
symmetric, depending on the material parameters.  The asymmetry of force stress is a hidden constitutive assumption, compared to the development in \cite{NeffGhibaMadeoMunch}.
Postulating a particular form of the power of external actions, the equilibrium equation can therefore be written as
 \begin{align}\label{ec11}
 \Div \,\widetilde{\sigma}_{\rm total}+f^{\rm ext}=0,
 \end{align}
where we clearly set $\widetilde{\sigma}_{\rm total}=\sigma-\widetilde{\tau}$.

\subsection{Classical (weakly independent) boundary conditions in the indeterminate couple stress model}\label{MindlinReas}
In the previous section, we have shown that  the power of internal actions can be finally written
in compact form as
\begin{equation}
\mathcal{P}^{\rm int}=\delta \mathcal{A}=\int_{\Omega}\langle\mathrm{Div}(\sigma-\widetilde{\tau}),\delta u\rangle \,dv-\int_{\partial\Omega}\langle(\sigma-\widetilde{\tau})\cdot n,\delta u\rangle\, da-\frac{1}{2}\int_{\partial\Omega}\langle\widetilde{m}\cdot n,\curl\delta u\rangle\, da\,.\label{integralasuprafata}
\end{equation}

Such form of the power of internal actions seems to suggest 6 possible independent
prescriptions of mechanical boundary conditions; three for the normal
components of the total force stress $(\sigma-\widetilde{\tau})\cdot n$
and three for the normal components of the couple stress tensor. The
possible Dirichlet boundary conditions on $\Gamma\subset\partial\Omega$
would seem to be the 6 conditions%
\footnote{as indeed proposed by Grioli \cite{Grioli60} in concordance with
the Cosserat kinematics for independent fields of displacements and
microrotation.%
}
\begin{align}
u={u}^{\rm ext},\qquad\axl(\skw\nabla u)=\frac{1}{2}\mathrm{curl}\, u=\widetilde{a}^{\rm ext}\quad(\text{or equivalently}\ \ \curl u=2\,\widetilde{a}^{\rm ext}),\label{bc11}
\end{align}
for two given functions ${u}^{\rm ext},\widetilde{a}^{\rm ext}:\mathbb{R}^{3}\rightarrow\mathbb{R}^{3}$
at the boundary $\Gamma\subset\partial\Omega$ (3+3 boundary conditions).

However, following Koiter \cite{Koiter64} we must note that the following remark
 holds true: \begin{remark} Assume that $u\in C^{\infty}(\overline{\Omega})$
and $u\Big|_{\Gamma}$ is known. Then ${\rm curl}\, u\Big|_{\Gamma}$
exists and for all open subsets $\Gamma\subset\partial\Omega$
the integral $\int_{\Gamma}\langle{\rm curl}\,u,n\rangle\, da$
is already known, while $\int_{\Gamma}\langle\curl u,\tau\rangle\, da$
is still free, where $\tau$ is any tangential vector field on $\Gamma\subset\partial\Omega$.
This fact follows using Stokes' circulation theorem
\begin{align}
\int_{\Gamma}\langle{\rm curl}\, u,n\rangle\, da=\int_{\partial\Gamma}\langle u(\gamma(t)),\gamma^{\prime}(t)\rangle\, ds,
\end{align}
where $\tau$ is a continuous unit vector field tangent to the curve
$\partial\Gamma=\{\gamma(t)\,|\, t\in[a,b]\}$ compatible
with the unit vector field $n$ normal to the surface $\Gamma$.
\end{remark}

This leads us to the next correct observation: \begin{remark}\label{remarkcc} Only the two tangential
components of ${\rm curl}\, u$ may be independently prescribed on
an open subset of the boundary. However, we may have six independent
conditions in one point on $\Gamma$, but not on an open subset of
the boundary. \end{remark}

Already Mindlin and Tiersten \cite{Mindlin62} have also correctly remarked
that  in this formulation only 5 mechanical boundary conditions
can be prescribed. Using our notations, their argument runs as follows:
\begin{align}&
\frac{1}{2}\langle\widetilde{m} \cdot\, n,\curl\delta u\rangle=\frac{1}{2}\langle(T+Q)\cdot(\widetilde{m}\cdot n),\curl\delta u\rangle=\frac{1}{2}\langle T\cdot(\widetilde{m}\cdot n),\curl\delta u\rangle+\frac{1}{2}\langle Q\cdot(\widetilde{m}\cdot n),\curl\delta u\rangle\notag\label{MEs}\\
 & =\frac{1}{2}\langle T\cdot(\widetilde{m}\cdot n),\curl\delta u\rangle+\frac{1}{2}\langle(n\otimes n)\cdot\left(\widetilde{m}\cdot n\right),\curl\delta u\rangle=\frac{1}{2}\langle T\cdot(\widetilde{m}\cdot n),\curl\delta u\rangle+\frac{1}{2}\langle\left(\langle(\mathrm{sym}\,\widetilde{m})\cdot n,n\rangle\right)n,\curl\delta u\rangle\notag\\
 & =\frac{1}{2}\langle T\cdot(\widetilde{m}\cdot n),\curl\delta u\rangle+\frac{1}{2}\langle n,\left(\langle(\mathrm{sym}\,\widetilde{m})\cdot n,n\rangle\right)\curl\delta u\rangle\\
 & =\frac{1}{2}\langle T\cdot(\widetilde{m}\cdot n),\curl\delta u\rangle+\frac{1}{2}\langle n,\curl\left[\left(\langle(\mathrm{sym}\,\widetilde{m})\cdot n,n\rangle\right)\delta u\right]\rangle-\frac{1}{2}\langle n,\nabla\left(\langle(\mathrm{sym}\,\widetilde{m})\cdot n,n\rangle\right)\times\delta u\rangle\notag\\
 & =\frac{1}{2}\langle T\cdot(\widetilde{m}\cdot n),\curl\delta u\rangle+\frac{1}{2}\langle n,\curl\left[\left(\langle(\mathrm{sym}\,\widetilde{m})\cdot n,n\rangle\right)\delta u\right]\rangle-\frac{1}{2}\langle n\times\nabla\left(\langle(\mathrm{sym}\,\widetilde{m})\cdot n,n\rangle\right),\delta u\rangle\notag\\
 & =\frac{1}{2}\langle T\cdot(\widetilde{m}\cdot n),\curl\delta u\rangle+\frac{1}{2}\langle n,\curl\left[\left(\langle(\mathrm{sym}\,\widetilde{m})\cdot n,n\rangle\right)\delta u\right]\rangle-\frac{1}{2}\langle n\times\nabla\left(\langle(\mathrm{sym}\,\widetilde{m})\cdot n,n\rangle\right),\delta u\rangle,\notag
\end{align}
where $\otimes$ denotes the dyadic product of two vectors, we have
used the property $(\eta\otimes\xi)\cdot a=\eta\,\langle\xi,a\rangle$
(for vectors $\eta$, $\xi$ and $a$), the formula $\curl(\psi\,\,\delta u)=\nabla\psi\,\times\delta u+\psi\,\,\curl\delta u$
(for any scalar field $\psi\,$) and the fact that $n\otimes n$ is
a symmetric second order tensor. The power of internal actions \eqref{integralasuprafata}
can hence be rewritten as
\begin{align*}
\mathcal{P}^{\rm int}=&\int_{\Omega}\langle\mathrm{Div}(\sigma-\widetilde{\tau}),\delta u\rangle\, dv-\int_{\partial\Omega}\langle(\sigma-\widetilde{\tau})\cdot n,\delta u\rangle\, da\\
&-\int_{\partial\Omega}  \left\{ \frac{1}{2}\langle T\cdot(\widetilde{m}\cdot n),\curl\delta u\rangle +\frac{1}{2}\langle n,\curl\left[\left(\langle(\mathrm{sym}\,\widetilde{m})\cdot n,n\rangle\right)\delta u\right]\rangle-\frac{1}{2}\langle n\times\nabla\left(\langle(\mathrm{sym}\,\widetilde{m})\cdot n,n\rangle\right),\delta u\rangle\right\}\, da.
\end{align*}
It can be remarked that the second term in the last surface integral
can be rewritten as a bulk integral by means of the divergence theorem,
so that the power of internal actions can also be rewritten
in a further equivalent form in addition to the one already established in \eqref{integralasuprafata}
\begin{align}
\mathcal{P}^{\rm int}= & \int_{\Omega}\langle\mathrm{Div}(\sigma-\widetilde{\tau}),\delta u\rangle\, dv-\int_{\Omega}\frac{1}{2}\,\mathrm{div}\left\{ \curl\left[\left(\langle(\mathrm{sym}\,\widetilde{m})\cdot n,n\rangle\right)\delta u\right]\right\}\, dv
\label{eq:Pint_Mindlin}\\
 & -\int_{\partial\Omega}\langle(\sigma-\widetilde{\tau})\cdot n-\frac{1}{2}n\times\nabla\left(\langle(\mathrm{sym}\,\widetilde{m})\cdot n,n\rangle\right),\delta u\rangle\, da-\int_{\partial\Omega}\frac{1}{2}\langle(\widetilde{m}\cdot n),T\cdot\curl\delta u\rangle\, da,\nonumber
\end{align}
where the fact that the tangential projector $T$ is symmetric has
also be used.

Mindlin and Tiersten \cite{Mindlin62} concluded that 3 boundary conditions
derive from the first surface integral  and two other from
the second surface integral, since \cite[p.~432]{Mindlin62} \textit{``the
normal component of the couple stress vector {\rm [}$\langle\widetilde{m}\cdot n,n\rangle=\langle(\mathrm{sym}\,\widetilde{m})\cdot n,n\rangle${\rm ]}
on $\partial\Omega$ enters only in the combination with the force-stress
vector shown in the coefficient of $\delta u$ in the surface integral ...''}.
Indeed, Mindlin and Tiersten are assuming to assign arbitrarily the displacement and
the tangential components of its curl on the surface $\partial\Omega$. 

As we will deeply discuss in 
 Section \ref{sectaxlb}, this choice leads to a possible set of boundary conditions in the indeterminate couple stress model. Nevertheless, this choice is not unique and assigning at the boundary different virtual fields as the virtual displacement and its normal derivative will lead us to a set of boundary conditions that are not  equivalent to those proposed by Mindlin and Tiersten.

\subsubsection{Geometric (essential, or kinematical) weakly independent boundary conditions}

Based on the expression \eqref{eq:Pint_Mindlin} of the power of internal
actions, Mindlin and Tiersten \cite{Mindlin62} concluded that the
geometric boundary conditions on $\Gamma\subset\partial\Omega$ are
the five independent conditions
\begin{align}\label{bcme1}
 \begin{array}{rcll}
\qquad u\big|_{\Gamma} & = & u^{\rm ext},&\qquad\qquad\qquad\qquad\qquad{\rm (3bc)}\vspace{1.5mm}\\
\qquad(\id-n\otimes n)\cdot\curl u\big|_{\Gamma} & = & (\id-n\otimes n)\cdot \widetilde{a}^{\rm ext}, &\qquad\qquad\qquad\qquad\qquad{\rm (2bc)}
\end{array}
\end{align}
for  given functions $u^{\rm ext}, \widetilde{a}^{\rm ext}:\mathbb{R}^{3}\rightarrow\mathbb{R}^{3}$
on the portion $\Gamma$ of the boundary.

An equivalent form of the above boundary condition is
\begin{align}\label{bcme10}
\hspace{-0.6cm}\begin{array}{rcll}
u\big|_{\Gamma} & = & u^{\rm ext},&\qquad\qquad\qquad{\rm (3bc)}\vspace{1.5mm}\\
(\id-n\otimes n)\cdot\axl(\skw \nabla\, u)\big|_{\Gamma} & = &\frac{1}{2} (\id-n\otimes n)\cdot \widetilde{a}^{\rm ext},&\qquad\qquad\qquad{\rm (2bc)}
\end{array}
\end{align}
for  given functions $u^{\rm ext}, \widetilde{a}^{\rm ext}:\mathbb{R}^{3}\rightarrow\mathbb{R}^{3}$
at the boundary.
The latter condition prescribes only the tangential
component of $\axl(\skew\nabla u)=\frac{1}{2}\,\mathrm{curl}\, u$. Therefore,
we may prescribe only 3+2 independent geometric boundary conditions. Regarding
this formulation, an existence result was proven in \cite{Neff_JeongMMS08}.

In order to give a first comparison with the boundary conditions which are coming from the full strain gradient approach, let us remark that:
\begin{lemma}\label{lemmadesprecurlaxl} {\rm {[}Equivalence of geometric
boundary conditions{]}} We consider a vector field\break  $u:\mathbb{R}^{3}\rightarrow\mathbb{R}^{3}$,
$u\in C^{\infty}(\Omega)$. The following sets of boundary conditions
are equivalent:
\begin{align} \label{equivalentBCs}
\hspace{-0.6cm}\left.\begin{array}{rcl}
u\big|_{\Gamma} \!\!\!& = &\!\!\!  u^{\rm ext},\vspace{1.5mm}\\
(\id-n\otimes n)\cdot{\rm curl}\, u\big|_{\Gamma} \!\!\!& = &\!\!\! (\id-n\otimes n)\cdot \widetilde{a}^{\rm ext}
\end{array}\!\!\!\right\} \ \Leftrightarrow\ \left\{\!\!\! \begin{array}{rcl}
u\big|_{\Gamma}\!\!\!& = &\!\!\!  u^{\rm ext},\vspace{1.5mm}\\
(\id-n\otimes n)\cdot\nabla u\cdot n\big|_{\Gamma} \!\!\!& = &\!\!\! (\id-n\otimes n)\cdot  a^{\rm ext}
\end{array}\right.
\end{align}
in the sense that one set of boundary conditions defines completely
the other set of boundary conditions, where $n$ is the unit outward
normal vector at the surface ${\Gamma}\subset\partial\Omega$ and $\widetilde{a}^{\rm ext}$ and $a^{\rm ext}$ can be a priori related. \end{lemma}
\begin{proof} 
	The proof is included in Appendix \ref{appendixlemmadesprecurlaxl}.
	\end{proof}

In summary, we can say that if purely kinematical boundary conditions are assigned in the indeterminate couple stress model, in virtue of the previous Lemma \ref{lemmadesprecurlaxl}, it is equivalent to assign on one portion of the boundary the displacement and the tangential part of its $\mathrm{curl}$ {\bf or} the displacement and the tangential part of its normal derivative. As we will see, things become much more complicated when one wants to deal with traction or mixed boundary conditions, since it is not straightforward to individuate the equivalence between different sets of boundary conditions which are nevertheless equally legitimate.

\subsubsection{Classical (weakly independent) traction boundary conditions}

Always considering the expression \eqref{eq:Pint_Mindlin} of the
power of internal actions, the possible traction boundary conditions
on $\partial\Omega\setminus\overline{\Gamma}$ given by Mindlin and
Tiersten \cite{Mindlin62} in the axial formulation are
\begin{align}
&(\sigma-\widetilde{\tau})\cdot n-\frac{1}{2}n\times\nabla\left(\langle(\mathrm{sym}\,\widetilde{m})\cdot n,n\rangle\right)  =\widetilde{t}^{\rm ext},\qquad\qquad\quad\qquad\ \ \text{traction \qquad\qquad\qquad\qquad\ (3 bc)}\notag\label{bcme2}\\
&\qquad\qquad\qquad\qquad\ \ (\id-n\otimes n)\cdot\widetilde{m}\cdot n  =(\id-n\otimes n)\cdot \widetilde{g}^{\rm ext}, \qquad 
 \text{double force traction\ \,  \qquad  (2 bc) }
\end{align}
for prescribed functions $\widetilde{t}^{\rm ext}, \widetilde{g}^{\rm ext}:\mathbb{R}^{3}\rightarrow\mathbb{R}^{3}$
on the portion $\partial \Omega\setminus \overline{\Gamma}$ of the boundary. 

Since $\delta u$ and $(\id-n\otimes n)\cdot\curl\delta u$ are  weakly independent, at this point we are tempted to conclude that the equality
\begin{align}
& \int_{\partial\Omega\setminus \overline{\Gamma}}\langle\widetilde{t},\delta u\rangle\, da
+\int_{\partial\Omega\setminus \overline{\Gamma}}\langle\widetilde{g},(\id-n\otimes n)\cdot{\rm curl}\,\delta u\rangle\, da=0
\end{align}
  for all  $\delta u\in C^2(\overline{\Omega})$ does not imply that $\widetilde{t}\Big|_{\partial\Omega\setminus \overline{\Gamma}}=0$ and $(\id-n\otimes n)\, \widetilde{g}\Big|_{\partial\Omega\setminus \overline{\Gamma}}=0$. 
However, this holds true after using the Lemmas   included in  Appendix \ref{appendixlemma}.

We want also to explicitely point out that \eqref{bcme1} and \eqref{bcme2} correctly describe the
maximal number of independent boundary conditions in the indeterminate
couple stress model but even if these conditions have been re-derived again
and again by Yang et al. \cite{Yang02}, Park and Gao \cite{Park07},
\cite{lubarda2003effects}, etc. among others they are not the only possible choice in the couple stress model.
This is explained in the following two subsections. Prescribing $\delta u$
and $(\id-n\otimes n)\cdot\curl\delta u$ on the boundary means that
we have prescribed independent geometrical boundary conditions, this
is also the argumentation of Mindlin and Tiersten \cite{Mindlin62},
Koiter \cite{Koiter64}, Sokolowski\cite{Sokolowski72}, etc. However,
the prescribed traction conditions are {\bf not stronlgy independent} but only weakly independent in the sense established in Section  \ref{newsubsectionmadeo}. For this reason we claim
that, in order to prescribe strongly independent geometric boundary conditions
and their corresponding energetic conjugate, we have to prescribe $ u$
and $(\id-n\otimes n)\cdot\nabla u\cdot n$. In other words,
as it is well assessed in the framework of full second gradient theories
(see also the expression of the external power given in \eqref{eq:Pext}),
we prescribe on the boundary the following part of the power of external
actions
\begin{align}
\int_{\partial\Omega}\langle \widetilde{t}^{\rm ext},\delta u\rangle\, da+\int_{\partial\Omega}\langle \widetilde{g}^{\rm ext},(\id-n\otimes n)\cdot\nabla\delta u\cdot n\rangle \,da,
\end{align}
in which now $\delta u$ and $(\id-n\otimes n)\cdot\nabla\delta u\cdot n$
are strongly independent and $\widetilde{g}^{\rm ext}$ does not produce (anymore hidden) work against $\delta u$.
This type of strongly independent boundary conditions are also correctly considered
already by Bleustein \cite{bleustein1967note}, but for the full strain
gradient elasticity case only. We give more details in the following section.

\section{Through  second gradient elasticity towards the  indeterminate couple stress theory: a direct approach}\label{fromgradientbc}

Independently of the method that one wants to choose to set up
the correct set of bulk equations and associated boundary conditions
for the indeterminate couple stress model, such set of equations must
be compatible with a variational principle based on the form \eqref{energy_couple_stress}
of the strain energy density. We present two different ways of performing
such variational treatment: the first one passes through a full second
gradient approach and the second one, which we call direct approach,
is based on the fact that the curvature energy is regarded as a function
of the second order tensor $\nabla\mathrm{curl}\, u$ instead than
of the third order tensor $\nabla\nabla u$ or $\nabla \sym \nabla u$.

Worthless to say, as expected, we will find that such two approaches are
 equivalent and we will explicitly establish this equivalence in eqs. \eqref{eq:Bulk_Fin}-\eqref{eq:Boundary_Fin_3}.
\subsection{\label{Second gradient}Second gradient model: general variational
setting}\setcounter{equation}{0}

In this section, we  show how the couple stress model can be regarded as a particular case of the second gradient model.
\subsubsection{First variation of the action functional: power of internal actions}
Let us consider the second gradient strain energy density 
$W(\nabla u,\nabla\nabla u)$ and the associated action functional
in the static case (no inertia considered here)
\[
\mathcal{A}=-\int_{\Omega}W(\nabla u,\nabla\nabla u)\, dv.
\]
The first variation of the action functional can be interpreted as
the power of internal actions $\mathcal{P}^{\rm int}$ of the considered
system and can be computed as follows
\[
\mathcal{P}^{\rm int}=\delta\mathcal{A}=-\int_{\Omega}\left(\frac{\partial W}{\partial u_{i,j}}\,\delta u_{i,j}+\frac{\partial W}{\partial u_{i,jk}}\,\delta u_{i,jk}\right)\, dv,
\]
where we used Levi-Civita index notation together with Einstein notation
of sum over repeated indices. Integrating a first time by parts and
using the divergence theorem we get%
\[
\delta\mathcal{A}=-\int_{\partial\Omega}\frac{\partial W}{\partial u_{i,j}}\, n_{j}\delta u_{i}\,da+\int_{\Omega}\left(\frac{\partial W}{\partial u_{i,j}}\right)_{,j}\delta u_{i}\,dv-\int_{\partial\Omega}\frac{\partial W}{\partial u_{i,jk}}\, n_{k}\delta u_{i,j}\,da+\int_{\Omega}\left(\frac{\partial W}{\partial u_{i,jk}}\right)_{,k}\delta u_{i,j}\, dv.
\]
Integrating again by parts the last bulk term we get
\[
\delta\mathcal{A}=-\int_{\partial\Omega}\left(\frac{\partial W}{\partial u_{i,j}}-\left(\frac{\partial W}{\partial u_{i,jk}}\right)_{,k}\right)\, n_{j}\delta u_{i}\,da+\int_{\Omega}\left[\frac{\partial W}{\partial u_{i,j}}-\left(\frac{\partial W}{\partial u_{i,jk}}\right)_{,k}\right]_{,j}\delta u_{i}\,dv-\int_{\partial\Omega}\frac{\partial W}{\partial u_{i,jk}}\, n_{k}\delta u_{i,j}\,da,
\]
which can also be rewritten as
\begin{equation}
\delta\mathcal{A}=-\int_{\partial\Omega}\left(\sigma_{ij}-\mathfrak{m}_{ijk,k}\right)\, n_{j}\delta u_{i}\,da+\int_{\Omega}\left(\sigma_{ij}-\mathfrak{m}_{ijk,k}\right)_{,j}\delta u_{i}\,dv-\int_{\partial\Omega}\mathfrak{m}_{ijk}\, n_{k}\delta u_{i,j}\,da\,,\label{eq:Intermediate_var}
\end{equation}
if one sets
\[
\sigma_{ij}=\frac{\partial W}{\partial u_{i,j}},\qquad\qquad\mathfrak{m}_{ijk}=\frac{\partial W}{\partial u_{i,jk}},
\]
or equivalently, in compact notation:
\begin{equation}
\sigma=\frac{\partial W}{\partial\nabla u},\qquad\qquad\mathfrak{m}=\frac{\partial W}{\partial\nabla\nabla u}.\label{eq:Stress Couple Stress}
\end{equation}

\subsubsection{Surface integration by parts and independent variations}

At this point, it must be considered that expression \eqref{eq:Intermediate_var} can still be manipulated remarking that the tangential trace of
the gradient of virtual displacement can be integrated by parts once
again and that the surface divergence theorem can be applied to this
tangential part of $\nabla\delta {u}$. Using the brief digression concerning differential geometry and recalling
the properties (\ref{eq:Projectors_Properties}) of the tangential and normal projectors, we can now ulteriorly
manipulate the last term in eq. (\ref{eq:Intermediate_var}) as follows,
\begin{eqnarray}
\int_{\partial\Omega}\mathfrak{m}_{ijk}\, n_{k}\delta u_{i,j} \,da& =:&  \int_{\partial\Omega}B_{ij}\delta u_{i,h}\delta_{hj}\,da=\int_{\partial\Omega}B_{ij}\delta u_{i,h}\left(T_{hj}+Q_{hj}\right)\,da\notag\\
 & = & \int_{\partial\Omega}B_{ij}\delta u_{i,h}T_{hj}\,da+\int_{\partial\Omega}B_{ij}\delta u_{i,h}Q_{hj}\,da,\notag\\
 & = & \int_{\partial\Omega}\left(T_{hj}B_{ij}\right)\delta u_{i,h}\,da+\int_{\partial\Omega}\left(B_{ij}n_{j}\right)\delta\left(u_{i,h}\right)n_{h}\,da\\
 & = & \int_{\partial\Omega}\left(T_{hp}T_{pj}B_{ij}\right)\delta u_{i,h}\,da+\int_{\partial\Omega}\left(B_{ij}n_{j}\right)\delta\left(u_{i,h}n_{h}\right)\,da,\notag
\end{eqnarray}
where we clearly set
\begin{align}
B_{ij}=\mathfrak{m}_{ijk}\, n_{k}.
\end{align}
We can hence recognize in the last term of this formula that the virtual
variation of the normal derivative $u_{i,h}n_{h}=:u_{i}^{n}=(\nabla u\cdot n)_i$ of the
displacement field appears. As for the first term, it can be still manipulated
suitably integrating by parts and then using the surface divergence
theorem (\ref{eq:Surface_Div_Th-1-1}), so that we can  write
\begin{eqnarray}
\int_{\partial\Omega}\mathfrak{m}_{ijk}\, n_{k}\delta u_{i,j}\, da & = & \int_{\partial\Omega}T_{hp}\left[\left(T_{pj}B_{ij}\delta u_{i}\right)_{,h}-\left(T_{pj}B_{ij}\right)_{,h}\delta u_{i}\right]\, da+\int_{\partial\Omega}\left(B_{ij}n_{j}\right)\delta u_{i}^{n}\, da,
\label{eq:IPP_Surf}\\
 & = & \int_{\partial\Gamma}\jump{B_{ip}\nu_{p}\delta u_{i}}\, ds-\int_{\partial\Omega}\left(T_{pj}B_{ij}\right)_{,h}T_{hp}\delta u_{i}\, da+\int_{\partial\Omega}\left(B_{ij}n_{j}\right)\delta u_{i}^{n}\, da.\nonumber
\end{eqnarray}
The final variation of the second gradient action functional given
in (\ref{eq:Intermediate_var}), can therefore be  written as%
\begin{align}\label{eq:Pint}
\mathcal{P}^{\rm int}=\delta\mathcal{A}= & \int_{\Omega}\left(\sigma_{ij}-\mathfrak{m}_{ijk,k}\right)_{,j}\delta u_{i}\, dv-\int_{\partial\Omega}\left[\left(\sigma_{ij}-\mathfrak{m}_{ijk,k}\right)\, n_{j}-\left(T_{pj}\mathfrak{m}_{ijk}\, n_{k}\right)_{,h}T_{hp}\right]\delta u_{i}\, da\\
 & -\int_{\partial\Omega}\left(\mathfrak{m}_{ijk}\, n_{k}n_{j}\right)\delta u_{i}^{n}\,da-\int_{\partial\Gamma}\jump{\mathfrak{m}_{ijk}\, n_{k}\nu_{j}\delta u_{i}}\, ds,\notag
\end{align}
or equivalently in compact notation%
\begin{align}\label{eq:Pint-1}
\mathcal{P}^{\rm int}=\delta\mathcal{A}=&\int_{\Omega}\left.\langle  \mathrm{Div}\left(\sigma-\mathrm{Div}\:\mathfrak{m}\right),\delta u\right.\rangle \, dv -\int_{\partial\Omega}\left.\langle  \left(\sigma-\mathrm{Div}\:\mathfrak{m}\right)\cdot n-\nabla\left[(\mathfrak{m}\cdot n)\cdot T\right]:T,\delta u\right.\rangle \, da\notag\\
&-\int_{\partial\Omega}\left.\langle  \left(\mathfrak{m}\cdot n\right)\cdot n,\left(\nabla\delta u\right)\cdot n\right.\rangle \, da -\int_{\partial\Gamma}\jump{\left.\langle  \left(\mathfrak{m}\cdot n\right)\cdot\nu,\delta u\right.\rangle }\, ds.
\end{align}
If now one recalls the principle of virtual powers according to which
a given system is in equilibrium  if the power of internal
forces is equal to the power of external forces, it is straightforward
that the expression (\ref{eq:Pint}) naturally suggests which is the
correct expression for the power of external forces that a second
gradient continuum may sustain, namely:
\begin{equation}
\mathcal{P}^{\rm ext}=\int_{\Omega}f_{j}^{\rm ext}\,\delta u_{j}\, dv+\int_{\partial\Omega}t_{j}^{\rm ext}\,\delta u_{j}\, da+\int_{\partial\Omega}g_{j}^{\rm ext}\,\delta u_{j}^{n}\, da+\int_{\partial\Gamma}\jump{\pi_{j}^{\rm ext}\,\delta u_{j}}\, ds,\label{eq:Pext}
\end{equation}
where $f{}^{\rm ext}$ are external bulk forces (expending power on  displacement),
$t{}^{\rm ext}$ are external surface forces (expending power on displacement),
$g^{\rm ext}$ are external surface double-forces (expending power on
the normal derivative of displacement) and $\pi^{\rm ext}$ are external
line forces (expending power on displacement). Imposing that
\begin{equation}
\mathcal{P}^{\rm int}+\mathcal{P}^{\rm ext}=0\label{eq:Eq of Motion}
\end{equation}
and localizing, one can get the strong form of the equations of motion
and associated boundary conditions for a second gradient continuum.

Therefore, the equilibrium equation for a second gradient continuum is
\begin{align}\label{ecgrad}
\mathrm{Div}\left(\sigma-\mathrm{Div}\:\mathfrak{m}\right)+f^{\rm ext}=0.
\end{align}
This set of partial differential equation can be complemented with the following boundary conditions:
\begin{itemize}
\item \textbf{Strongly independent, second gradient, geometric boundary conditions}
\begin{align}\label{bgg}
\hspace{-0.6cm}\begin{array}{rcll}
u\big|_{\Gamma} & = & u^{\rm ext},&\qquad\qquad\qquad\qquad\qquad\qquad\qquad\qquad\qquad{\rm (3 bc)}\vspace{1.5mm}\\
\nabla u\cdot n\big|_{\Gamma} & = & a^{\rm ext},&\qquad\qquad\qquad\qquad\qquad\qquad\qquad\qquad\qquad{\rm (3 bc)}
\end{array}
\end{align}
for  given functions $u^{\rm ext}, a^{\rm ext}:\mathbb{R}^{3}\rightarrow\mathbb{R}^{3}$
on the portion $\Gamma$ of the boundary.
\item \textbf{Strongly independent, second gradient, traction boundary conditions}\\
Traction boundary condition on $\partial \Omega\setminus\overline{\Gamma}$:
\begin{align}
\left(\sigma-\mathrm{Div}\:\mathfrak{m}\right)\cdot n-\nabla\left[(\mathfrak{m}\cdot n)\cdot T\right]:T \big|_{\partial \Omega\setminus\overline{\Gamma}}&=t^{\rm ext},\qquad\ \quad\text{traction \ \qquad\quad\qquad\qquad(3 bc)}\notag\label{bgt}\\
\left(\mathfrak{m}\cdot n\right)\cdot n\big|_{\partial \Omega\setminus\overline{\Gamma}} & = g^{\rm ext},\quad\quad\ \ \text{``double force  traction''\qquad  (3 bc)}
\end{align}
for prescribed functions $t^{\rm ext}, g^{\rm ext}:\mathbb{R}^{3}\rightarrow\mathbb{R}^{3}$
at the boundary.

Traction boundary condition on $\partial {\Gamma}$:
\begin{align}\label{bgt2}
\qquad\qquad\qquad\quad\qquad\quad\quad \jump{\left(\mathfrak{m}\cdot n\right)\cdot\nu}\big|_{\partial {\Gamma}}  = {\pi}^{\rm ext}\qquad\qquad \text{``line force traction''\quad\quad\quad (3 bc)}
\end{align}
for a prescribed function $ {\pi}^{\rm ext}:\mathbb{R}^{3}\rightarrow\mathbb{R}^{3}$ on
$\partial {\Gamma}$.
\end{itemize}

We want to stress the fact that in the framework of a second gradient
theory the test functions that can be arbitrarily assigned on the
boundary $\partial\Omega$ are the virtual displacement $\delta u$
and the normal derivative of the virtual displacement $\nabla(\delta u)\cdot n$.
This means that one has $3+3=6$ independent geometric boundary conditions
that can be assigned on the boundary of the considered second gradient
medium. Analogously one can think to assign $3+3=6$ traction conditions
on the force (in duality of $\delta u$) and double force (in duality
of $\nabla(\delta u)\cdot n$) respectively. Hence, in a complete
second gradient theory $6$ independent scalar conditions must be
assigned on the boundary in order to have a well-posed problem.

\subsection{\label{sub:The-full-approach}The  indeterminate couple stress
model viewed as a subclass of the second gradient elasticity model}

As in the previous case, we consider a particular
strain energy density of the type
\[
W=W_{{\rm lin}}(\nabla u)+W_{{\rm curv}}(\nabla[\mathrm{axl}(\mathrm{skew}\nabla u)])=W_{{\rm lin}}(\nabla u)+\widetilde{W}_{{\rm curv}}(\nabla\mathrm{curl}\, u),
\]
where $W_{{\rm lin}}(\nabla u)$ is given in eq.~\eqref{gradeq11},
while the curvature energy $\widetilde{W}_{{\rm curv}}(\nabla\mathrm{curl}\, u)$ also discussed in the previous
section, is given by
\begin{align}
\widetilde{W}_{{\rm curv}}(\nabla\mathrm{curl}\, u) & =\frac{\alpha_{1}}{4}\left\Vert \,\mathrm{sym}\,\nabla\mathrm{curl}\, u\,\right\Vert ^{2}+\frac{\alpha_{2}}{4}\left\Vert \,\,\mathrm{skew}\,\nabla\mathrm{curl}\, u\,\right\Vert ^{2}
\label{eq:CoupleStressEnergy}\\
 & =:\frac{\alpha_{1}}{4}\left\Vert \, S\,\right\Vert ^{2}+\frac{\alpha_{2}}{4}\left\Vert \, A\,\right\Vert ^{2}=\frac{\alpha_{1}}{4}S_{lm}S_{lm}+\frac{\alpha_{2}}{4}A_{lm}A_{lm},\nonumber
\end{align}
where we set
\begin{equation}
S_{pq}:=\left(\mathrm{sym}\,\nabla\mathrm{curl}\, u\right)_{pq}=\frac{\epsilon_{prs}u_{s,rq}+\epsilon_{qrs}u_{s,rp}}{2},\qquad A_{pq}:=\left(\mathrm{skew}\,\nabla\mathrm{curl}\, u\right)_{pq}=\frac{\epsilon_{prs}u_{s,rq}-\epsilon_{qrs}u_{s,rp}}{2}.\label{eq:SymSkew}
\end{equation}
This decomposition of the curvature energy is equivalent to that which
can be found in Mindlin and Tiersten \cite{Mindlin62} and presented by us in eq. \eqref{energyindet}.

Regarding the curvature energy \eqref{eq:CoupleStressEnergy} as a
particular case of second gradient energy, we can directly calculate
the particular form of the third order hyperstress tensor  as
\begin{align}
\widetilde{\mathfrak{m}}_{ijk} & =\frac{\partial\widetilde{W}_{\mathrm{curv}}}{\partial u_{i,jk}}=\frac{\partial\widetilde{W}_{\mathrm{curv}}}{\partial S_{pq}}\frac{\partial S_{pq}}{\partial u_{i,jk}}+\frac{\partial\widetilde{W}_{\mathrm{curv}}}{\partial A_{pq}}\frac{\partial A_{pq}}{\partial u_{i,jk}}=\frac{\alpha_{1}}{2}S_{pq}\frac{\partial S_{pq}}{\partial u_{i,jk}}+\frac{\alpha_{2}}{2}A_{pq}\frac{\partial A_{pq}}{\partial u_{i,jk}}.\label{eq:CoupStress}
\end{align}
It can be checked that, from  \eqref{eq:SymSkew},
one gets
$$
\frac{\partial S_{pq}}{\partial u_{i,jk}}=\frac{1}{2}\left(\epsilon_{pji}\delta_{qk}+\epsilon_{qji}\delta_{pk}\right),\qquad\qquad\frac{\partial A_{pq}}{\partial u_{i,jk}}=\frac{1}{2}\left(\epsilon_{pji}\delta_{qk}-\epsilon_{qji}\delta_{pk}\right).
$$
Replacing these expressions in \eqref{eq:CoupStress}, using the definitions
\eqref{eq:SymSkew} together with the identities
$$\epsilon_{pji}\epsilon_{prs}=\delta_{jr}\delta_{is}-\delta_{js}\delta_{ir}$$
and $$\epsilon_{pji}\epsilon_{krs}=\delta_{pk}\delta_{jr}\delta_{is}-\delta_{pk}\delta_{js}\delta_{ir}-\delta_{pr}\delta_{jk}\delta_{is}
+\delta_{pr}\delta_{js}\delta_{ik}+\delta_{ps}\delta_{jk}\delta_{ir}-\delta_{ps}\delta_{jr}\delta_{ik},$$
the fact that $S_{pk}=S_{kp}$ and $A_{pk}=-A_{kp}$ and simplifying
gives
\begin{align}
\widetilde{\mathfrak{m}}_{ijk} & =\frac{\alpha_{1}}{4}\,\left(\epsilon_{pji}S_{pk}+\epsilon_{qji}S_{kq}\right)+\frac{\alpha_{2}}{4}\,\left(\epsilon_{pji}A_{pk}-\epsilon_{qji}A_{kq}\right)
=\frac{\alpha_{1}}{2}\,\epsilon_{pji}S_{pk}+\frac{\alpha_{2}}{2}\,\epsilon_{pji}A_{pk}\nonumber \nonumber \\
 & =\frac{\alpha_{1}}{4}\,\epsilon_{pji}(\epsilon_{prs}u_{s,rk}+\epsilon_{krs}u_{s,rp})+
 \frac{\alpha_{2}}{4}\,\epsilon_{pji}(\epsilon_{prs}u_{s,rk}-\epsilon_{krs}u_{s,rp})\label{eq:DoubleStress}\\
 & =\frac{\alpha_{1}}{2}\,\left(u_{i,jk}-u_{j,ik}\right)+\frac{1}{4}\,\left(\alpha_{1}-\alpha_{2}\right)\left[u_{p,ip}\delta_{jk}-u_{i,pp}\delta_{jk}+u_{j,pp}\delta_{ik}-u_{p,jp}\delta_{ik}\right].\nonumber
\end{align}
Such particular expression of the third order hyperstress  tensor can be also
written in compact form as
\begin{align}
\widetilde{\mathfrak{m}}  =&\,\alpha_{1}\nabla\left[\mathrm{skew}\left(\nabla u\right)\right]+\frac{1}{4}\left(\alpha_{1}-\alpha_{2}\right)\left[\nabla\left(\mathrm{div}\, u\right)\otimes\id-\mathrm{Div}\left(\text{\ensuremath{\nabla}}u\right)\otimes\id\right]
\label{eq:DoubleStressCompact}\\
 & +\frac{1}{4}\left(\alpha_{1}-\alpha_{2}\right)\left[\left(\mathrm{Div}\left(\text{\ensuremath{\nabla}}u\right)\otimes\id\right)^{T^{12}}
 -\left(\nabla\left(\mathrm{div}\, u\right)\otimes\id\right)^{T^{12}}\right],\nonumber
\end{align}
where we denote by the superscript $T^{12}$ the transposition over
the two first indices of the considered third order tensors. With
such definition of  the third order hyperstress  tensor $\widetilde{\mathfrak{m}}$
one can now write the principle of virtual powers for the considered
particular case in the form
\begin{align}
\int_{\Omega}\left(\sigma_{ij}-\widetilde{\mathfrak{m}}_{ijk,k}\right)_{,j}\delta u_{i}\, dv&-\int_{\partial\Omega}\left[\left(\sigma_{ij}-\widetilde{\mathfrak{m}}_{ijk,k}\right)\, n_{j}-\left(T_{pj}\widetilde{\mathfrak{m}}_{ijk}\, n_{k}\right)_{,h}T_{hp}\right]\delta u_{i}\, da\notag\\
&-\int_{\partial\Omega}\left(\widetilde{\mathfrak{m}}_{ijk}\, n_{k}n_{j}\right)\delta u_{i}^{n}\, da-\int_{\partial\Gamma}\jump{\widetilde{\mathfrak{m}}_{ijk}\, n_{k}\nu_{j}\delta u_{i}}\, ds\label{eq:PrincVirtPow}\\
&=-\int_{\Omega}f_{i}^{\rm ext}\delta u_{i}\, dv-\int_{\partial\Omega}t_{i}^{\rm ext}\delta u_{i}\, da-\int_{\partial\Omega}m_{i}^{\rm ext}\delta u_{i}^{n}\, da-\int_{\partial\Gamma}\pi_{j}^{\rm ext}\delta u_{j}\,ds\,.\nonumber
\end{align}
\begin{itemize}
\item We have to remark that the term
\begin{align}
\left(\widetilde{\mathfrak{m}}_{ijk}\, n_{k}n_{j}\right)\delta u_{i}^{n}=\Big[&\frac{\alpha_{1}}{2}(u_{i,jk}-u_{j,ik})n_{k}n_{j}\notag\\
&+\frac{1}{4}\left(\alpha_{1}-\alpha_{2}\right)\left(u_{p,ip}n_{j}n_{j}-u_{i,pp}n_{j}n_{j}+u_{j,pp}n_{j}n_{i}-u_{p,jp}n_{j}n_{i}\right)\Big]\delta u_{i}^{n}
\end{align}
is vanishing for some particular choices of the indices. In particular,
if, for the sake of simplicity, one considers the introduced quantities
to be all expressed in the local orthonormal basis $\left\{ n,\tau,\nu\right\} $,
then the aforementioned term can be rewritten as
\[
\left(\widetilde{\mathfrak{m}}_{ijk}\, n_{k}n_{j}\right)\delta u_{i}^{n}=\left[\frac{\alpha_{1}}{2}(u_{i,11}-u_{1,i1})+\frac{1}{4}\left(\alpha_{1}-\alpha_{2}\right)\left(u_{p,ip}-u_{i,pp}+u_{1,pp}n_{i}-u_{p,1p}n_{i}\right)\right]\delta u_{i}^{n}.
\]
It can be easily checked that such term is vanishing when $i=1$.
More precisely, we are saying that the normal component of the normal
derivative $\delta u_{i}^{n}$ does not contribute to the power of
internal forces when considering the indeterminate couple stress model.
This is equivalent to say that indeed only 2 geometric boundary conditions
can be imposed on the normal derivative of virtual displacement or,
equivalently, on its ``traction'' counterpart which is the double
force.
\end{itemize}
Hence, the governing equations of the considered system can also be
formally written in the form \begin{gather}
{\rm Div}(\sigma-{\rm Div}\,\widetilde{\mathfrak{m}})+f^{\rm ext}=0,\label{eq:Bulk_CS}\qquad
\text{{in\,\ duality\,\ of}}\quad \delta u
\end{gather}
together with the following  boundary conditions  induced
by \eqref{eq:Pint-1}\footnote{We recall that in the considered couple stress model expressed in the framework of a full second gradient theory the constitutive form for $m$ is given in eq. \eqref{eq:DoubleStress} or equivalently \eqref{eq:DoubleStressCompact}.}:
\begin{itemize}
\item \textbf{Strongly independent, geometric  boundary conditions for the couple stress model on $\Gamma$  (as derived by a full-gradient model) }
\begin{align}\label{bgg}
\hspace{-0.6cm}\begin{array}{rcll}
u\big|_{\Gamma} & = & u^{\rm ext},&\vspace{1.5mm}\\
T\cdot \nabla u\cdot n\big|_{\Gamma} & = & T\cdot a^{\rm ext},&
\end{array}
\end{align}
for given functions $u^{\rm ext},a^{\rm ext}:\mathbb{R}^{3}\rightarrow\mathbb{R}^{3}$
at the boundary.
\item \textbf{Strongly independent, traction boundary conditions on $\partial \Omega\setminus\overline{\Gamma}$  (as derived by a full-gradient model) }\\
\begin{align}
\left(\sigma-{\rm Div}\,\widetilde{\mathfrak{m}}\right)\cdot n-\nabla\left[\left(\widetilde{\mathfrak{m}}\cdot n\right)\cdot T\right]:T&=t^{\rm ext},\, \label{eq:Boundary_CS}\qquad\qquad
\text{{in\,\ duality\,\ of}}\quad\delta u
\\\
T\cdot\left[\left(\widetilde{\mathfrak{m}}\cdot n\right)\cdot n\right]&=T\cdot g^{\rm ext},\qquad\,\,\ 
\text{{in\,\ duality\,\ of}}\quad T\cdot \left(\nabla\delta u\right)\cdot n\notag
\end{align}
for prescribed functions $t^{\rm ext}, g^{\rm ext}:\mathbb{R}^{3}\rightarrow\mathbb{R}^{3}$
at the boundary.

Traction boundary condition on the curve $\partial {\Gamma}$:
\begin{gather}
\jump{  \left(\widetilde{\mathfrak{m}}\cdot n\right)\cdot\nu}=  \pi^{\rm ext},\label{eq:Boundary_CS-2}
\end{gather}
for a prescribed function $ {\pi}^{\rm ext}:\mathbb{R}^{3}\rightarrow\mathbb{R}^{3}$ on
$\partial {\Gamma}$.
\end{itemize}

\subsection{\label{sub:identification-1}Reduction
from the third  order hyperstress tensor $\widetilde{\mathfrak{m}}$ to  Mindlin's second order couple stress tensor $\widetilde{m}$ }

We want to prove here that the equations \eqref{ecgrad} and the traction
boundary conditions \eqref{bgg}-\eqref{eq:Boundary_CS}
can be equivalently rewritten using  Mindlin's second order couple stress tensor
\begin{align}
\widetilde{m} & =\frac{\alpha_{1}+\alpha_{2}}{2}\ \nabla\curl u+\frac{\alpha_{1}-\alpha_{2}}{2}\ (\nabla\curl u)^{T}\quad\qquad\qquad\qquad\text{\qquad\qquad{\bf Mindlin}}\label{eq:mMindlin}\\
 & ={\alpha_{1}}\dev\sym(\nabla\curl u)+{\alpha_{2}}\,\skw(\nabla\curl u)\ \ \qquad\qquad\qquad\qquad\qquad\text{{\bf equivalent form 1}}\notag\\
 & =2\,{\alpha_{1}}\dev\sym(\nabla\axl(\skw\nabla u))+2\,{\alpha_{2}}\,\skw(\nabla\axl(\skw\nabla u)),\notag\qquad\text{{\bf equivalent form 2} }\\
\widetilde{m}_{i {l}} & =\frac{\alpha_{1}+\alpha_{2}}{2}\ \epsilon_{ijk}u_{k,j {l}}+\frac{\alpha_{1}-\alpha_{2}}{2}\ \epsilon_{mjk}u_{k,ji}, \qquad\qquad\qquad\qquad\qquad\quad\ \,\!\!\text{{\bf index format}}\notag
\end{align}
instead of the third order tensor given in eq. \eqref{eq:DoubleStressCompact}.
Such second order hyper-stress tensor  has been introduced by Mindlin and Tiersten
\cite{Mindlin62} and we have shown in a previous section that it can
be obtained by means of a direct variational approach that does not
need the introduction of the third order couple stress tensor $\widetilde{\mathfrak{m}}$
(see eq. \eqref{miutilde}).

In order to be able to set up such equivalence, we have to remark that, for the choices \eqref{eq:DoubleStressCompact}
and \eqref{eq:mMindlin} of $\widetilde{\mathfrak{m}}$ and $\widetilde{m}$,
the following properties are verified (see Appendix \ref{sec:Some-useful-relationships-2}
for detailed calculations)\footnote{Using a classical notation $\Delta=\mathrm{Div} \nabla$ is the Laplacian operator.}
\begin{align}
{\rm Div}\underbrace{\widetilde{\mathfrak{m}}}_{\mathbb{R}^{3\times3\times3}}=\frac{1}{2}\anti\Div\underbrace{\widetilde{m}}_{\mathbb{R}^{3\times3}}=
\frac{\alpha_{1}+\alpha_{2}}{2}\,\Delta(\skew\nabla u),\label{eq:IdentBulk-1}
\end{align}
and
\begin{align}
\nabla\left[(\widetilde{\mathfrak{m}}\cdot n)\cdot T\right]:T= & \,\frac{1}{2}\nabla[\mathrm{anti}\left(\widetilde{m}\cdot n\right)\cdot T\,]:T,\label{eq:Ident1-1-1}\\
T\cdot[\left(\widetilde{\mathfrak{m}}\cdot n\right)\cdot n]= & \,\frac{1}{2}T\cdot\mathrm{anti}(\widetilde{m}\cdot n)\cdot n,\label{eq:Ident2-1-1}\\
\jump{(\widetilde{\mathfrak{m}}\cdot n)\cdot\nu}= & \,\frac{1}{2}\jump{[\mathrm{anti}(\widetilde{m}\cdot n)\cdot n]\cdot\nu},\label{eq:Ident3-1-1}
\end{align}
where
\begin{equation}
\widetilde{\mathfrak{m}}\cdot n=\frac{1}{2}\mathrm{anti}\left(\widetilde{m}\cdot n\right)=\alpha_{1}\left[\nabla\left(\mathrm{skew}\,\nabla u\right)\right]\cdot n+\frac{\left(\alpha_{1}-\alpha_{2}\right)}{2}\:\mathrm{skew}\left[\nabla\left(\mathrm{Div}\, u\right)\otimes n-\mathrm{Div}\left(\text{\ensuremath{\nabla}}u\right)\otimes n\right].\label{eq:Mn}
\end{equation}
Clearly, based upon such relationships, we can recognize the following equivalent forms for the bulk equations
\begin{gather}
{\rm Div}(\sigma-{\rm Div}\,\widetilde{\mathfrak{m}})+f^{\rm ext}=0\quad\Leftrightarrow\quad{\rm Div}\left(\sigma-\frac{1}{2}\anti\Div\widetilde{m}\right)+f^{\rm ext}=0,\label{eq:Bulk_Fin}\\
\text{{in\,\ duality\,\ of }}\,\delta u\nonumber
\end{gather}
together with the following equivalent forms of the traction boundary conditions\medskip{}
\begin{gather}
\left(\sigma-{\rm Div}\,\widetilde{\mathfrak{m}}\right)\cdot n-\nabla\left[\left(\widetilde{\mathfrak{m}}\cdot n\right)\cdot T\right]:T=t^{\rm ext}\ \Leftrightarrow\ \left(\sigma-\frac{1}{2}\anti\Div\widetilde{m}\right)\cdot n-\frac{1}{2}\nabla[\mathrm{anti}\left(\widetilde{m}\cdot n\right)\cdot T\,]:T=t^{\rm ext}\label{eq:Boundary_Fin_1}\\
\text{{in\,\ duality\,\ of }}\,\delta u,\nonumber
\end{gather}
\begin{gather}
T\cdot\left[\left(\widetilde{\mathfrak{m}}\cdot n\right)\cdot n\right]=T\cdot g^{\rm ext}\quad\Leftrightarrow\quad\frac{1}{2}T\cdot\mathrm{anti}(\widetilde{m}\cdot n)\cdot n=T\cdot g^{\rm ext}\label{eq:Boundary_Fin_2}\\
\text{{in\,\ duality\,\ of }}\, T\cdot \left(\nabla\delta u\right)\cdot n,\nonumber
\end{gather}
and finally the following equivalent conditions on the boundary of the boundary
$\partial\Gamma$ where traction is assigned\medskip{}
\begin{gather}
\jump{\left(\widetilde{\mathfrak{m}}\cdot n\right)\cdot\nu}=\pi^{\rm ext}\ \Leftrightarrow\ \frac{1}{2}\jump{\mathrm{anti}(\widetilde{m}\cdot n)\cdot\nu}=\pi^{\rm ext}\label{eq:Boundary_Fin_3}\\
\text{{in\,\ duality\,\ of }}\,\delta u\,.\nonumber
\end{gather}

\subsection{A direct way to obtain strongly independent boundary conditions in the \break indeterminate couple model}\label{sectionincomplete3}

Let us consider again the energy
\[
W=W_{{\rm lin}}(\nabla u)+\widetilde{W}_{{\rm curv}}(\nabla\mathrm{curl}\, u),
\]
with $W_{{\rm lin}}$ and $\widetilde{W}_{{\rm curv}}$ defined in equations
\eqref{gradeq11} and for which different equivalent forms of the curvature energy have been given in eq. \eqref{energyindet}.
To the sake of completeness, we derive in this section the equations
of motion and associated boundary conditions of the  couple stress
model by directly computing the first variation of the action functional
associated to the considered energy, without noticing that such energy
is indeed a very particular case of a second gradient energy. This
procedure follows what was done by Mindlin and Tiersten \cite{Mindlin62} and it is presented in Section \ref{directsection}.
Here, as done in \cite{Mindlin62}, the curvature energy
is regarded as a function of the second order tensor $\nabla\mathrm{curl}\,u$,
instead that of the third order tensor $\nabla\nabla u$. The difference of the calculation that we present here, with respect to what is done by Mindlin and Tiersten, is that we proceed further in the process of integration by parts up to the point of getting strongly independent quantities on the boundary. 
As it is
shown in Section \ref{sub:identification-1}, the two approaches can
be considered to be finally equivalent, provided that a suitable identification
of the second and third order tensors appearing in the governing equations
is performed.

 As usual, the power of internal actions is given by
the first variation of the action functional which can be directly
computed as
\begin{align}\label{internalnew}
\mathcal{P}^{\rm int} & =\delta\mathcal{A}=-\delta\int_{\Omega}\left[W_{{\rm lin}}(\nabla u)+\widetilde{W}_{\rm curv}\left(\nabla\mathrm{curl}\: u\right)\right]\,dv\notag
\\
 & =-\int_{\Omega}\left.\langle  \frac{\partial W_{{\rm lin}}}{\partial\nabla u},\:\delta\nabla u\right.\rangle\,dv -\int_{\Omega}\left.\langle  \frac{\partial\widetilde{W}_{\rm curv}}{\partial S},\:\delta S\right.\rangle \,dv-\int_{\Omega}\left.\langle  \frac{\partial\widetilde{W}_{\rm curv}}{\partial A},\:\delta A\right.\rangle\,dv
\\
 & =-\int_{\Omega}\frac{\partial W_{{\rm lin}}}{\partial u_{i,j}}\delta u_{i,j}\,dv-\int_{\Omega}\frac{\partial\widetilde{W}_{\rm curv}}{\partial S_{ij}}\delta S_{ij}\,dv-\int_{\Omega}\frac{\partial\widetilde{W}_{\rm curv}}{\partial A_{ij}}\delta A_{ij}\,dv.\notag
\end{align}

Using the expression  of $\widetilde{W}_{\rm curv}$ given in  \eqref{eq:CoupleStressEnergy} and then the definitions \eqref{eq:SymSkew} for $S$ and $A$ together with
the properties of Levi-Civita symbols, it can be checked that
\begin{align*}
&-\int_{\Omega}\frac{\partial\widetilde{W}_{\rm curv}}{\partial S_{ij}}\delta S_{ij}\,dv-\int_{\Omega}\frac{\partial\widetilde{W}_{\rm curv}}{\partial A_{ij}}\delta A_{ij}\,dv\\
&=-\frac{\text{\ensuremath{\alpha}}_{1}}{4}\int_{\Omega}S_{ij}\delta S_{ij}\, dv-\frac{\text{\ensuremath{\alpha}}_{2}}{4}\int_{\Omega}A_{ij}\delta A_{ij}\,dv=\\
&=-\frac{\text{\ensuremath{\alpha}}_{1}}{8}\int_{\Omega}\left(\epsilon_{ipq}u_{q,pj}+\epsilon_{jpq}u_{q,pi}\right)\left(\epsilon_{irs}\delta u_{s,rj}+\epsilon_{jrs}\delta u_{s,ri}\right)\,dv\\
&\qquad-\frac{\text{\ensuremath{\alpha}}_{2}}{8}\int_{\Omega}\left(\epsilon_{ipq}u_{q,pj}-\epsilon_{jpq}u_{q,pi}\right)\left(\epsilon_{irs}\delta u_{s,rj}-\epsilon_{jrs}\delta u_{s,ri}\right)\,dv\\
&=-\frac{\left(\text{\ensuremath{\alpha}}_{1}+\alpha_{2}\right)}{8}\int_{\Omega}\left(\epsilon_{ipq}\epsilon_{irs}u_{q,pj}\delta u_{s,rj}+\epsilon_{jpq}\epsilon_{jrs}u_{q,pi}\delta u_{s,ri}\right)\,dv\\
&\qquad-\frac{\left(\text{\ensuremath{\alpha}}_{1}-\alpha_{2}\right)}{8}\int_{\Omega}\left(\epsilon_{jpq}\epsilon_{irs}u_{q,pi}\delta u_{s,rj}+\epsilon_{ipq}\epsilon_{jrs}u_{q,pj}\delta u_{s,ri}\right)\,dv\\
&=-\frac{\left(\text{\ensuremath{\alpha}}_{1}+\alpha_{2}\right)}{4}\int_{\Omega}\left[\left(\delta_{pr}\delta_{qs}-\delta_{ps}\delta_{qr}\right)u_{q,pi}\delta u_{s,ri}\right]\,dv\\
&\qquad-\frac{\left(\text{\ensuremath{\alpha}}_{1}-\alpha_{2}\right)}{4}\int_{\Omega}\left(\left(\delta_{ij}\delta_{pr}\delta_{qs}-\delta_{ij}\delta_{ps}\delta_{qr}-\delta_{jr}\delta_{pi}\delta_{qs}+\delta_{jr}\delta_{ps}\delta_{qi}+\delta_{js}\delta_{pi}\delta_{qr}-\delta_{js}\delta_{pr}\delta_{qi}\right)u_{q,pi}\delta u_{s,rj}\right)\,dv\\
&=-\frac{\left(\text{\ensuremath{\alpha}}_{1}+\alpha_{2}\right)}{4}\int_{\Omega}\left(u_{s,r}-u_{r,s}\right)_{,i}\delta u_{s,ri}\,dv\\
&\qquad-\frac{\left(\text{\ensuremath{\alpha}}_{1}-\alpha_{2}\right)}{4}\int_{\Omega}\left(\left(u_{s,r}-u_{r,s}\right)_{,i}\delta u_{s,ri}+\left(u_{i,si}-u_{s,ii}\right)\delta u_{s,rr}+\left(u_{r,ii}-u_{i,ri}\right)\delta u_{s,rs}\delta u_{s,rs}\right)\,dv\\
&=-\frac{\alpha_{1}}{2}\int_{\Omega}\left(u_{s,r}-u_{r,s}\right)_{,i}\delta u_{s,ri}\,dv-\frac{\left(\text{\ensuremath{\alpha}}_{1}-\alpha_{2}\right)}{4}\int_{\Omega}\left(\left(u_{i,s}-u_{s,i}\right)_{,i}\delta u_{s,rr}+\left(u_{r,i}-u_{i,r}\right)_{,i}\delta u_{s,rs}\right)\,dv\,.
\end{align*}

Recalling also the results for the variation of the classical first gradient term given in \eqref{gradeq211} this last relation implies that the internal actions \eqref{internalnew} can be rewritten as
\begin{align*}
\mathcal{P}^{\rm int}=\delta\mathcal{A}&=-\int_{\Omega}\left(\mu\left(u_{i,j}+u_{j,i}\right)+\lambda\, u_{k,k}\delta_{ij}\right)\delta u_{i,j}\, dv-\frac{\text{\ensuremath{\alpha}}_{1}}{4}\int_{\Omega}S_{ij}\delta S_{ij}-\frac{\text{\ensuremath{\alpha}}_{2}}{4}\int_{\Omega}A_{ij}\delta A_{ij}\, dv
\\
&=-\int_{\Omega}\left(\mu\left(u_{i,j}+u_{j,i}\right)+\lambda\, u_{k,k}\delta_{ij}\right)\delta u_{i,j}\, dv-\frac{\alpha_{1}}{2}\int_{\Omega}\left(u_{s,r}-u_{r,s}\right)_{,i}\delta u_{s,ri}\, dv
\\
&\quad\ -\frac{\left(\text{\ensuremath{\alpha}}_{1}-\alpha_{2}\right)}{4}\int_{\Omega}\left[\left(u_{i,s}-u_{s,i}\right)_{,i}\delta u_{s,rr}+\left(u_{r,i}-u_{i,r}\right)_{,i}\delta u_{s,rs}\right]\, dv.
\end{align*}
Suitably integrating by parts we can hence write
\begin{align}
\mathcal{P}^{\rm int} & =\delta\mathcal{A}=-\int_{\partial\Omega}\sigma_{ij}n_{j}\delta u_{i}\, da+\int_{\Omega}\sigma_{ij,j}\delta u_{i}\, dv
\nonumber \\
 &\qquad\qquad +\int_{\Omega}\left[\frac{\alpha_{1}}{2}\left(u_{s,r}-u_{r,s}\right)_{,ii}
 +\frac{\left(\text{\ensuremath{\alpha}}_{1}-\alpha_{2}\right)}{4}\left(\left(u_{i,s}-u_{s,i}\right)_{,ir}+\left(u_{r,i}-u_{i,r}\right)_{,is}\right)\right]\delta u_{s,r}\, dv
\label{eq:PintMindlin}\\
 &\qquad\qquad -\int_{\partial\Omega}\left[\frac{\alpha_{1}}{2}\left(u_{s,r}-u_{r,s}\right)_{,i}n_{i}+
 \frac{\left(\text{\ensuremath{\alpha}}_{1}-\alpha_{2}\right)}{4}\left(\left(u_{i,s}-u_{s,i}\right)_{,i}n_{r}+
 \left(u_{r,i}-u_{i,r}\right)_{,i}n_{s}\right)\right]\delta u_{s,r}\, da
\nonumber \\
 & \quad\qquad=-\int_{\partial\Omega}\left(\sigma_{ij}-\widetilde{\tau}_{ij}\right)n_{j}\delta u_{i}\, da+\int_{\Omega}\left(\sigma_{ij}-\widetilde{\tau}_{ij}\right)_{,j}\delta u_{i}\, dv-\int_{\partial\Omega}B_{ij}\delta u_{i,j}\, da\,,\nonumber
\end{align}
where we set
\begin{align*}
\sigma_{ij} & =\left(\mu\left(u_{i,j}+u_{j,i}\right)+\lambda\, u_{k,k}\delta_{ij}\right),\\
\widetilde{\tau}_{ij} & =\left[\frac{\alpha_{1}}{2}\left(u_{i,jpp}-u_{j,ipp}\right)-\frac{\left(\text{\ensuremath{\alpha}}_{1}-\alpha_{2}\right)}{4}\left(u_{i,ppj}-u_{j,ppi}\right)\right]=\frac{\left(\alpha_{1}+\alpha_{2}\right)}{4}\left(u_{i,j}-u_{j,i}\right)_{,pp},
\end{align*}
and
\begin{equation}
B_{ij}=\frac{\alpha_{1}}{2}\left(u_{i,j}-u_{j,i}\right)_{,p}n_{p}+\frac{\left(\text{\ensuremath{\alpha}}_{1}-\alpha_{2}\right)}{4}\left(\left(u_{p,i}-u_{i,p}\right)_{,p}n_{j}+\left(u_{j,p}-u_{p,j}\right)_{,p}n_{i}\right).\label{eq:Btens}
\end{equation}
With reference to eqs. \eqref{eq:IdentBulk-1} and \eqref{eq:Mn},
it can be recognized that
\begin{align}
\widetilde{\tau}&=\mathrm{Div}\,\widetilde{\mathfrak{m}}=\frac{1}{2}\mathrm{anti}\,\mathrm{Div}\,\widetilde{m}=\frac{\alpha_{1}+\alpha_{2}}{2}\Delta(\skew\nabla u),\label{eq:Tensors_Direct_Approach}\\
B&=\widetilde{\mathfrak{m}}\cdot n=\frac{1}{2}\mathrm{anti}\mathrm{\:}\left(\widetilde{m}\cdot n\right)=\alpha_{1}\left[\nabla\left(\mathrm{skew}\,\nabla u\right)\right]\cdot n+\frac{\left(\alpha_{1}-\alpha_{2}\right)}{2}\:\mathrm{skew}\left[\nabla\left(\mathrm{Div}\, u\right)\otimes n-\mathrm{Div}\left(\text{\ensuremath{\nabla}}u\right)\otimes n\right],\label{eq:B_Def}
\end{align}
with $\widetilde{\mathfrak{m}}$ and $\widetilde{m}$ given in \eqref{eq:DoubleStressCompact}
and \eqref{eq:mMindlin} respectively.

\begin{remark}We explicitly remark at this point (and we will point
it out more precisely in the next section) that the results presented
by Mindlin and Tiersten \cite{Mindlin62} are compatible with a variational
procedure which stops at this point (eq. \eqref{eq:PintMindlin})
without proceeding further in the process of integration by parts. \end{remark}

Indeed, in the view of proceeding towards the determination of strongly independent virtual variations, the last term
in the expression \eqref{eq:PintMindlin} of the power of internal
forces can still be  manipulated according to the procedure \eqref{eq:IPP_Surf}
of surface integration by parts, so that one finally gets
\begin{eqnarray*}
\int_{\partial\Omega}B_{ij}\delta u_{i,j}\, da & = & \int_{\partial\Gamma}\jump{B_{ij}\nu_{p}\delta u_{i}}\, ds-\int_{\partial\Omega}\left(T_{pj}B_{ij}\right)_{,h}T_{hp}\delta u_{i}\, da+\int_{\partial\Omega}\left(B_{ij}n_{j}\right)\delta u_{i}^{n}\, da.
\end{eqnarray*}
Hence, supposing that the virtual displacement is continuous through
the curves $\partial \Gamma$, the power of internal forces of
the couple stress model calculated by means of a direct approach reads
\[
\mathcal{P}^{\rm int}=\int_{\Omega}\left(\sigma_{ij}-\widetilde{\tau}_{ij}\right)_{,j}\delta u_{i}\, dv-\int_{\partial\Omega}\left[\left(\sigma_{ij}-\widetilde{\tau}_{ij}\right)n_{j}-\left(T_{pj}B_{ij}\right)_{,h}T_{hp}\right]\delta u_{i}\, da-\int_{\partial\Omega}\left(B_{ij}n_{j}\right)\delta u_{i}^{n}\, da-\int_{\partial\Gamma}\jump{B_{ij}\nu_{p}}\delta u_{i}\, ds.
\]
It has already been proven in Subsection \ref{sub:The-full-approach}
that only the tangent part of the normal derivative $\delta u_{i}^{n}$
contributes to the power of internal actions when considering the
indeterminate couple-stress model, so that the power of internal
actions can be finally written as
\begin{align}
\mathcal{P}^{\rm int}=&\int_{\Omega}\left(\sigma_{ij}-\widetilde{\tau}_{ij}\right)_{,j}\delta u_{i}\, dv-\int_{\partial\Omega}\left[\left(\sigma_{ij}-\widetilde{\tau}_{ij}\right)n_{j}-\left(T_{pj}B_{ij}\right)_{,h}T_{hp}\right]\delta u_{i}\, da-\int_{\partial\Omega}\left(T_{ip}B_{pj}n_{j}\right)\left(T_{ih}\delta u_{h}^{n}\right)\, da\notag\\&-\int_{\partial\Gamma}\jump{B_{ij}\nu_{p}}\delta u_{i}\, ds,
\end{align}
or equivalently, in compact form
\begin{align}\label{sepMind}
\mathcal{P}^{\rm int}  =&\int_{\Omega}\left.\langle  \mathrm{Div}\left(\sigma-\widetilde{\tau}\right)\delta u\right.\rangle\, dv -\int_{\partial\Omega}\left.\langle  \left[\left(\sigma-\widetilde{\tau}\right)\cdot n-\left(\nabla\left(B\cdot T\right)\right):T\right],\delta u\right.\rangle \, da
\\
 & -\int_{\partial\Omega}\left.\langle  \left(T\cdot B\cdot n\right),T\cdot\delta(\nabla u\cdot n)\right.\rangle\, da -\int_{\partial\Gamma}\left.\langle  \jump{B\cdot\nu},\delta u\right.\rangle\, ds ,\notag
\end{align}
where we recall once again that the tensors $\widetilde{\tau}$ and $B$
are given by eqs. \eqref{eq:Tensors_Direct_Approach}, \eqref{eq:B_Def}.

\noindent Considering the power of external actions to take the form
\eqref{eq:Pext} imposing $\mathcal{P}^{\rm int}+\mathcal{P}^{\rm ext}=0$
and localizing, one gets the bulk equations and associated traction
boundary conditions for the couple stress model by means of a direct
approach
\begin{gather}
\mathrm{Div}\left(\sigma-\widetilde{\tau}\right)+f^{\rm ext}=0\qquad\text{{in\,\ duality\,\ of }}\,\delta u,\label{eq:Bulk_CS-1}
\end{gather}
together with the following traction boundary conditions on the portion
of the boundary $\partial\Omega\setminus\overline{\Gamma}$

\begin{align}
\left(\sigma-\widetilde{\tau}\right)\cdot n-\left[\nabla\left(B\cdot T\right)\right]:T&=t^{\rm ext}\qquad\qquad\ \text{{in\,\ duality\,\ of }}\,\delta u,\label{eq:Boundary_CS-3}
\\
T\cdot B\cdot n&=T\cdot g^{\rm ext}\qquad\,\, \text{{in\,\ duality\,\ of }}\,\left(\nabla\delta u\right)\cdot n\label{eq:Boundary_CS-1-1}
\end{align}
and finally the following condition on the boundary of the boundary
$\partial\Gamma$ where traction is assigned
\begin{gather}
\jump{B\cdot\nu}=\pi^{\rm ext}\qquad\text{{in\,\ duality\,\ of }}\,\delta u.\label{eq:Boundary_CS-2-1}
\end{gather}
Given the identification of the tensors $\widetilde{\tau}$ and $B$ with
the tensors $\widetilde{\mathfrak{m}}$ and $\widetilde{m}$ as specified
in eqs. \eqref{eq:Tensors_Direct_Approach}, \eqref{eq:B_Def}, the bulk equations and
traction boundary conditions \eqref{eq:Bulk_CS-1}-\eqref{eq:Boundary_CS-2-1}
as derived by means of a direct approach are completely equivalent
to eqs. \eqref{eq:Bulk_Fin}-\eqref{eq:Boundary_Fin_3}.

\subsection{The geometric and traction, strongly independent, boundary conditions for the indeterminate
couple stress model}\label{bcbune}

We have proven up to now that, independently of the method that one wants to choose to obtain
the correct set of bulk equations and associated boundary conditions, passing  through a full second
gradient approach or a direct approach based on second order tensors instead the third order ones, one finally arrives at the following complete set of boundary conditions which can be used to complement the bulk equilibrium equation \eqref{eq:Bulk_CS-1} of the couple stress model:
\subsubsection{Geometric (essential or kinematical), strongly independent, boundary conditions on $\Gamma$}
\begin{align}
u & = u^{\rm ext}\qquad\qquad\qquad\qquad\quad\quad\,\qquad\qquad\qquad\qquad\qquad(3\ \text{bc})\label{bc1111}\\
(\id-n\otimes n)\cdot (\nabla {u}\cdot n) & =(\id-n\otimes n)\cdot a^{\rm ext}\qquad\qquad\quad\ \ \quad  \qquad\qquad\qquad\qquad(2\ \text{bc})\notag
\end{align}
where $ u^{\rm ext}, a^{\rm ext}:\mathbb{R}^{3}\rightarrow\mathbb{R}^{3}$
are prescribed functions on the subportion $\Gamma$ of the boundary $\partial \Omega$, where kinematical boundary conditions are assigned.

\subsubsection{Traction, strongly independent, boundary conditions on $\partial \Omega\setminus\overline{\Gamma}$ }
Correspondingly  to the geometric boundary conditions, we may prescribe
the following traction boundary conditions based on \eqref{eq:Boundary_CS-3} (or equivalently \eqref{eq:Boundary_Fin_1})
\begin{align}
\left.\begin{array}{rcl}
(\sigma-\widetilde{\tau})\cdot n-\frac{1}{2}\nabla[\anti(\widetilde{m}\cdot n)\cdot T]:\, T & = & t^{\rm ext},\vspace{1.2mm}\\
\dd(\id-n\otimes n)\cdot\anti(\widetilde{m}\cdot n)\cdot n & = & (\id-n\otimes n)\cdot g^{\rm ext},
\end{array}\right\} \qquad & \text{on}\ \partial{\Omega}\setminus\overline{\Gamma}\qquad\begin{array}{r}
(3\ \text{bc})\vspace{1.2mm}\\
(2\ \text{bc})
\end{array}\label{bc1002}\\
\begin{array}{rrl}
\frac{1}{2}\jump{\anti(\widetilde{m}\cdot n)\cdot \nu}&  & =\  {\pi}^{\rm ext},\qquad\qquad\qquad\ \,\ \,
\end{array} \qquad & \text{on}\ \partial{\Gamma}\quad\qquad\ \ \begin{array}{r}
(3\ \text{bc})
\end{array}\notag
\end{align}
where $t^{\rm ext}, g^{\rm ext}:\mathbb{R}^{3}\rightarrow\mathbb{R}^{3}$
are prescribed functions on $\partial\Omega\setminus\overline{\Gamma}$, while
$ {\pi}^{\rm ext}$ is prescribed on $\partial\Gamma$ and leads to
3 boundary conditions on the curve $\partial \Gamma$.

It can be shown (see Appendix \ref{identities_tractions} for the proof of the needed identities \eqref{propr2} and \eqref{propr3}) that such set of traction boundary conditions can be ulteriorly simplified in the following form

\begin{align}
\left.\begin{array}{rcl}
(\sigma-\widetilde{\tau})\cdot n-\frac{1}{2}\nabla[\anti(\widetilde{m}\cdot n)\cdot T]:\, T & = & t^{\rm ext},\vspace{1.2mm}\\
\dd(\id-n\otimes n)\cdot\anti( (\id-n\otimes n)\cdot \widetilde{m}\cdot n)\cdot n & = & (\id-n\otimes n)\cdot g^{\rm ext},
\end{array}\right\} \qquad & \text{on}\ \partial{\Omega}\setminus\overline{\Gamma}\qquad\begin{array}{r}
(3\ \text{bc})\vspace{1.2mm}\\
(2\ \text{bc})
\end{array}\label{bc1002_bis}\\
\begin{array}{rcl}
\frac{1}{2}\jump{\anti((\id-n\otimes n)\cdot \widetilde{m}\cdot n)\cdot \nu}&  & =\    {\pi}^{\rm ext},\qquad\qquad\qquad\ \,\ \,
\end{array} \qquad & \text{on}\ \partial{\Gamma}\quad\qquad\ \ \begin{array}{r}
(3\ \text{bc})
\end{array}\notag
\end{align} 
where $t^{\rm ext}, g^{\rm ext}:\mathbb{R}^{3}\rightarrow\mathbb{R}^{3}$
are prescribed functions on $\partial\Omega\setminus\overline{\Gamma}$, while
$ {\pi}^{\rm ext}$ is prescribed on $\partial\Gamma$ and leads to
3 boundary conditions on the curve $\partial \Gamma$.

\section{Assessment of the strongly independent boundary conditions for the
indeterminate couple stress model in a form directly comparable to Mindlin and Tiersten's ones}\label{sectaxlb}\setcounter{equation}{0}

Given that the bulk equations \eqref{eq:Bulk_CS-1} that we obtained are the same as Mindlin and Tiersten's ones, the delicate point is now to compare our boundary
conditions \eqref{bc1111}-\eqref{bc1002} with those provided by Mindlin and Tiersten in \cite{Mindlin62}.
 If a proof of the equivalence of the purely kinematical boundary conditions \eqref{bc1111} with those proposed by Mindlin and Tiersten has already been provided in Lemma \ref{lemmadesprecurlaxl}, the equivalence between traction boundary conditions as derived with our and Mindlin's approach is not straightforward.   This is why we need here to rewrite the boundary conditions  \eqref{bc1002_bis} in a suitable form.

\subsection{Towards a direct comparison with Mindlin's traction boundary conditions}\label{toMbc}

In order to be able to directly compare the traction boundary conditions for
the  indeterminate couple stress model which we obtained both
passing through a second gradient theory and by means of a direct
approach with those proposed by Mindlin, we need to rewrite our equations
in a suitable form. In this section we show some calculations which are needed in order
to reach this goal.

\begin{proposition}\label{prop:ComparisonToMindlin} For all $\widetilde{m}\in \mathbb{R}^{3\times 3}$  and for all smooth surfaces $\Sigma=\{(x_{1},x_{2},x_{3})\in \mathbb{R}^3\, |\, F(x_{1},x_{2},x_{3})=0\}$,\break  $F:\omega\subset \mathbb{R}^3\rightarrow\mathbb{R}^3$ of class $C^2$,  the following identity is satisfied:
\begin{align}
\frac{1}{2}\nabla[\mathrm{anti}\left(\widetilde{m}\cdot n\right)\cdot T\,]:T=\frac{1}{2}n\times\nabla\left[\left.\langle  n,(\mathrm{sym}\,\widetilde{m})\cdot n\right.\rangle \right]+\frac{1}{2}\nabla\left[\mathrm{anti}\left(T\cdot\widetilde{m}\cdot n\right)\cdot T\right]:T.
\end{align}
\end{proposition}
\begin{proof}
The proof is included in Appendix \ref{appendixprop:ComparisonToMindlin}
\end{proof}

\begin{remark}
From the above proposition, it follows that the first of the boundary conditions \eqref{bc1002}
(or equivalently \eqref{eq:Boundary_Fin_1}) can be finally re-written
in the form
\begin{align}
\left(\sigma-\widetilde{\tau}\right)\cdot n-\frac{1}{2}n\times\left[\nabla\left(\left.\langle  n,(\mathrm{sym}\,\widetilde{m})\cdot n\right.\rangle \right)\right]-\frac{1}{2}\nabla\left[\mathrm{anti}\left(T\cdot\widetilde{m}\cdot n\right)\cdot T\right]:T=t^{\rm ext}.\label{conclusionMTU}
\end{align}
\end{remark}
In  Section \ref{MindlinReas} we have recalled the argument of Mindlin and Tiersten
and we have remarked, see \eqref{bcme2}, that the term $-\frac{1}{2}\nabla\left[\mathrm{anti}\left(T\cdot\widetilde{m}\cdot n\right)\cdot T\right]:T$
is absent in their formulation since it remains somehow hidden in duality of $\mathrm{curl}(\delta  u)$ which is not manipulated further in their formulation.

\subsection{Final form of the strongly independent, geometric and traction boundary conditions for the indeterminate
couple stress model}\label{toMincom}

Basing ourselves on the previously results obtained in Subsection \ref{sub:identification-1},
we can now establish which is the set of geometric and traction
boundary conditions to be used in the  indeterminate couple stress
model, alternatively to those proposed by Mindlin and Tiersten. As we will better point out in the remainder of this section, the boundary conditions that we derive by our direct approach are as legitimate as those proposed by Mindlin and Tiersten. Nevertheless, if in one case one can equivalently pass from one set of imposed boundary conditions to the other one, such equivalence cannot be stated for the case of mixed boundary conditions.

\subsubsection{Geometric (kinematical, essential) strongly independent boundary conditions for the indeterminate couple
stress model}

As for the geometric boundary conditions, we recall that one can assign on $\Gamma\subset\partial\Omega$
the following conditions
\begin{align}
u& = u^{\rm ext}\qquad\qquad\quad\quad\quad\quad\quad\quad\qquad\qquad\qquad\qquad\ (3\ \text{bc})\notag\label{bc1110}\\
{}(\id-n\otimes n)\cdot\nabla u\cdot n\, & =(\id-n\otimes n)\cdot  a^{\rm ext}\,\qquad\quad\quad\ \qquad\qquad\qquad\qquad(2\ \text{bc})
\end{align}
where $ u^{\rm ext},a^{\rm ext}:\mathbb{R}^{3}\rightarrow\mathbb{R}^{3}$
are prescribed functions.
Such conditions are the geometric boundary conditions which are known
to be valid in the framework of second gradient theories, with the
peculiarity that here only the tangent part of the normal derivative
of displacement can be assigned here.

We have already shown that the fact of assigning the tangent part of $\nabla u\cdot n$ is indeed  equivalent to assigning the tangent part of $\mathrm{curl}\, u$, so that such set of geometric boundary conditions can be seen to be equivalent to Mindlin and Tiersten one's according to Lemma \ref{equivalentBCs}.

\subsubsection{ Traction strongly independent boundary conditions for the indeterminate
couple stress model}

As far as the traction boundary conditions are concerned, considering
 the manipulated form \eqref{conclusionMTU} of equation \eqref{bc1002_bis}$_1$,
the strongly independent boundary conditions \eqref{bc1002_bis} for the indeterminate couple stress model can be finally rewritten as
\begin{align}
\left.\begin{array}{rcl}
\label{bc1001}[(\sigma-\widetilde{\tau})\cdot n-\frac{1}{2}n\times\nabla[\langle n,(\sym\widetilde{m})\cdot n\rangle]\hspace{2cm} &  & \vspace{1.2mm}\\
-\frac{1}{2}\nabla[(\anti[(\id-n\otimes n)\cdot\widetilde{m}\cdot n])\cdot T]:T & = & t^{\rm ext},\vspace{1.2mm}\\
\dd(\id-n\otimes n)\cdot\anti[(\id-n\otimes n)\cdot\widetilde{m}\cdot n]\cdot n & = & (\id-n\otimes n)\, g^{\rm ext}
\end{array}\right\}  &\quad \text{on}\quad\partial{\Omega}\setminus\overline{\Gamma}\quad\begin{array}{r}
(3\ \text{bc})\vspace{1.2mm}\\
(2\ \text{bc})
\end{array}\\
\begin{array}{rcl}
\dd\jump{\anti[(\id-n\otimes n)\cdot\widetilde{m}\cdot n]\cdot\nu} &  & =\hspace{0.15cm} {\pi}^{\rm ext},\hspace{2.15cm}\end{array} & \quad \text{on}\quad \partial{\Gamma}\quad\quad\ \ \begin{array}{r}
(3\ \text{bc}),\end{array}\notag
\end{align}
where $t^{\rm ext}, g^{\rm ext}:\mathbb{R}^{3}\rightarrow\mathbb{R}^{3}$
are prescribed functions on $\partial\Omega\setminus\overline{\Gamma}$,
while $ {\pi}^{\rm ext}:\mathbb{R}^{3}\rightarrow\mathbb{R}^{3}$ is
prescribed on $\partial\Gamma$ and leads to 3 boundary conditions.

In this section we have deduced the strongly independent traction boundary conditions
which are coming in a natural way from second gradient elasticity and we have compared them to those presented by Mindlin and Tiersten thus showing their apparent disagreement.

\section{Are Mindlin and Tiersten's weakly independent boundary conditions equivalent to our strongly independent ones?}

Up to this point, we have shown that the boundary conditions  derived by Mindlin and Tiersten \cite{Mindlin62} for the indeterminate couple stress model are not directly superposable to those that we obtain by means of a standard variational approach in the spirit of second gradient theories.

Even if these sets of boundary conditions are formally not the same, they both follow from the same strain energy density.
The only difference that we can point out in the two approaches is related to the process of integration by parts which is perfomed on the action functional based upon the considered strain energy density.
Indeed, Mindlin and Tiersten's boundary conditions are only "weakly independent", while those obtained by means of our direct approach can be considered to be "strongly independent" in the sense established in Subsection \ref{newsubsectionmadeo}.

To the sake of compacteness, we use in the sequel the following notations for the internal tractions and hypertractions respectively as obtained by Mindlin and Tiersten's and our approach
\begin{align}
\widetilde{t}^{\rm int}&:=\ \left(\sigma-\frac{1}{2}\,\anti({\rm Div}\,\widetilde{m})\right)\cdot\, n-\frac{1}{2}n\times\nabla[\langle n,(\sym\widetilde{{m}})\cdot n\rangle],&\text{Mindlin-Tiersten's formulation}\notag\\
t^{\rm int} &:=\ \Big(\sigma  -\frac{1}{2}\,\anti({\rm Div}\,\widetilde{m})\Big)\cdot n-\frac{1}{2}n\times\nabla[\langle n,(\sym\widetilde{{m}})\cdot n\rangle] & \text{our formulation}\notag\\
&\qquad-\frac{1}{2}\nabla\left[\:\anti\left(\:(\id-n\otimes n)\cdot\widetilde{{m}}\cdot n\:\right)\cdot(\id-n\otimes n)\:\right]:(\id-n\otimes n),& \label{definitions_tractions} \\
\widetilde{g}^{ \rm int}&:=(\id-n\otimes n)\cdot \widetilde{{m}}\cdot n, &\text{Mindlin-Tiersten's formulation}\notag\\
g^{\rm int}&:=(\id-n\otimes n)\cdot\anti[(\id-n\otimes n)\cdot\widetilde{{m}}\cdot n]\cdot n=\anti[(\id-n\otimes n)\cdot\widetilde{{m}}\cdot n]\cdot n. &\text{our formulation}\notag
\end{align}

In the last equality for $g^{\rm int}$ we have used the fact that $g^{\rm int}$ is indeed the dual of $T\cdot \nabla (\delta u) \cdot n$ which is a vector tangent to the boundary, which is equivalent to say that the normal part of $g^{\rm int}$ does not intervene in the balance equations.

To the sake of of simplicity, we summarize the two sets of possible geometric and traction boundary conditions as obtained by Mindlin and Tiersten and by ourselves in the following summarizing box
\begin{center}
	{\small {\bf Table 1. Possible  sets of boundary conditions in the indeterminate couple stress model. \medskip\medskip}}
	\begin{tabular}{|l|l|l|}
		\hline& &\\
		&{\bf Mindlin and Tiersten} &{\bf  Our approach}
		\\
		\hline& &
		\\
		{\bf 	Geometric} & (I) \ \  ${u}=\widetilde{u}^{\rm ext}$, & (III) \ $
		{u}=u^{\rm ext}
		$
		\\
		& (II) \  $ T\cdot\curl u= T\cdot \widetilde{a}^{\rm ext}$ & (IV) $
		T\cdot\nabla u \cdot n= T\cdot a^{\rm ext}
		$
		\\
		\hline& &
		\\
		{\bf Traction }&(A) \ \ $ \widetilde{t}^{\ \rm int}=\widetilde{t}^{\rm ext}$, &   (C)\ \  $
	{t}^{\rm int}={t}^{\rm ext}
		$
		\\
	& (B)\ \  $ \widetilde{g}^{\rm int}=\widetilde{g}^{\rm ext}$ & (D)\ \ ${g}^{\rm int}={g}^{\rm ext}
		$
		\\

		\hline& &
		\\
		{\bf 	Mixed BCs 1} & (I) \ \  ${u}=\widetilde{u}^{\rm ext}$, & (III) \ $
		{u}=u^{\rm ext}
		$
		\\
		& (B) \   $ \widetilde{g}^{\rm int}=\widetilde{g}^{\rm ext}$  & (D) ${g}^{\rm int}={g}^{\rm ext}
		$
		\\
		\hline& &
		\\
		{\bf 	Mixed BCs 2} & (II) \  $  T\cdot\curl u= T\cdot \widetilde{a}^{\rm ext}$, & (IV) $
		T\cdot\nabla u \cdot n= T\cdot a^{\rm ext}
		$
		\\
		& (A) \ \ $ \widetilde{t}^{\ \rm int}=\widetilde{t}^{\rm ext}$  &  (C)\ \  $
	{t}^{\rm int}={t}^{\rm ext}
		$
		\\
		\hline
	\end{tabular}
\end{center}

The problem now arises to establish the equivalence between analogous sets of boundary conditions in the two approaches. Since all the presented boundary conditions arise from the same strain energy density, we would naively  expect a complete equivalence between the two models. We will instead show that, if a direct equivalence can be established in some cases, this is not indeed feasible for all possible sets of boundary conditions that may be introduced in couple-stress continua. 
More particularly, we individuate different possible sets of boundary conditions that are allowed in couple-stress continua being compatible with the Principle of Virtual Powers as settled in Mindlin's and Tiersten's and our approach respectively:

\begin{itemize}
\item \textbf{Fully geometric boundary conditions}. The boundary conditions (I) and (II) (Mindlin and Tiersten) \textbf{or} (III) and (IV) (our approach) are simultaneously assigned on the same portion of the boundary.
\item \textbf{Fully traction boundary conditions}. The boundary conditions (A) and (B) (Mindlin and Tiersten) \textbf{or} (C) and (D) (our approach) are simultaneously assigned on the same portion of the boundary.
\item \textbf{Mixed 1: displacement/double-force boundary conditions}. The boundary conditions (I) and (B) (Mindlin and Tiersten) \textbf{or} (III) and (D) (our approach) are simultaneously assigned on the same portion of the boundary.
\item \textbf{Mixed 2:  force/$D^1(u)$ boundary conditions}. The boundary conditions (II) and (A) (Mindlin and Tiersten) \textbf{or} (IV) and (C) (our approach) are simultaneously assigned on the same portion of the boundary. We recall that by $D^1(u)$ we compactly indicate the operator $T\cdot \mathrm{curl}\,u$ when we consider Mindlin and Tiersten's approach {\bf or}  the operator $T\cdot \nabla u\cdot n$ when considering our approach.
\end{itemize}
We explicitly remark that, in order to be consistent with the introduced Principle of Virtual Powers, when the first sets of conditions is applied on a portion $\Gamma$ of the boundary, the second ones must be assigned on the portion $\partial \Omega \setminus \Gamma$. Analogously, when assigning the third set of boundary conditions on  $\Gamma$, the fourth one must be assigned on $\partial \Omega \setminus \Gamma$.

In the following subsections we carefully study the four introduced cases by establishing whether Mindlin and Tiersten's approach is equivalent with our formulation of the indeterminate couple-stress model.

\subsection{Fully kinematical boundary conditions}
We have already shown (see Lemma \ref{lemmadesprecurlaxl}) that it is equivalent to simultaneously assign the displacement and the tangential part of its curl \textbf{or} the displacement and the tangential part of its normal derivative \textbf{on the same portion of the boundary}.  This means that Mindlin and Tiersten's boundary conditions (I)+(II) are completely equivalent to our conditions (III)+(IV) in the sense that one system of equations can be directly obtained from the other and vice-versa. We would like to thank an unknown reviewer for pointing out to us this equivalency.

\subsection{Fully traction boundary conditions} \label{fullyTraction}
We consider here the case in which forces and double forces are simultaneously applied \textbf{on the same portion of the boundary}. More particularly, this means that we are simultaneously applying on the same portion of the boundary conditions (A) and (B) (Mindlin and Tiersten's) \textbf{or} conditions (C) and (D) (our approach). 

We start by showing that conditions (B) and (D) are equivalent. To do so, we notice that, given two vectors $v$ and $n$, one can check that, according to definitions \ref{eq:AxlAnti_indices}, the following equalities hold $$(\mathrm{anti}\,v)_{ij} \, n_j=-\epsilon_{ijk} \, v_k \, n_j=-(n\times n)_i=(v\times n)_i.$$
This means that we can write:
$$g^{\rm int}=\anti[(\id-n\otimes n)\cdot\widetilde{{m}}\cdot n]\cdot n=((\id-n\otimes n)\cdot\widetilde{{m}}\cdot n)\times n,$$
so that the boundary condition (D) can be rewritten as $((\id-n\otimes n)\cdot\widetilde{{m}}\cdot n)\times n=g^{\rm ext}$. Comparing this last equation with equation (B), we can finally conclude by direct inspection that equations (B) and (D) are equivalent when setting 
\begin{equation}
g^{\rm ext}=\widetilde g^{\rm ext}\times n. \label{ident_df}
\end{equation}

On the other hand, eq. (C) can be rewritten as

$$ \Big(\sigma  -\frac{1}{2}\,\anti({\rm Div}\,\widetilde{m})\Big)\cdot n-\frac{1}{2}n\times\nabla[\langle n,(\sym\widetilde{{m}})\cdot n\rangle]=t^{\rm ext}+\frac{1}{2}\nabla\left[\:\anti\left(\:T\cdot\widetilde{{m}}\cdot n\:\right)\cdot T \:\right]:T,$$
which, considering eq. (B) can also be rewritten as
$$ \Big(\sigma  -\frac{1}{2}\,\anti({\rm Div}\,\widetilde{m})\Big)\cdot n-\frac{1}{2}n\times\nabla[\langle n,(\sym\widetilde{{m}})\cdot n\rangle]=t^{\rm ext}+\!\!\!\!\!\!\!\!\!\!\!\!\!\!\!\!\!\!\!\!\!\!\!\!\!\!\!\!\!\!\!\!\!\!\underbrace{\frac{1}{2}\nabla\left[\:\anti\left(\:\widetilde g^{\rm ext}\:\right)\cdot T \:\right]:T}_{\text{already known since only the tangetial derivatives are considered}}.$$
It is easy to check that this last equation is equivalent to eq. (A) when setting
\begin{equation}
t^{ \rm ext}=\widetilde t^{\rm ext}-\frac{1}{2}\nabla\left[\:\anti\left(\:\widetilde g^{\rm ext}\:\right)\cdot T \:\right]:T. \label{ident_force}
\end{equation}
We have thus proved that, in the case of fully traction boundary conditions, given a couple of tractions $(\widetilde t^{\rm ext},\, \widetilde g^{\rm ext})$ in Mindlin and Tiersten's model, one can always ``a priori'' find a corresponding pair of  tractions $(t^{\rm ext},g^{\rm ext})$ in our model such that the two sets of boundary conditions (A)+(B) and (C)+(D) are equivalent thanks to the relationships \eqref{ident_df} and \eqref{ident_force}. The converse is clearly also true.

\subsection{Mixed 1: displacement/double-force boundary conditions}
We treat here the case in which we simultaneously assign the displacement and the double force on the same portion of the boundary. This is equivalent to say that one is assigning eqs. (I) and (B) in the Mindlin and Tiersten's approach \textbf{or} (III) and (D) when considering our approach.
As already proven in the previous subsection, the equivalence between equations (B) and (D) can be obtained by setting the relationship \eqref{ident_df} between $ g^{\rm ext}$ and  $ \widetilde g^{\rm ext}$. Moreover eqs. (I) and (III) are clearly equivalent when $  u^{\rm ext}= \widetilde u^{\rm ext}.$

\subsection{Mixed 2: force/$D^1(u)$ boundary conditions}
We have shown up to now that for the three preceeding cases of boundary conditions an \textit{a priori} equivalence can be established between Mindlin and Tiersten's and our couple stress model. More particularly, this means that, given a boundary value problem stemming from Mindlin and Tiersten's model, we can set up (thanks to suitable identifications between tractions and double tractions in the two models) another boundary value problem which give rise to the same solution. We show here that the establishment of such an \textit{a priori} equivalence is not possible in this last "Mixed 2" case in which forces and higher derivatives of the displacement field are simultaneously assigned on the same portion of the boundary.
We start by noticing that if the displacement field is not assigned on the boundary, then  Lemma \ref{lemmadesprecurlaxl}) is not valid any more and, as a consequence, equations (II) and (IV) are no longer equivalent. Indeed, starting from the equations \eqref{Curl_NormDer}, it is possible to easily deduce that the tangent part of the  normal derivative and of the $\mathrm{curl}$ of the displacement field are respectively given by 
\begin{equation}T\cdot \nabla u \cdot n= \left( \frac{\partial u_\tau}{\partial x_n}, \ \frac{\partial u_\nu}{\partial x_n}, \ 0 \right)^T  \quad \text{and} \quad  
T \cdot \mathrm{curl}\,     u=\left( \frac{\partial u_n}{\partial x_\nu}- \frac{\partial u_\nu}{\partial x_n}, \ \frac{\partial u_\tau}{\partial x_n}- \frac{\partial u_n}{\partial x_\tau}, \ 0 \right)^T .
\end{equation}
When, as in this case, the displacement is not assigned on the boundary, both its normal and tangential derivatives are free, so that we cannot establish an \textit{a priori} equivalence between equations (II) and (IV). The same is true when one wants to compare equations (A) and (C). In fact, we can always recognize that eq. (C) can be rewritten as:
$$ \Big(\sigma  -\frac{1}{2}\,\anti({\rm Div}\,\widetilde{m})\Big)\cdot n-\frac{1}{2}n\times\nabla[\langle n,(\sym\widetilde{{m}})\cdot n\rangle]=t^{\rm ext}+\frac{1}{2}\nabla\left[\:\anti\left(\:\widetilde g^{\rm ext}\:\right)\cdot T \:\right]:T.$$ Nevertheless, contrary to the case treated in Subsection \ref{fullyTraction}, this last equation cannot be claimed to be equivalent to eq. (A) just by setting
\begin{equation}
t^{ \rm ext}=\widetilde t^{\rm ext}-\frac{1}{2}\nabla\left[\:\anti\left(\: T\cdot \widetilde m \cdot n\:\right)\cdot T \:\right]:T=\widetilde t^{\rm ext}-\frac{1}{2}\nabla\left[\:\anti\left(\:\widetilde g^{\rm int}(u)\:\right)\cdot T \:\right]:T. \label{ident_force1}
\end{equation}
As a matter of fact, the term $g^{\rm int}(u)$ appearing in equation \eqref{ident_force1} and defined in eq. \eqref{definitions_tractions}$_3$ depends on the displacement field $u$ through the couple-stress tensor $\widetilde m$ defined in \eqref{eq:mMindlin}. This means that, actually, formula \eqref{ident_force1} does not allow to calculate \textit{a priori} the force $t^{ \rm ext}$ in our model which is equivalent to an assigned $\widetilde t^{\rm ext}$ in Mindlin and Tiersten's model. Of course, one could think to assign, e.g., the boundary conditions (II) and (A) on the same boundary, solve the associated boundary value problem so finding its solution $u^*$ and then calculate the quantities $a^{\rm ext}$ and $t^{\rm ext}$ to be assigned in our model in order to give rise to the same solution $u^*$. The converse operation of imposing our boundary conditions (IV) and (C) in our model and then calculate \textit{a posteriori} the quantities $\widetilde t^{\rm ext}$ and $\widetilde a^{\rm ext}$ to be assigned in Mindlin and Tiersten's model to obtain the same solution, can also be envisaged. In other words, we are saying that only an \textit{a posteriori} equivalence is eventually possible in the case that the considered force/$D^1(u)$ boundary conditions are applied on the boundary. This fact poses, at least, some philosophical problems by giving rise to the question: \textit{how to chose among Mindlin and Tiersten's or our mixed boundary conditions?}. This indeterminacy leaves many questions open concerning the physical transparency of higher gradient theories. 

\section{Conclusions}
The present paper gives a comprehensive analysis of the indeterminate couple stress model and of the  boundary conditions arising in this theory. We have seen the indeterminate couple stress model as a special case of the full strain gradient elasticity, we have directly derived the equilibrium equations and the boundary conditions in the same spirit as in the general strain gradient elasticity approach and we have compared our approach with that proposed by Mindlin and Tiersten \cite{Mindlin62}.

As a balance for the present paper we can state that an apparent inconsistency is found between the classical (Mindlin and Tiersten's) approach and our direct approach to the indeterminate couple stress model.

Indeed, if an "a priori" equivalence can be found in most cases between the two models, we point out that this is not the case when considering "mixed" boundary conditions for which "forces" and suitable combinations of first order derivatives of displacement are simultaneously assigned on the same portion of the boundary.

It turns out that for such particular mixed boundary conditions an "a priori" equivalence cannot be established between the two models.

This fact poses serious conceptual and philosophical problems concerning the transparency of the physical meaning of the boundary conditions in higher gradient models. The question remains open  whether the use of one model would be preferable to the other, at least for the quoted case of mixed boundary conditions.

\section*{Acknowledgement}  We thank  Samuel Forest  and for the discussions on   a prior version of the paper and to the unknown reviewers for the helpful suggestions and comments.  The work of the second author was partial  supported by a grant of the Romanian National Authority for Scientific Research and Innovation, CNCS-UEFISCDI, project number PN II-ID-PCE-2011-3-0521,
contract nr. 144/2011.

\bibliographystyle{plain} 
\addcontentsline{toc}{section}{References}
\begin{footnotesize}

\end{footnotesize}

\begin{footnotesize}
\appendix
\section*{Appendix}
\addcontentsline{toc}{section}{Appendix}
\setcounter{section}{1}

\subsection{First variation of a second gradient action functional and principle
of virtual power in compact form}\setcounter{equation}{0}

In this section we basically propose again the calculations concerning
the first variation of a second gradient action functional by means
of a compact notation instead of using Levi-Civita index notation
as done instead in Section \ref{Second gradient}. To this purpose,
let us consider the second gradient energy $W(\nabla\mathbf{u},\nabla\nabla\mathbf{u})=W(\nabla u)+W_{{\rm curv}}(\nabla\nabla u)$
and the associated action functional in the static case (no inertia
considered here)
\[
\mathcal{A}=-\int_{\Omega}W(\nabla u,\nabla\nabla u)\,dv=-\int_{\Omega}[W(\nabla u)+W_{{\rm curv}}(\nabla\nabla u)]\,dv.
\]
The first variation of the action functional can be interpreted as
the power of internal actions $\mathcal{P}^{\rm int}$ of the considered
system and can be computed as follows
\[
\mathcal{P}^{\rm int}=\delta\mathcal{A}=-\int_{\Omega}\left(\langle D_{\nabla u}W(\nabla u),\nabla\delta u\rangle+\langle D_{\nabla\nabla u}W_{{\rm curv}}(\nabla\nabla u),\nabla\nabla\delta u\rangle\right)\,dv,
\]
Integrating a first time by parts and using the divergence theorem
we get
\begin{align*}
\delta\mathcal{A}= & \int_{\Omega}\langle{\rm Div}[D_{\nabla u}W(\nabla u)],\delta u\rangle \,dv-\int_{\partial\Omega}\langle D_{\nabla u}W(\nabla u)\cdot n,\delta u\rangle \,da\\
 & -\int_{\partial\Omega}\langle D_{\nabla\nabla u}W_{{\rm curv}}(\nabla\nabla u)\cdot n,\nabla\delta u\rangle \,da+\int_{\Omega}{\rm Div}[D_{\nabla\nabla u}W_{{\rm curv}}(\nabla\nabla u)],\nabla\delta u\rangle \,dv.
\end{align*}
Integrating again by parts the last bulk term we get
\begin{align*}
\delta\mathcal{A}= & -\int_{\partial\Omega}\langle\left\{ D_{\nabla u}W(\nabla u)-{\rm Div}[D_{\nabla\nabla u}W_{{\rm curv}}(\nabla\nabla u)]\right\} \cdot n,\delta u\rangle \,da\\
 & +\int_{\Omega}\langle{\rm Div}\left\{ D_{\nabla u}W(\nabla u)-{\rm Div}[D_{\nabla\nabla u}W_{{\rm curv}}(\nabla\nabla u)]\right\} ,\delta u\rangle \,dv-\int_{\partial\Omega}\langle D_{\nabla\nabla u}W_{{\rm curv}}(\nabla\nabla u)\cdot n,\nabla\delta u\rangle \,da,
\end{align*}
which can also be rewritten as
\begin{align}
\delta\mathcal{A}= & -\int_{\partial\Omega}\langle\left\{ \sigma-{\rm Div}[\mathfrak{m}]\right\} \cdot n,\delta u\rangle \,da+\int_{\Omega}\langle{\rm Div}\left\{ \sigma-{\rm Div}[\mathfrak{m}]\right\} ,\delta u\rangle \,dv-\int_{\partial\Omega}\langle\mathfrak{m}\cdot n,\nabla\delta u\rangle \,da\label{eq:Intermediate_var-1}
\end{align}
if one sets
\[
\sigma=D_{\nabla u}W(\nabla u),\quad\qquad\mathfrak{m}=D_{\nabla\nabla u}W_{{\rm curv}}(\nabla\nabla u).
\]
Using the brief digression concerning differential geometry (see Section
\ref{sub:Diff_Geometry}) and recalling the properties (\ref{eq:Projectors_Properties}),
we can now ulteriorly manipulate the last term in eq. (\ref{eq:Intermediate_var-1})
as follows
\begin{align}
\langle\mathfrak{m}\cdot n,\nabla\delta u\rangle & =\langle\mathfrak{m}\cdot n,\nabla\delta u\,\id\rangle=\langle\mathfrak{m}\cdot n,\nabla\delta u\,(T+Q)\rangle=\langle(\mathfrak{m}\cdot n)\cdot T,\nabla\delta u\rangle+\langle\mathfrak{m}\cdot n,(\nabla\delta u)\cdot Q\rangle\\
 & =\langle(\mathfrak{m}\cdot n)\cdot T\cdot T,\nabla\delta u\rangle+\langle\mathfrak{m}\cdot n,(\nabla\delta u)\cdot Q\rangle\notag =\langle\, T,T\cdot(\mathfrak{m}\cdot n)^{T}\cdot\nabla\delta u\rangle+\langle\mathfrak{m}\cdot n,(\nabla\delta u)\cdot\left(n\otimes n\right)\rangle\notag \\
 &=\langle\, T,T\cdot(\mathfrak{m}\cdot n)^{T}\cdot\nabla\delta u\rangle+\langle\{n\otimes[(\nabla\delta u)\cdot n]\}\mathfrak{m}\cdot n,\id\rangle\notag =\langle\, T,T\cdot(\mathfrak{m}\cdot n)^{T}\cdot\nabla\delta u\rangle+\langle n\otimes\{(\mathfrak{m}\cdot n)^{T}\cdot(\nabla\delta u)\cdot n\},\id\rangle\notag\\
 &=\langle\, T,T\cdot(\mathfrak{m}\cdot n)^{T}\cdot\nabla\delta u\rangle+\langle n,(\mathfrak{m}\cdot n)^{T}\cdot(\nabla\delta u)\cdot n\rangle\notag=\langle\, T,T\cdot(\mathfrak{m}\cdot n)^{T}\cdot\nabla\delta u\rangle+\langle(\mathfrak{m}\cdot n)\cdot n,(\nabla\delta u)\cdot n\rangle.\notag
\end{align}
We can hence recognize in the last term of this formula that the normal
derivative $\left(\frac{\partial\delta u}{\partial n}\right)_{i}=\left[(\nabla\delta u)\cdot n\right]_{i}=\delta u_{i,h}n_{h}$
of the displacement field appears. As for the other term, it can be
still manipulated, suitably integrating by parts for $\Gamma$ an open
subset of $\partial\Omega$ and then using the surface divergence
theorem (\ref{eq:Surface_Div_Th-2}), so that we can finally write
\begin{align}
\int_{\partial\Omega}\langle & \mathfrak{m}\cdot n,\nabla\delta u\rangle \,da=\int_{\partial\Omega}\langle\, T,\nabla[T\cdot(\mathfrak{m}\cdot n)^{T}\cdot\delta u]-\nabla[T\cdot(\mathfrak{m}\cdot n)^{T}]\cdot\delta u\rangle \,da+\int_{\partial\Omega}\langle(\mathfrak{m}\cdot n)\cdot n,(\nabla\delta u)\cdot n\rangle\, da\notag\\
 & =\int_{\partial\Gamma}\jump{\langle(\mathfrak{m}\cdot n)^{T}\cdot\delta u,\nu\rangle}\,ds-\int_{\partial\Omega}\langle\, T,\nabla[T\cdot(\mathfrak{m}\cdot n)^{T}]\cdot\delta u\rangle \,da+\int_{\partial\Omega}\langle(\mathfrak{m}\cdot n)\cdot n,(\nabla\delta u)\cdot n\rangle \,da\notag\\
 & =\int_{\partial{\Gamma}}\jump{\langle(\mathfrak{m}\cdot n)\cdot\nu,\delta u\rangle}\,ds-\int_{\partial\Omega}\langle\,\nabla[(\mathfrak{m}\cdot n)\cdot T]:T,\delta u\rangle \,da+\int_{\partial\Omega}\langle(\mathfrak{m}\cdot n)\cdot n,(\nabla\delta u)\cdot n\rangle \,da.
\end{align}
The final variation of the second gradient action functional given
in (\ref{eq:Intermediate_var-1}), can hence be finally written as
\begin{align}
\delta\mathcal{A} & =\int_{\Omega}\langle{\rm Div}\left\{ \sigma-{\rm Div}[\mathfrak{m}]\right\} ,\delta u\rangle \,dv\notag\label{eq:Pint-2}-\int_{\partial\Omega}\langle\left\{ \sigma-{\rm Div}[\mathfrak{m}]\right\} \cdot n-\nabla[(\mathfrak{m}\cdot n)\cdot T]:T,\delta u\rangle \,da\\
 & \qquad-\int_{\partial\Omega}\langle(\mathfrak{m}\cdot n)\cdot n,(\nabla\delta u)\cdot n\rangle \,da-\int_{\partial{\Gamma}}\jump{\langle(\mathfrak{m}\cdot n)\cdot\nu,\delta u\rangle}\,ds,
\end{align}
which is directly comparable with eq. \eqref{eq:Pint-1} obtained
by means of calculations via Levi-Civita index notation. If now one recalls
the principle of virtual powers according to which a given system
is in equilibrium if and only if the power of internal forces is equal
to the power of external forces, it is straightforward that the expression
(\ref{eq:Pint-2}) naturally suggests which is the correct expression
for the power of external forces that a second gradient continuum
may sustain, namely:
\[
\mathcal{P}^{{\rm ext}}=\int_{\Omega}\langle f^{{\rm ext}},\delta u\rangle \,dv+\int_{\partial\Omega\setminus\overline{\Gamma}}\langle t^{{\rm ext}},\delta u\rangle \,da+\int_{\partial\Omega\setminus\overline{\Gamma}}\langle M^{{\rm ext}},(\nabla\delta u)\cdot n\rangle \,da+\int_{\partial{\Gamma}}\langle\pi^{{\rm ext}},\delta u\rangle \,ds,
\]
where
\begin{itemize}
\item $f^{{\rm ext}}$ are external bulk forces (expending power on displacement),
\item $t^{{\rm ext}}$ are external surface forces (expending power on displacement),
\item $M^{{\rm ext}}$ are external surface double-forces (expending power
on the normal derivative of displacement) and
\item $\pi^{{\rm ext}}$ are external line forces (expending power on displacement).
\end{itemize}
Imposing that $\mathcal{P}^{{\rm int}}=-\mathcal{P}^{{\rm ext}}$
and localizing, one can get the strong form of the equations of motion
and associated boundary conditions for a second gradient continuum.

\subsection{\label{sec:Some-useful-relationships-2}Some useful relationships
between the third order hyperstress tensor and the second order
couple stress tensor for the  indeterminate couple stress model }

Let us consider the third order couple stress tensor
\begin{align}
\widetilde{\mathfrak{m}}  =&\, \alpha_{1}\nabla\left[\mathrm{skew}\left(\nabla u\right)\right]+\frac{1}{4}\left(\alpha_{1}-\alpha_{2}\right)\left[\nabla\left(\mathrm{Div}\, u\right)\otimes\id-\mathrm{Div}\left(\text{\ensuremath{\nabla}}u\right)\otimes\id\right]
\label{eq:DoubleStressCompact-1}\\
 & +\frac{1}{4}\left(\alpha_{1}-\alpha_{2}\right)\left[\left(\mathrm{Div}\left(\text{\ensuremath{\nabla}}u\right)\otimes\id\right)^{T_{12}}-\left(\nabla\left(\mathrm{Div}\, u\right)\otimes\id\right)^{T_{12}}\right],\nonumber
\end{align}
together with the second order couple stress tensor
\begin{align}
\widetilde{m} & =\frac{\alpha_{1}+\alpha_{2}}{2}\ \nabla\curl u+\frac{\alpha_{1}-\alpha_{2}}{2}\ \nabla(\curl u)^{T}={\alpha_{1}}\dev\sym(\nabla\curl u)+{\alpha_{2}}\,\skw(\nabla\curl u).\label{eq:mMindlin-1}
\end{align}
It can be checked that, suitably deriving eq. \eqref{eq:DoubleStress}
one gets
\begin{align}
(\mathrm{Div}\,\widetilde{\mathfrak{m}})_{ij}=\,\widetilde{\mathfrak{m}}_{ijk,k} & =\frac{\alpha_{1}}{2}\left(u_{i,jkk}-u_{j,ikk}\right)+\frac{1}{4}\left(\alpha_{1}-\alpha_{2}\right)\left[u_{p,ipj}-u_{i,ppj}+u_{j,ppi}-u_{p,jpi}\right]=\frac{\alpha_{1}+\alpha_{2}}{4}\left(u_{i,j}-u_{j,i}\right)_{,kk},\label{eq:Int0}
\end{align}
or equivalently in compact form
\begin{align}
\mathrm{Div}\,\widetilde{\mathfrak{m}}= & \frac{\alpha_{1}+\alpha_{2}}{2}\Delta(\skew\nabla u).\label{eq:Int1}
\end{align}

On the other hand, one has that
\begin{align*}
\frac{1}{2}\left[\mathrm{anti}\:\mathrm{Div}\,\widetilde{m}\right]_{ij}  =&\,-\frac{1}{2}\epsilon_{ijp}\widetilde{m}_{pm,m}=-\frac{1}{2}\epsilon_{ijp}\left(\frac{\alpha_{1}+\alpha_{2}}{2}\ \epsilon_{pqk}u_{k,qmm}+\frac{\alpha_{1}-\alpha_{2}}{2}\ \epsilon_{mqk}u_{k,qpm}\right)
\\
  =&\,\frac{\alpha_{1}+\alpha_{2}}{4}(\delta_{qj}\delta_{ik}-\delta_{qj}\delta_{jk})\: u_{k,qmm}\\
 & +\frac{\alpha_{1}-\alpha_{2}}{4}\left(-\delta_{mi}\delta_{qj}\delta_{kp}+\delta_{mi}\delta_{qp}\delta_{kj}+\delta_{mj}\delta_{qi}\delta_{kp}-\delta_{mj}\delta_{qp}\delta_{ki}-\delta_{mp}\delta_{qi}\delta_{kj}+\delta_{mp}\delta_{qj}\delta_{ki}\right)\: u_{k,qpm}
\\
  =&\,\frac{\alpha_{1}+\alpha_{2}}{4}\left(u_{i,jmm}-u_{j,imm}\right)+\frac{\alpha_{1}-\alpha_{2}}{4}\left(-u_{p,jpi}+u_{j,ppi}+u_{p,ipj}-u_{i,ppj}-u_{j,ipp}+u_{i,jpp}\right)\\
  =&\,\frac{\alpha_{1}+\alpha_{2}}{4}\left(u_{i,j}-u_{j,i}\right)_{mm}.
\end{align*}
This relationship, together with \eqref{eq:Int0} and \eqref{eq:Int1}
implies the relationship \eqref{eq:IdentBulk-1}.

As for proving the equalities \eqref{eq:IdentBulk-1}-\eqref{eq:Ident3-1-1},
we start remarking that
\begin{align}
\left(\widetilde{\mathfrak{m}}\cdot n\right)_{ij}=&\,\frac{\alpha_{1}}{2}\left(u_{i,jk}-u_{j,ik}\right)n_{k}+\frac{1}{4}\left(\alpha_{1}-\alpha_{2}\right)\left[u_{p,ip}\delta_{jk}-u_{i,pp}\delta_{jk}+u_{j,pp}\delta_{ik}-u_{p,jp}\delta_{ik}\right]n_{k}\nonumber \label{eq:IntermSec}\\
=&\,\frac{\alpha_{1}}{2}\left(u_{i,j}-u_{j,i}\right)_{,k}n_{k}+\frac{1}{4}\left(\alpha_{1}-\alpha_{2}\right)\left[u_{p,ip}n_{j}-u_{i,pp}n_{j}+u_{j,pp}n_{i}-u_{p,jp}n_{i}\right],\nonumber
\end{align}
which in compact form equivalently reads
\begin{align*}
\widetilde{\mathfrak{m}}\cdot n= & \,\alpha_{1}\left[\nabla\left(\mathrm{skew}\,\nabla u\right)\right]\cdot n+\frac{1}{4}\left(\alpha_{1}-\alpha_{2}\right)\left[\nabla\left(\mathrm{Div}\, u\right)\otimes n-\mathrm{Div}\left(\text{\ensuremath{\nabla}}u\right)\otimes n\right]
\\
= & -\frac{1}{4}\left(\alpha_{1}-\alpha_{2}\right)\left[\left(\nabla\left(\mathrm{Div}\, u\right)\otimes n\right)-\left(\mathrm{Div}\left(\text{\ensuremath{\nabla}}u\right)\otimes n\right)\right]^{T}\\
= &\, \alpha_{1}\left[\nabla\left(\mathrm{skew}\,\nabla u\right)\right]\cdot n+\frac{\left(\alpha_{1}-\alpha_{2}\right)}{2}\:\mathrm{skew}\left[\nabla\left(\mathrm{Div}\, u\right)\otimes n-\mathrm{Div}\left(\text{\ensuremath{\nabla}}u\right)\otimes n\right].
\end{align*}
On the other hand, using the notations introduced in \eqref{eq:AxlAnti_indices},
one has that
\begin{align*}
\left[\frac{1}{2}\mathrm{anti}\left(\widetilde{m}\cdot n\right)\right]_{ij}=&\,-\frac{1}{2}\epsilon_{ijk}\widetilde{m}_{km}n_{m}=-\frac{1}{2}\epsilon_{ijk}\left(\frac{\alpha_{1}+\alpha_{2}}{2}\ \epsilon_{kpq}u_{q,pm}+\frac{\alpha_{1}-\alpha_{2}}{2}\ \epsilon_{mpq}u_{q,pk}\right)n_{m},
\\
=&\,-\frac{\alpha_{1}+\alpha_{2}}{4}(\delta_{ip}\delta_{jq}-\delta_{iq}\delta_{jp})u_{q,pm}n_{m}\\
&+\frac{\alpha_{1}-\alpha_{2}}{4}\left(-\delta_{mi}\delta_{jp}\delta_{kq}+\delta_{mi}\delta_{qj}\delta_{kp}+\delta_{ip}\delta_{jm}\delta_{kq}-\delta_{ip}\delta_{qj}\delta_{km}-\delta_{kp}\delta_{qi}\delta_{mj}+\delta_{mk}\delta_{pj}\delta_{qi}\right)u_{q,pk}n_{m}\\
=&\,\frac{\alpha_{1}+\alpha_{2}}{4}(u_{i,j}-u_{j,i})_{,m}n_{m}+\frac{\alpha_{1}-\alpha_{2}}{4}\left(u_{i,jm}n_{m}-u_{j,im}n_{m}-u_{q,jq}n_{i}+u_{j,pp}n_{i}+u_{q,iq}n_{j}-u_{i,pp}n_{j}\right)\\
=&\,\frac{\alpha_{1}}{2}(u_{i,j}-u_{j,i})_{,m}n_{m}+\frac{\alpha_{1}-\alpha_{2}}{4}\left(u_{q,iq}n_{j}-u_{i,pp}n_{j}u_{j,pp}n_{i}-u_{q,jq}n_{i}\right).
\end{align*}
This relation, when compared to \eqref{eq:IntermSec} finally allows
to prove that
\begin{align*}
\widetilde{\mathfrak{m}}\cdot n=\frac{1}{2}\mathrm{anti}\left(\widetilde{m}\cdot n\right)= & \alpha_{1}\left[\nabla\left(\mathrm{skew}\,\nabla u\right)\right]\cdot n+\frac{\left(\alpha_{1}-\alpha_{2}\right)}{2}\:\mathrm{skew}\left[\nabla\left(\mathrm{Div}\, u\right)\otimes n-\mathrm{Div}\left(\text{\ensuremath{\nabla}}u\right)\otimes n\right].
\end{align*}
It is clear that, since also in this case $\widetilde{\mathfrak{m}}\cdot n=\frac{1}{2}\mathrm{anti}\left(\widetilde{m}\cdot n\right)$,
then \eqref{eq:IdentBulk-1}-\eqref{eq:Ident3-1-1} are straightforwardly
verified.

\subsection{Some alternative  calculations useful to rewrite  the governing equations
and boundary conditions in a form which is directly comparable to
Mindlin's one}\label{sec:ComparisonToMindlin}

\textcolor{black}{\small{}{}In this subsection we just report some
alternative calculations to obtain the same results as before, so
that they are not uncontournable to the understanding of the main
results of the paper. Indeed, the result of the first part of this
sections can also be re-obtained remarking that
\begin{align}
\nabla[(\anti[\widetilde{m}\cdot n])\cdot T] -\nabla[(\anti[(\id-n\otimes n)\cdot\widetilde{m}\cdot n])\cdot T]&=\nabla[(\anti[(n\otimes n)\cdot\widetilde{m}\cdot n])\cdot T]\notag\\
 =\nabla[(\anti[n\,\langle n,(\sym\widetilde{m})\cdot n\rangle])\cdot T]&=\nabla[\,\langle n,(\sym\widetilde{m})\cdot n\rangle\anti(n)\cdot T\,]
\end{align}
and
\begin{align}
\anti(n)\,(n\otimes n) & =(\anti(n)\, n)\otimes n=(n\times n)\otimes n=0,\notag\label{normalF3}\\
(n\otimes n)\,\anti(n) & =-n\otimes(\anti(n)\, n)=-n\otimes(n\times n)=0.
\end{align}
and also that
\begin{align}
\big(\nabla[\langle n, & (\sym\widetilde{m})\cdot n\rangle\anti(n)\cdot T\,]:T)_{i}=[\langle n,(\sym\widetilde{m})\cdot n\rangle\anti(n)\,(\id-n\otimes n)\,]_{ij,k}\, T_{jk}\notag\\
= & [\langle n,(\sym\widetilde{m})\cdot n\rangle\,\anti(n)]_{ij,k}\, T_{jk}-[\langle n,(\sym\widetilde{m})\cdot n\rangle\anti(n)\,(n\otimes n)\,]_{ij,k}\, T_{jk}\notag\\
= & [\langle n,(\sym\widetilde{m})\cdot n\rangle\,\anti(n)]_{ij,k}\, T_{jk}\\
= & [\langle n,(\sym\widetilde{m})\cdot n\rangle]_{,k}\,[\anti(n)]_{ij}\, T_{jk}+\langle n,(\sym\widetilde{m})\cdot n\rangle\,[\anti(n)]_{ij,k}\, T_{jk}\notag\\
= & [\langle n,(\sym\widetilde{m})\cdot n\rangle]_{,k}\,[\anti(n)]_{ij}\,\delta_{jk}-[\langle n,(\sym\widetilde{m})\cdot n\rangle]_{,k}\,[\anti(n)]_{ij}\, n_{j}n_{k}+\langle n,(\sym\widetilde{m})\cdot n\rangle\,[\anti(n)]_{ij,k}\, T_{jk}\notag\\
= & [\langle n,(\sym\widetilde{m})\cdot n\rangle]_{,k}\,[\anti(n)]_{ik}-[\langle n,(\sym\widetilde{m})\cdot n\rangle]_{,k}\,[\anti(n)]_{ij}\, n_{j}n_{k}+\langle n,(\sym\widetilde{m})\cdot n\rangle\,[\anti(n)]_{ij,k}\, T_{jk}\notag\\
= & -\{n\times\nabla[\langle n,(\sym\widetilde{m})\cdot n\rangle]\}_{i}-[\langle n,(\sym\widetilde{m})\cdot n\rangle]_{,k}\,[\anti(n)]_{ij}\,(n\otimes n)_{jk}+\langle n,(\sym\widetilde{m})\cdot n\rangle\,[\anti(n)]_{ij,k}\, T_{jk}\notag\\
= & -\{n\times\nabla[\langle n,(\sym\widetilde{m})\cdot n\rangle]\}_{i}-[\langle n,(\sym\widetilde{m})\cdot n\rangle]_{,k}\,[\anti(n)\,(n\otimes n)]_{ik}+\langle n,(\sym\widetilde{m})\cdot n\rangle\,[\anti(n)]_{ij,k}\, T_{jk}\notag\\
= & -\{n\times\nabla[\langle n,(\sym\widetilde{m})\cdot n\rangle]\}_{i}+\langle n,(\sym\widetilde{m})\cdot n\rangle\,[\anti(n)]_{ij,k}\, T_{jk}.\notag
\end{align}
}{\small \par}

\subsection{\label{sectionincomplete2}The missing steps in Mindlin and Tiersten's
classical approach}
In this section, we present once again the reasoning followed by Mindlin and Tiersten to obtain their set of bulk equations and boundary conditions trying to highlight the points in which their approach had to be further developed.

We start our analysis, remarking that the quantity $\langle\widetilde{m}\cdot n,(\id-n\otimes n)\cdot[\axl(\skw\nabla\delta u)]\rangle$
still does contain contributions performing work against $\delta u$
alone (even though there is a projection $\id-n\otimes n$ involved),
which can be assigned arbitrarily and are therefore somehow related
to independent variation $\delta u$. This case is not similar with
the Cosserat theory in which we assume a priori that displacement
$u$ and microrotation $\overline{A}\in\so(3)$ are independent kinematical
degrees of freedom. On the other hand, the indeterminate couple stress
model is not simply obtained as a constraint Cosserat model \cite{Sokolowski72},
i.e. assuming that $\overline{A}=\axl\skew\nabla u$. In the indeterminate
couple stress model the only independent kinematical degree of freedom
is $u$. We believe that the indeterminate couple stress model constructed
as a constraint Cosserat model represents only an approximation of
the indeterminate couple stress model, in the sense that the boundary
conditions are not correctly and completely considered.

Indeed, using the projectors  ($T=\id-n\otimes n$
and $Q=n\otimes n$) we obtain
\begin{align}
\langle[(\id-n\otimes n)\cdot\widetilde{m}  \cdot n],\curl\delta u\rangle&=\langle\anti[(\id-n\otimes n)\cdot\widetilde{m}\cdot n]\cdot\id,\nabla\delta u\rangle=\langle\anti[(\id-n\otimes n)\cdot\widetilde{m}\cdot n]\cdot(T+Q),\nabla\delta u\rangle\notag\\
 & =\langle\anti[(\id-n\otimes n)\cdot\widetilde{m}\cdot n]\cdot T,\nabla\delta u\rangle+\langle\anti[(\id-n\otimes n)\cdot\widetilde{m}\cdot n]\cdot Q,\nabla\delta u\rangle\notag\\
 & =\langle\anti[(\id-n\otimes n)\cdot\widetilde{m}\cdot n]\cdot T\cdot T,\nabla\delta u\rangle+\langle\anti[(\id-n\otimes n)\cdot\widetilde{m}\cdot n]\nabla\delta u\cdot Q\rangle \,da\\
 & =\langle T,T\cdot\anti[(\id-n\otimes n)\cdot\widetilde{m}\cdot n]^{T}\cdot\nabla\delta u\rangle+\langle\anti[(\id-n\otimes n)\cdot\widetilde{m}\cdot n],\nabla\delta u\cdot Q\rangle\notag\\
 & =-\langle T,T\cdot\anti[(\id-n\otimes n)\cdot\widetilde{m}\cdot n]\cdot\nabla\delta u\rangle+\langle\anti[(\id-n\otimes n)\cdot\widetilde{m}\cdot n],\nabla\delta u\cdot Q\rangle.\notag
\end{align}
On the other hand we have
\begin{align}
\langle\anti[(\id-n\otimes n)\cdot\widetilde{m}\cdot n],\nabla\delta u\cdot Q\rangle= & \,\langle\anti[(\id-n\otimes n)\cdot\widetilde{m}\cdot n],\nabla\delta u\cdot(n\otimes n)\rangle\notag\\=&\langle\anti[(\id-n\otimes n)\cdot\widetilde{m}\cdot n]\cdot n,\nabla\delta u\cdot n\rangle\notag\\
= & \,\langle(\id-n\otimes n)\cdot\anti[(\id-n\otimes n)\cdot\widetilde{m}\cdot n]\cdot n,\nabla\delta u\cdot n\rangle\\
 & +\langle(n\otimes n)\cdot\anti[(\id-n\otimes n)\cdot\widetilde{m}\cdot n]\cdot n,\nabla\delta u\cdot n\rangle\notag\\
= &\, \langle(\id-n\otimes n)\cdot\anti[(\id-n\otimes n)\cdot\widetilde{m}\cdot n]\cdot n,\nabla\delta u\cdot n\rangle\notag\\
 & -\langle n\otimes[\anti[(\id-n\otimes n)\cdot\widetilde{m}\cdot n]\cdot n]\cdot n,\nabla\delta u\cdot n\rangle\notag\notag\\
= &\, \langle(\id-n\otimes n)\cdot\anti[(\id-n\otimes n)\cdot\widetilde{m}\cdot n]\cdot n,\nabla\delta u\cdot n\rangle\notag\\
 & -\langle n\langle[\anti(\widetilde{m}\cdot n)\cdot n],n\rangle,\nabla\delta u\cdot n\rangle\notag\notag\\
= & \,\langle(\id-n\otimes n)\cdot\anti[(\id-n\otimes n)\cdot\widetilde{m}\cdot n]\cdot n,\nabla\delta u\cdot n\rangle\notag\\
 & -\langle\anti[(\id-n\otimes n)\cdot\widetilde{m}\cdot n]\cdot n,n\rangle\,\langle n,\nabla\delta u\cdot n\rangle\notag
\end{align}
and, moreover,
\begin{align}
\langle\underbrace{\anti[(\id-n\otimes n)\cdot\widetilde{m}\cdot n]}_{\in\so(3)}\cdot n,n\rangle=0.
\end{align}
Hence
\begin{align}
\langle\anti[(\id-n\otimes n)\cdot\widetilde{m}\cdot n],\nabla\delta u\cdot Q\rangle & =\langle(\id-n\otimes n)\cdot\anti[(\id-n\otimes n)\cdot\widetilde{m}\cdot n]\cdot n,\nabla\delta u\cdot n\rangle.
\end{align}
Thus, we deduce
\begin{align}
\langle\widetilde{m}\cdot n,\underbrace{(\id-n\otimes n)\cdot\curl\delta u}_{\begin{array}{c}
\text{\tiny{not completely independent}}\vspace{-1.5mm}\\
\text{\tiny{second order variation}}
\end{array}}\rangle= & -\langle\, T,T\cdot(\anti[(\id-n\otimes n)\cdot\widetilde{m}\cdot n])\cdot\nabla\delta u\rangle\\
 & +\langle(\id-n\otimes n)\cdot\anti[(\id-n\otimes n)\cdot\widetilde{m}\cdot n]\cdot n,\hspace{-1cm}\underbrace{\nabla\delta u\cdot n}_{\begin{array}{c}
\text{\tiny{completely} $\delta u$ } \text{\tiny{-independent second}}\vspace{-1mm}\\
\text{\tiny{order normal variation of gradient}}
\end{array}}\hspace{-1cm}\rangle.\notag
\end{align}
Since
\begin{align}
\{\nabla[ & \underbrace{T\cdot(\anti[(\id-n\otimes n)\cdot\widetilde{m}\cdot n])\cdot\delta u}_{\in\mathbb{R}^{3}}]\}_{ik}=\{T\cdot(\anti[(\id-n\otimes n)\cdot\widetilde{m}\cdot n])\cdot\delta u\}_{i,k}\notag\\
 & =\{\big(T\cdot(\anti[(\id-n\otimes n)\cdot\widetilde{m}\cdot n])\big)_{ij}(\delta u)_{j}\}_{,k}\notag\\
 & =\big(T\cdot(\anti[(\id-n\otimes n)\cdot\widetilde{m}\cdot n])\big)_{ij,k}(\delta u)_{j}+\big(T\cdot(\anti[(\id-n\otimes n)\cdot\widetilde{m}\cdot n])\big)_{ij}(\delta u)_{j,k}\\
 & =\big(T\cdot(\anti[(\id-n\otimes n)\cdot\widetilde{m}\cdot n])\big)_{ij,k}(\delta u)_{j}+\big(T\cdot(\anti[(\id-n\otimes n)\cdot\widetilde{m}\cdot n])\big)_{ij}(\nabla\delta u)_{jk}\notag\\
 & =\big(T\cdot(\anti[(\id-n\otimes n)\cdot\widetilde{m}\cdot n])\big)_{ij,k}(\delta u)_{j}+\{T\cdot(\anti[(\id-n\otimes n)\cdot\widetilde{m}\cdot n])\cdot\nabla\delta u\}_{ik}\notag
\end{align}
using the surface divergence theorem%
\footnote{The surface divergence theorem in this context reads $\int_{\partial\Omega}\langle\nabla(T\, v),T\rangle \,da=\int_{\partial\Gamma}\jump{\langle v,\nu\rangle}\,ds,$ see (\ref{eq:Surface_Div_Th}).%
} we obtain
\begin{align}
\,\int_{\partial\Omega}\langle\widetilde{m}\cdot n,(\id-n\otimes n)\cdot\curl\delta u\rangle \,da= & -\int_{\partial\Omega}\langle\, T,\nabla[T\cdot(\anti[(\id-n\otimes n)\cdot\widetilde{m}\cdot n])\cdot\delta u]\rangle \,da\notag\\
 & +\int_{\partial\Omega}\, T_{ik}\,\big(T\cdot(\anti[(\id-n\otimes n)\cdot\widetilde{m}\cdot n])\big)_{ij,k}(\delta u)_{j}\,da\notag\\
 & +\int_{\partial\Omega}\langle(\id-n\otimes n)\cdot\anti[(\id-n\otimes n)\cdot\widetilde{m}\cdot n]\cdot n,\nabla\delta u\cdot n\rangle \,da\notag\\
\overset{\tiny{\begin{array}{cc}
\text{surface}\\
\text{divergence}
\end{array}}}{=} & -\int_{\partial{\Gamma}}\jump{\langle\anti[(\id-n\otimes n)\cdot\widetilde{m}\cdot n])\delta u\,,\nu\rangle}\,ds\\
 & +\int_{\partial\Omega}\big(T\cdot(\anti[(\id-n\otimes n)\cdot\widetilde{m}\cdot n])\big)_{ij,k}\, T_{ik}\,(\delta u)_{j}\,da\notag\\
 & +\int_{\partial\Omega}\langle(\id-n\otimes n)\cdot\anti[(\id-n\otimes n)\cdot\widetilde{m}\cdot n]\cdot n,\nabla\delta u\cdot n\rangle \,da\notag\\
= & \int_{\partial\Omega}\langle(\id-n\otimes n)\cdot\anti[(\id-n\otimes n)\cdot\widetilde{m}\cdot n]\cdot n,\nabla\delta u\cdot n\rangle \,da\notag\\
 & +\int_{\partial\Omega}\big(T\cdot(\anti[(\id-n\otimes n)\cdot\widetilde{m}\cdot n])\big)_{ij,k}\, T_{ik}\,(\delta u)_{j}\,da\notag\\
 & +\int_{\partial{\Gamma}}\langle\jump{\anti[(\id-n\otimes n)\cdot\widetilde{m}\cdot n]\cdot\nu},\delta u\rangle \,ds.\notag
\end{align}
On the other hand, we have
\begin{align}
\big(T\cdot(\anti[(\id-n\otimes n)\cdot\widetilde{m}\cdot n])\big)_{ij} & =T_{il}(\anti[(\id-n\otimes n)\cdot\widetilde{m}\cdot n])_{lj}\notag\\
 & =-T_{il}(\anti[(\id-n\otimes n)\cdot\widetilde{m}\cdot n])_{jl}=-T_{li}(\anti[(\id-n\otimes n)\cdot\widetilde{m}\cdot n])_{jl}\\
 & =-(\anti[(\id-n\otimes n)\cdot\widetilde{m}\cdot n])_{jl}T_{li}=-\{(\anti[(\id-n\otimes n)\cdot\widetilde{m}\cdot n])\cdot T\}_{ji}.\notag
\end{align}
Hence, we deduce
\begin{align}
\big(T\cdot(\anti[(\id-n\otimes n)\cdot\widetilde{m}\cdot n])\big)_{ij,k}\, T_{ik} & =-\{(\anti[(\id-n\otimes n)\cdot\widetilde{m}\cdot n])\cdot T\}_{ji,k}\, T_{ik}\\
 & =-\big(\nabla[(\anti[(\id-n\otimes n)\cdot\widetilde{m}\cdot n])\cdot T]:T\big)_{j}.\notag
\end{align}

In view of the above computations, we deduce
\begin{align}
- & \int_{\partial\Omega}\langle(\sigma-\widetilde{\tau})\cdot n,\delta u\rangle\, da-\int_{\partial\Omega}\langle\widetilde{m}\cdot n,\axl(\skw\nabla\delta u)\rangle\, da\notag\label{axldeltan1}\\
= & -\int_{\partial\Omega}\langle\underbrace{(\sigma-\widetilde{\tau})\cdot n-\frac{1}{2}n\times\nabla[\langle n,(\sym\widetilde{m})\cdot n\rangle]}_{\text{classical first boundary term of Mindlin, Gao, Yang, etc.}},\delta u\rangle\, da\notag-\int_{\partial\Omega}\langle\underbrace{\widetilde{m}\cdot n,(\id-n\otimes n)\cdot[\axl(\skw\nabla\delta u)]\rangle}_{\begin{array}{c}
\text{\tiny{classical second boundary term of}}\vspace{-1.5mm}\\
\text{\tiny{Mindlin, Gao, Yang, etc.}}
\end{array}}\, da\notag\\
 & \qquad\qquad\text{ ...must be split further to obtain strongly independent  variations...}\\
= & -\int_{\partial\Omega}\langle(\sigma-\frac{1}{2}\anti{\rm Div}[\widetilde{m}])\cdot n-\frac{1}{2}n\times\nabla[\langle n,(\sym\widetilde{m})\cdot n\rangle]-\frac{1}{2}\nabla[(\anti[(\id-n\otimes n)\cdot\widetilde{m}\cdot n])\cdot T]:T,\delta u\rangle\, da\notag\\
 & -\frac{1}{2}\int_{\partial\Omega}\underbrace{\langle(\id-n\otimes n)\cdot\anti[(\id-n\otimes n)\cdot\widetilde{m}\cdot n]\cdot n,\nabla\delta u\cdot n\rangle}_{\begin{array}{c}
\text{\tiny{completely} $\delta u$} \text{\tiny{-independent second}}\vspace{-1mm}\\
\text{\tiny{order normal variation of gradient}}
\end{array}}\,da\notag\\
 & -\frac{1}{2}\int_{\partial{\Gamma}}\langle\jump{\anti[(\id-n\otimes n)\cdot\widetilde{m}\cdot n]\cdot\nu},\delta u\rangle \,ds
\end{align}
for all variations $\delta u\in C^{\infty}(\Omega)$. Note that $\partial(\partial\Omega\setminus\overline{\Gamma})=\partial\Gamma$.
Hence, there are indeed two terms
\[
(\sigma-\frac{1}{2}\anti{\rm Div}[\widetilde{m}])\cdot n-\frac{1}{2}n\times\nabla[\langle n,(\sym\widetilde{m})\cdot n\rangle]-\frac{1}{2}\nabla[(\anti[(\id-n\otimes n)\cdot\widetilde{m}\cdot n])\cdot T]:T
\]
and
\[
\jump{\anti[(\id-n\otimes n)\cdot\widetilde{m}\cdot n]\cdot\nu}
\]
which perform work against $\delta u$, while only the term
\[
(\id-n\otimes n)\cdot\anti[(\id-n\otimes n)\cdot\widetilde{m}\cdot n]\cdot n
\]
is related solely to the independent second order normal variation
of the gradient $\nabla\delta u\cdot n$. This split of the boundary
condition is not the one as obtained e.g. by Gao and Park \cite{park2008variational}
among others and seems to be entirely new in the context of the couple
stress model.

\subsection{The proof of Lemma \ref{lemmadesprecurlaxl}}\label{appendixlemmadesprecurlaxl}

In other words, this lemma implies that, for any tangential vector field $\tau_i$ on $\Gamma\subset \partial \Omega$
\begin{align}
\left.\begin{array}{rcl}
u\big|_{\Gamma} & = &  u^{\rm ext},\vspace{1.5mm}\\
\langle{\rm curl}\, u,\tau_{i}\rangle\big|_{\Gamma} & = & \text{known}
\end{array}\right\} \quad\Leftrightarrow\quad\left\{ \begin{array}{rcl}
u\big|_{\Gamma} & = &  u^{\rm ext},\vspace{1.5mm}\\
\langle\nabla u\cdot n,\tau_{i}\rangle\big|_{\Gamma} & = & \text{known}\,,
\end{array}\right..
\end{align}

Since $u\in C^{\infty}(\Omega)$, on the one hand we obtain
\begin{align}
\left.\begin{array}{rcl}
u\big|_{\Gamma} & = &  u^{\rm ext}\quad\overset{\text{(Stokes)}}{\Rightarrow}\quad\text{also}\quad\langle{\rm curl}\, u,n\rangle\big|_{\Gamma}=\text{known}),\vspace{1.5mm}\\
\langle{\rm curl}\, u,\tau_{i}\rangle\big|_{\Gamma} & = & \text{known}
\end{array}\right\} \quad & \Leftrightarrow\quad\left\{ \begin{array}{rcl}
u\big|_{\Gamma} & = &  u^{\rm ext},\vspace{1.5mm}\\
{\rm curl}\, u\,\big|_{\Gamma} & = & \text{known}
\end{array}\right.\\
\quad & \Leftrightarrow\quad\left\{ \begin{array}{rcl}
u\big|_{\Gamma} & = &  u^{\rm ext},\vspace{1.5mm}\\
{\rm skew}\,\nabla u\,\big|_{\Gamma} & = & \text{known}.
\end{array}\right.\notag
\end{align}
On the other hand we deduce
\begin{align}
& \left.\begin{array}{rcl}
u\big|_{\Gamma} & = &  u^{\rm ext}\quad({\Rightarrow}\quad\text{also}\quad\nabla\, u.\,\tau_{i}\big|_{\Gamma}=\text{known})\quad{\Leftrightarrow}\quad\left\{ \begin{array}{rcl}
u\big|_{\Gamma} & = &  u^{\rm ext},\vspace{1.5mm}\\
\langle\nabla\, u.\,\tau_{i},\tau_{i}\rangle\big|_{\Gamma} & = & \text{known}\vspace{1.5mm}\\
\langle\tau_{i},(\nabla\, u)^{T}.\,\tau_{i}\rangle\big|_{\Gamma} & = & \text{known},\vspace{1.5mm}\\
\langle\nabla\, u.\,\tau_{i},n\rangle\big|_{\Gamma} & = & \text{known},\vspace{1.5mm}\\
\langle\tau_{i},(\nabla\, u)^{T}\cdot n\rangle\big|_{\Gamma} & = & \text{known}
\end{array}\right.\vspace{1.5mm}\\
\langle\nabla u\cdot n,\tau_{i}\rangle\big|_{\Gamma} & = & \langle b_{0},\tau_{i}\rangle\quad{\Leftrightarrow}\quad\langle n,(\nabla u)^{T}\cdot\tau_{i}\rangle\big|_{\Gamma}=\langle b_{0},\tau_{i}\rangle
\end{array}\right\} \vspace{2mm}\\
& \Leftrightarrow\quad\left\{ \begin{array}{rcl}
u\big|_{\Gamma} & = &  u^{\rm ext},\vspace{1.5mm}\\
(\nabla u)^{T}\cdot\tau_{i}\big|_{\Gamma} & = & \text{known},\vspace{1.5mm}\\
(\nabla u)\cdot\tau_{i}\big|_{\Gamma} & = & \text{known}
\end{array}\right.\Leftrightarrow\qquad\left\{ \begin{array}{rcl}
u\big|_{\Gamma} & = &  u^{\rm ext},\vspace{1.5mm}\\
(\skew\nabla u)\cdot\tau_{i}\big|_{\Gamma} & = & \text{known},\vspace{1.5mm}\\
(\sym\nabla u)\cdot\tau_{i}\big|_{\Gamma} & = & \text{known}.
\end{array}\right.\notag
\end{align}
Until now, we have obtained
\begin{align}
& \left.\begin{array}{rcl}
u\big|_{\Gamma} & = &  u^{\rm ext},\vspace{1.5mm}\\
(\id-n\otimes n)\cdot{\rm curl}\, u\big|_{\Gamma} & = & \text{known}
\end{array}\right\} \qquad\Leftrightarrow\qquad\left\{ \begin{array}{rcl}
u\big|_{\Gamma} & = &  u^{\rm ext},\vspace{1.5mm}\\
\nabla u\cdot\tau_{i}\big|_{\Gamma} & = & \text{known},\vspace{1.5mm}\\
{\rm skew}\,\nabla u\,\big|_{\Gamma} & = & \text{known}
\end{array}\right.\notag\label{curleqn1}\\
& \Leftrightarrow\qquad\left\{ \begin{array}{rcl}
u\big|_{\Gamma} & = &  u^{\rm ext},\vspace{1.5mm}\\
\nabla u\cdot\tau_{i}\big|_{\Gamma} & = & \text{known},\vspace{1.5mm}\\
({\rm skew}\,\nabla u)\cdot\tau_{i}\big|_{\Gamma} & = & \text{known},\vspace{1.5mm}\\
({\rm skew}\,\nabla u)\cdot n\big|_{\Gamma} & = & \text{known}
\end{array}\right.\qquad\Leftrightarrow\qquad\left\{ \begin{array}{rcl}
u\big|_{\Gamma} & = &  u^{\rm ext},\vspace{1.5mm}\\
\nabla u\cdot\tau_{i}\big|_{\Gamma} & = & \text{known},\vspace{1.5mm}\\
({\rm sym}\,\nabla u)\cdot\tau_{i}\big|_{\Gamma} & = & \text{known},\vspace{1.5mm}\\
({\rm skew}\,\nabla u)\cdot\tau_{i}\big|_{\Gamma} & = & \text{known},\vspace{1.5mm}\\
({\rm skew}\,\nabla u)\cdot n\big|_{\Gamma} & = & \text{known},
\end{array}\right.
\end{align}
while
\begin{align}
\left.\begin{array}{rcl}
u\big|_{\Gamma} & = &  u^{\rm ext},\vspace{1.5mm}\\
(\id-n\otimes n)\cdot\nabla u\cdot n\big|_{\Gamma} & = & \text{known}
\end{array}\right\} \qquad\Leftrightarrow\qquad\left\{ \begin{array}{rcl}
u\big|_{\partial\Omega} & = &  u^{\rm ext},\vspace{1.5mm}\\
\nabla u\cdot\tau_{i}\big|_{\Gamma} & = & \text{known},\vspace{1.5mm}\\
\dd(\skew\nabla u)\cdot\tau_{i}\big|_{\Gamma} & = & \text{known},\vspace{1.5mm}\\
\dd(\sym\nabla u)\cdot\tau_{i}\big|_{\Gamma} & = & \text{known}.
\end{array}\right.\label{curleqn2}
\end{align}
The only one boundary condition which appears in the last term of
\eqref{curleqn1} but does not appear in the last term of \eqref{curleqn2}
is $({\rm skew}\,\nabla u)\cdot n\big|_{\Gamma}=\text{known}$. However,
$\dd(\skew\nabla u)\cdot\tau_{i}\big|_{\Gamma}=\text{known}$ implies
also $({\rm skew}\,\nabla u)\cdot n\big|_{\Gamma}=\text{known}$, since
\begin{align}
\left.\begin{array}{rcl}
\langle({\rm skew}\,\nabla u)\cdot n,n\rangle & = & 0\qquad\text{always},\vspace{1.5mm}\\
\langle({\rm skew}\,\nabla u)\cdot n,\tau_{i}\rangle & = & -\langle n,({\rm skew}\,\nabla u)\cdot\tau_{i}\rangle=\text{known}.
\end{array}\right\} \qquad\Leftrightarrow\qquad({\rm skew}\,\nabla u)\cdot n=\text{known}.
\end{align}
Therefore, by eliminating the redundant information we obtain
\begin{align}
& \left.\begin{array}{rcl}
u\big|_{\Gamma} & = &  u^{\rm ext},\vspace{1.5mm}\\
(\id-n\otimes n)\cdot{\rm curl}\, u\big|_{\Gamma} & = & \text{known}
\end{array}\right\} \qquad\Leftrightarrow\qquad\left\{ \begin{array}{rcl}
u\big|_{\Gamma} & = &  u^{\rm ext},\vspace{1.5mm}\\
({\rm skew}\,\nabla u)\cdot\tau_{i}\big|_{\Gamma} & = & \text{known},
\end{array}\right.\label{curleqn11}
\end{align}
and
\begin{align}
\left.\begin{array}{rcl}
u\big|_{\Gamma} & = &  u^{\rm ext},\vspace{1.5mm}\\
(\id-n\otimes n)\cdot\nabla u\cdot n\big|_{\Gamma} & = & \text{known}
\end{array}\right\} \qquad\Leftrightarrow\qquad\left\{ \begin{array}{rcl}
u\big|_{\Gamma} & = &  u^{\rm ext},\vspace{1.5mm}\\
\dd(\skew\nabla u)\cdot\tau_{i}\big|_{\Gamma} & = & \text{known},
\end{array}\right.\label{curleqn21}
\end{align}
and the proof is complete. 

\subsection{The proof of some identities to ulteriorly simplify the traction boundary conditions}\label{identities_tractions}

In this appendix we prove some identities which must be used in order to show the equivalence between some traction boundary conditions. 
We start by checking that the following identity holds
\begin{align}
\frac{1}{2}(\id-n\otimes n)\cdot\anti[(\id-n\otimes n)\cdot\widetilde{m}\cdot n]\cdot n=\frac{1}{2}(\id-n\otimes n)\cdot\anti(\widetilde{m}\cdot n)\cdot n. \label{propr2}
\end{align}
In other words, we have to check if
\begin{align*}
(\id-n\otimes n)\cdot\anti[(n\otimes n)\cdot\widetilde{m}\cdot n]\cdot n=0.
\end{align*}
But this fact follows immediately from
\begin{align*}
\anti[(n\otimes n)\cdot\widetilde{m}\cdot n]\cdot n=\anti[n\langle n,\widetilde{m}\cdot n\rangle]\cdot n=\langle n,\widetilde{m}\cdot n\rangle\,\anti[n]\cdot n=\langle n,\widetilde{m}\cdot n\rangle\, n\times n=\langle n,\widetilde{m}\cdot n\rangle\,0=0.
\end{align*}

The final step in order to be able to complete our comparison to Mindlin and Tiersten's boundary conditions, is to prove that
\begin{align}
\jump{\anti[\widetilde{m}\cdot n]\cdot\nu}= & \jump{\anti[(\id-n\otimes n)\cdot\widetilde{m}\cdot n]\cdot\nu} \label{propr3}
\end{align}
holds also true pointwise. To this aim, let us first remark that
\begin{align}
[\anti[(n\otimes n)\cdot\widetilde{m}\cdot n]]^{\pm}\cdot\nu^{\pm} & =[\anti[n\langle n,\widetilde{m}\cdot n\rangle]]^{\pm}\cdot\nu^{\pm}=\langle n,[\widetilde{m}\cdot n]^{\pm}\rangle[\anti(n)]^{\pm}\cdot\nu^{\pm}\notag\label{np}\\
 & =\langle n,[\widetilde{m}\cdot n]^{\pm}\rangle(n\times\nu^{\pm})=\langle n,[\widetilde{m}\cdot n]^{\pm}\rangle\,\tau^{\pm},
\end{align}
where $\tau^{\pm}=n\times\nu^{\pm}$ is tangent to curves $\partial\Gamma$,
according to the orientation on ${\partial\Omega\setminus\overline{\Gamma}}$
and $\Gamma$, respectively.

Hence, we deduce
\begin{align}
[\anti[(n\otimes n)\cdot\widetilde{m}\cdot n]]^{+}\cdot\nu^{+}+[\anti[(n\otimes n)\cdot\widetilde{m}\cdot n]]^{-}\cdot\nu^{-}=\langle n,[\widetilde{m}\cdot n]^{+}\rangle\,\tau^{+}+\langle n,[\widetilde{m}\cdot n]^{-}\rangle\,\tau^{-}.
\end{align}

Therefore, we have
\begin{align}
\jump{[\anti[\widetilde{m}\cdot n]\cdot\nu}= & \left([\anti[\widetilde{m}\cdot n]]^{+}-[\anti[\widetilde{m}\cdot n]]^{-}\right)\cdot\nu\notag\\
= & \left([\anti[(\id-n\otimes n)\cdot\widetilde{m}\cdot n]]^{+}-[\anti[(\id-n\otimes n)\cdot\widetilde{m}\cdot n]]^{-}\right)\cdot\nu\notag\\
 & -\big(\langle n,[\widetilde{m}\cdot n]^{+}\rangle\,\tau^{+}+\langle n,[\widetilde{m}\cdot n]^{-}\rangle\,\tau^{-}\big)\\
= & \,\jump{\anti[(\id-n\otimes n)\widetilde{m}\cdot n]\cdot\nu}-\big(\langle n,[\widetilde{m}\cdot n]^{+}\rangle\,\tau^{+}+\langle n,[\widetilde{m}\cdot n]^{-}\rangle\,\tau^{-}\big).\notag
\end{align}
Using Stokes theorem and the divergence theorem, we obtain
\begin{align}
\int_{\partial\Gamma} & \langle\langle n,[\widetilde{m}\cdot n]^{+}\rangle\,\tau^{+},\delta u\rangle\, ds+\int_{\partial\Gamma}\langle\langle n,[\widetilde{m}\cdot n]^{-}\rangle\,\tau^{-},\delta u\rangle\, ds\notag\\
 & =\int_{\partial\Gamma}\langle\tau^{+},\langle n,[\widetilde{m}\cdot n]^{+}\rangle\,\delta u\rangle\, ds+\int_{\partial\Gamma}\langle\tau^{-},\langle n,[\widetilde{m}\cdot n]^{-}\rangle\,\delta u\rangle\, ds\notag\\
 & =\int_{\partial\Omega\setminus\Gamma}\langle n,\curl(\langle n,\widetilde{m}\cdot n\rangle\,\delta u)\rangle\, da+\int_{\Gamma}\langle n,\curl(\langle n,\widetilde{m}\cdot n\rangle\,\delta u)\rangle\, da\\
 & =\int_{\partial\Omega}\langle n,\curl(\langle n,\widetilde{m}\cdot n\rangle\,\delta u)\rangle\, da=\int_{\partial\Omega}\dv[\curl(\langle n,\widetilde{m}\cdot n\rangle\,\delta u)]\, dv=0.\notag
\end{align}
Therefore, we have
\begin{align}
\int_{\partial\Gamma}\langle\jump{\anti[\widetilde{m}\cdot n]\cdot\nu}-\jump{\anti[(\id-n\otimes n)\cdot\widetilde{m}\cdot n]\cdot\nu},\delta u\rangle \,ds=0,
\end{align}
for all $\delta u\in C^{\infty}(\overline{\Omega})$. We choose
\begin{align}
\delta u=\jump{\anti[\widetilde{m}\cdot n]\cdot\nu}-\jump{\anti[(\id-n\otimes n)\cdot\widetilde{m}\cdot n]\cdot\nu}
\end{align}
and we obtain
\begin{align}
\int_{\partial\Gamma}\|\jump{\anti[\widetilde{m}\cdot n]\cdot\nu}-\jump{\anti[(\id-n\otimes n)\cdot\widetilde{m}\cdot n]\cdot\nu}\|^{2}\,ds=0.
\end{align}
Let us consider an arbitrary parametrization $\gamma:[a,b]\rightarrow\partial\Gamma$
of the curve $\partial\Gamma$. We obtain
\begin{align}
\int_{a}^{b}\Big( & \|\jump{\anti[\widetilde{m}\cdot n]\cdot\nu}-\jump{\anti[(\id-n\otimes n)\cdot\widetilde{m}\cdot n]\cdot\nu}\|(\gamma(s))\,|\gamma^{\prime}(s)|\Big)^{2}\,ds=0,\notag
\end{align}
which implies
\begin{align}
\|\jump{\anti[\widetilde{m}\cdot n]\cdot\nu}-\jump{\anti[(\id-n\otimes n)\cdot\widetilde{m}\cdot n]\cdot\nu}\|(\gamma(s))=0,\qquad \forall s\in[a,b].\notag
\end{align}
Hence
\begin{align}
\|\jump{\anti[\widetilde{m}\cdot n]\cdot\nu}-\jump{\anti[(\id-n\otimes n)\cdot\widetilde{m}\cdot n]\cdot\nu}\|=0\qquad\text{on}\quad\partial\Gamma\notag
\end{align}
which is equivalent to
\begin{align}
\jump{\anti[\widetilde{m}\cdot n]\cdot\nu}=\jump{\anti[(\id-n\otimes n)\cdot\widetilde{m}\cdot n]\cdot\nu}\quad\text{on}\quad\partial\Gamma.
\end{align}

\subsection{The proof of Proposition  \ref{prop:ComparisonToMindlin}}\label{appendixprop:ComparisonToMindlin}

First, we prove the following lemma:
\begin{lemma}\label{sub:vanishing} Let us consider a smooth level surface $\Sigma=\{(x_{1},x_{2},x_{3})\in \mathbb{R}^3\, |\, F(x_{1},x_{2},x_{3})=0\}$,\break  $F:\omega\subset \mathbb{R}^3\rightarrow\mathbb{R}^3$ of class $C^2$, then it holds true:
	\begin{align}
	\nabla[\anti(n)]:T=0.
	\end{align}
\end{lemma}
\begin{proof}
	Let us first remark that $\nabla[\anti(n)]:T$ is well defined, since only the tangential derivative of $\anti(n)$ are involved. The normal vector to the surface $\Sigma$
	is given by
	\begin{align}
	n=\frac{\nabla F}{\|\nabla F\|},\qquad(n_{1},n_{2},n_{3})=\frac{1}{\sqrt{F_{,l}F_{,l}}}\,(F_{,1},F_{,2},F_{,3}).
	\end{align}
	Let us remark that
	\begin{align}
	\big[\nabla[\anti(n)]:T\big]_{i} & =[\anti(n)]_{ij,k}\, T_{jk}=\left[\anti\left(\frac{\nabla F}{\|\nabla F\|}\right)\right]_{ij,k}\, T_{jk}=[\|\nabla F\|^{-1}\anti(\nabla F)]_{ij,k}\, T_{jk}\label{normalF1}\\
	& =[\|\nabla F\|^{-1}]_{,k}[\anti(\nabla F)]_{ij}\, T_{jk}+\|\nabla F\|^{-1}\,[\anti(\nabla F)]_{ij,k}\, T_{jk}.\notag
	\end{align}
	We compute
	\begin{align}
	[\|\nabla F\|^{-1}]_{,k} & =[(F_{,l}F_{,l})^{-1/2}]_{,k}=-\frac{1}{2}\,[(F_{,l}F_{,l})^{-3/2}]\,2\,(F_{,l}\, F_{,lk})=-\|\nabla F\|^{-3}\, F_{,l}\, F_{,lk}\,,\label{normalF2}\\
	{\anti(n)}_{ij,k} & =-[\epsilon_{ijl}F_{,l}]_{,k}=-\epsilon_{ijl}F_{,lk},\notag\\
	T_{jk} & =\delta_{jk}-n_{j}n_{k}=\delta_{jk}-\|\nabla F\|^{-2}F_{,j}F_{,k}\,.\notag
	\end{align}
	Moreover, using \eqref{normalF3}, we have
	\begin{align}
	[\anti(\nabla F)]_{ij}\, T_{jk} & =[\anti(\nabla F)\cdot T]_{ik}=\|\nabla F\|\,[\anti(n)\cdot T]_{ik}\notag\label{normalF4}\\
	& =\|\nabla F\|\,[\anti(n)\cdot(\id-n\otimes n)]_{ik}=\|\nabla F\|\,[\anti(n)]_{ik}=-\|\nabla F\|\,\epsilon_{ikl}n_{l}=-\epsilon_{ikl}F_{,l}.
	\end{align}
	Using \eqref{normalF2}--\eqref{normalF4} in \eqref{normalF1}, we
	obtain
	\begin{align}
	\big[\nabla[\anti(n)]:T\big]_{i} & =\left(-\|\nabla F\|^{-3}\, F_{,l}\, F_{,lk}\right)\,\left(-\epsilon_{iks}F_{,s}\right)+\|\nabla F\|^{-1}\left(-\epsilon_{ijl}F_{,lk}\right)\,\left(\delta_{jk}-\|\nabla F\|^{-2}F_{,j}F_{,k}\right)\notag\label{normalF5}\\
	& =\|\nabla F\|^{-3}\left(\epsilon_{iks}\, F_{,s}\, F_{,l}\, F_{,lk}\,+\epsilon_{ijl}F_{,lk}F_{,j}F_{,k}\right)-\|\nabla F\|^{-1}\,\epsilon_{ijl}F_{,lk}\delta_{jk}\\
	& =\|\nabla F\|^{-3}\left(\epsilon_{ikj}\, F_{,j}\, F_{,l}\, F_{,lk}\,+\epsilon_{ijl}F_{,lk}F_{,j}F_{,k}\right)-\|\nabla F\|^{-1}\,\epsilon_{ikl}F_{,lk}\notag\\
	& =\|\nabla F\|^{-3}\left(\epsilon_{ilj}\, F_{,j}\, F_{,k}\, F_{,kl}\,+\epsilon_{ijl}F_{,lk}F_{,j}F_{,k}\right)\notag\\
	& =\|\nabla F\|^{-3}\left(\epsilon_{ilj}+\epsilon_{ijl}\right)F_{,lk}F_{,j}F_{,k}=0,\notag
	\end{align}
	and the proof is complete.
\end{proof}
\noindent{We now proceed with \it the proof of Proposition \ref{prop:ComparisonToMindlin}}:
	We start remarking that,
	using the properties of the projectors $T$ and $Q$ introduced in
	\eqref{eq:Projectors_Properties} and the definition \eqref{eq:B_Def}
	of the tensor $B$ one has
	\begin{align*}
	B & =\frac{1}{2}\mathrm{anti}\left(\widetilde{m}\cdot n\right)=\frac{1}{2}\mathrm{anti}\left((T+Q)\cdot\widetilde{m}\cdot n\right)=\frac{1}{2}\mathrm{anti}\left(T\cdot\widetilde{m}\cdot n\right)+\frac{1}{2}\mathrm{anti}\left(Q\cdot\widetilde{m}\cdot n\right)\\
	& =\frac{1}{2}\mathrm{anti}\left(T\cdot\widetilde{m}\cdot n\right)+\frac{1}{2}\mathrm{anti}\left(n\otimes n\cdot\widetilde{m}\cdot n\right)\\
	& =\frac{1}{2}\mathrm{anti}\left(T\cdot\widetilde{m}\cdot n\right)+\frac{1}{2}\mathrm{anti}\left(\left.\langle  n,(\mathrm{sym}\,\widetilde{m})\cdot n\right.\rangle n\right)=:A+\frac{1}{2}\mathrm{anti}\left(\psi\,\: n\right).
	\end{align*}
	Recalling the definition \eqref{eq:AxlAnti_indices} of the $\mathrm{anti}$-operator, the second term in the traction boundary condition \eqref{eq:Boundary_CS-3}
	can hence be manipulated as follows:
	\begin{gather}
	\left[\nabla\left(B\cdot T\right):T\right]_{i}=\left(A_{ij}T_{jp}-\frac{1}{2}\epsilon_{ijk}\psi\, n_{k}T_{jp}\right)_{,h}T_{hp}=\left(A_{ij}T_{jp}\right)_{,h}T_{hp}-\frac{1}{2}\left(\epsilon_{ijk}\psi\, n_{k}T_{jp}\right)_{,h}T_{hp}.\label{eq:Calculation}
	\end{gather}
	We can now remark that
	\begin{equation}
	\left(\epsilon_{ijk}\psi\, n_{k}T_{jp}\right)_{,h}T_{hp}=\epsilon_{ijk}\left[\psi\,_{,h}T_{jh}n_{k}+\psi\, T_{jh}n_{k,h}+\psi\, T_{jp,h}T_{hp}n_{k}\right].\label{eq:Interm-1}
	\end{equation}
	On the other hand, recalling that $T+Q=\id$ and that $Q=n\otimes n$
	is a symmetric tensor, we also have
	\begin{equation}
	0=(T_{jp}+Q_{jp})_{,h}=T_{jp,h}+n_{j,h}n_{p}+n_{p,h}n_{j}\quad\Rightarrow\quad T_{jp,h}=-2n_{j,h}n_{p}.
	\end{equation}
	Using this last equality in \eqref{eq:Interm-1} and the fact that
	$T\cdot n=0$ we finally have
	\begin{align}
	\left(\epsilon_{ijk}\psi\, n_{k}T_{jp}\right)_{,h}T_{hp} & =\epsilon_{ijk}\left[\psi\,_{,h}T_{jh}n_{k}+\psi\, T_{jh}n_{k,h}-2\psi\, n_{j,h}n_{p}T_{hp}n_{k}\right]\notag\\
	& =\epsilon_{ijk}\psi\,_{,h}T_{jh}n_{k}+\epsilon_{ijk}\psi\, T_{jh}n_{k,h}-0,
	\end{align}
	and we can hence rewrite the term in \eqref{eq:Calculation} as
	\begin{align}
	\left[\nabla\left(B\cdot T\right):T\right]_{i} & =\left(A_{ij}T_{jp}-\frac{1}{2}\epsilon_{ijk}\psi\, n_{k}T_{jp}\right)_{,h}T_{hp}\notag\label{eq:Calculation-1}\\
	& =\left(A_{ij}T_{jp}\right)_{,h}T_{hp}-\frac{1}{2}\left(\epsilon_{ijk}\psi\,_{,h}T_{jh}n_{k}+\epsilon_{ijk}\psi\, T_{jh}n_{k,h}\right)\notag\\
	& =\left(A_{ij}T_{jp}\right)_{,h}T_{hp}+\frac{1}{2}\left(\epsilon_{ijk}\psi\,_{,h}(\delta_{jh}-T_{jh}-\delta_{jh})n_{k}\right)-\frac{1}{2}\left(\epsilon_{ijk}\psi\, T_{jh}n_{k,h}\right)\notag\\
	& =\left(A_{ij}T_{jp}\right)_{,h}T_{hp}+\frac{1}{2}\left(\epsilon_{ijk}\psi\,_{,h}n_{j}n_{h}n_{k}\right)-\frac{1}{2}\left(\epsilon_{ijk}\psi\,_{,j}n_{k}\right)-\frac{1}{2}\left(\epsilon_{ijk}\psi\, T_{jh}n_{k,h}\right)\\
	& =\left(A_{ij}T_{jp}\right)_{,h}T_{hp}+0-\frac{1}{2}\left(\epsilon_{ijk}\psi\,_{,j}n_{k}\right)-\frac{1}{2}\left(\epsilon_{ijk}\psi\, T_{jh}n_{k,h}\right)\notag\\
	& =\left(A_{ij}T_{jp}\right)_{,h}T_{hp}-\frac{1}{2}\left(\epsilon_{ijk}\psi\,_{,j}n_{k}\right)-\frac{1}{2}\left(\epsilon_{ijk}\psi\, T_{jh}n_{k,h}\right)\notag
	\end{align}
	or equivalently in compact form (see the definition of vector product
	in \eqref{eq:AxlAnti_indices}):
	\begin{align}
	\nabla\left(B\cdot T\right):T & =\nabla\left(A\cdot T\right):T-\frac{1}{2}(\nabla\psi\,)\times n-\frac{1}{2}\psi\,\nabla[\anti(n)]:T\notag\\
	& =\nabla\left(A\cdot T\right):T+\frac{1}{2}n\times(\nabla\psi\,\cdot T)-\frac{1}{2}\psi\,\nabla[\anti(n)]:T=\\
	& =\frac{1}{2}\nabla\left[\mathrm{anti}\left(T\cdot\widetilde{m}\cdot n\right)\cdot T\right]:T+\frac{1}{2}n\times\left[\nabla\left(\left.\langle  n,(\mathrm{sym}\,\widetilde{m})\cdot n\right.\rangle \right)\right]-\frac{1}{2}\psi\,\nabla[\anti(n)]:T.\notag
	\end{align}
	
	Moreover, using Lemma \ref{sub:vanishing}, we have
	\begin{align}
	\nabla[\anti(n)]:T=0.
	\end{align}
	Therefore, the proof is complete.

\subsection{Some  lemmas useful to understand  Mindlin and Tiersten's approach }\label{appendixlemma}

\begin{lemma}\label{lemacurlbun}
	Let $\Omega$ be an open subset of $\mathbb{R}^3$, $\Gamma$ an open subsets of $\partial \Omega$ and  $\widetilde{t},\widetilde{g}:\partial\Omega\setminus \overline{\Gamma}\rightarrow\mathbb{R}^3$ two functions. Then, the equality
	\begin{align}
	& \int_{\partial\Omega\setminus \overline{\Gamma}}\langle\widetilde{t},\delta u\rangle\, da
	+\int_{\partial\Omega\setminus \overline{\Gamma}}\langle\widetilde{g},(\id-n\otimes n)\cdot{\rm curl}\,\delta u\rangle\, da=0
	\end{align}
	holds for all  $\delta u\in C^2(\overline{\Omega})$ if and only if  $\widetilde{t}\Big|_{\partial\Omega\setminus \overline{\Gamma}}=0$ and $(\id-n\otimes n)\, \widetilde{g}\Big|_{\partial\Omega\setminus \overline{\Gamma}}=0$.
\end{lemma}
\begin{proof}
Similar calculations as in Appendix \ref{sectionincomplete2} lead to 
	\begin{align}
	& \int_{\partial \Omega\setminus\overline{\Gamma}}\langle\widetilde{t},\delta u\rangle\, da
	+\int_{\partial \Omega\setminus\overline{\Gamma}}\langle\widetilde{g},(\id-n\otimes n)\cdot{\rm curl}\,\delta u\rangle\, da\\& = \int_{\partial \Omega\setminus\overline{\Gamma}}\langle\widetilde{t}-\nabla[(\anti[(\id-n\otimes n)\cdot\widetilde{g}])\cdot T]:T,\delta u\rangle\, da\notag\\
	& -\int_{\partial \Omega\setminus\overline{\Gamma}}\langle(\id-n\otimes n)\cdot\anti[(\id-n\otimes n)\cdot\widetilde{g}]\cdot n,\nabla\delta u\cdot n\rangle\,da\notag\\
	& -\int_{\partial{\partial \Omega\setminus\overline{\Gamma}}}\langle\jump{\anti[(\id-n\otimes n)\cdot\widetilde{g}]\cdot\nu},\delta u\rangle \,ds\notag
	\end{align}
	for all variations $\delta u\in C^2(\overline{\Omega})$. Therefore
	\begin{align}
	& \int_{\partial \Omega\setminus\overline{\Gamma}}\langle\widetilde{t},\delta u\rangle\, da
	+\int_{\partial \Omega\setminus\overline{\Gamma}}\langle\widetilde{g},(\id-n\otimes n)\cdot{\rm curl}\,\delta u\rangle\, da=0
	\end{align}
	implies
	\begin{align}
	&\int_{\partial \Omega\setminus\overline{\Gamma}}\langle\widetilde{t}-\frac{1}{2}\nabla[(\anti[(\id-n\otimes n)\cdot\widetilde{g}])\cdot T]:T,\delta u\rangle\, da\notag\\
	& -\int_{\partial \Omega\setminus\overline{\Gamma}}\langle(\id-n\otimes n)\cdot\anti[(\id-n\otimes n)\cdot\widetilde{g}]\cdot n,\nabla\delta u\cdot n\rangle\,da\notag\\
	& -\int_{\partial{\partial \Omega\setminus\overline{\Gamma}}}\langle\jump{\anti[(\id-n\otimes n)\cdot\widetilde{g}]\cdot\nu},\delta u\rangle \,ds=0
	\end{align}
	for all variations $\delta u\in C^2(\overline{\Omega})$. Since $\delta u\in C^2(\overline{\Omega})$, this fact is equivalent to
	\begin{align}
	&\int_{\partial \Omega\setminus\overline{\Gamma}}\langle\widetilde{t}-\nabla[(\anti[(\id-n\otimes n)\cdot\widetilde{g}])\cdot T]:T,\delta u\rangle\, da\notag\\
	& -\int_{\partial \Omega\setminus\overline{\Gamma}}\underbrace{\langle(\id-n\otimes n)\cdot\anti[(\id-n\otimes n)\cdot\widetilde{g}]\cdot n,\nabla\delta u\cdot n\rangle}_{\begin{array}{c}
		\text{\tiny{completely\ $\delta u$}}\text{\tiny{-independent second}}\vspace{-1mm}\\
		\text{\tiny{order normal variation of gradient}}
		\end{array}}\,da=0
	\end{align}
	for all variations $\delta u\in C^2(\overline{\Omega})$.
	Having now the possibility to consider independent variations of $\delta u$ and $ \nabla\delta u\cdot n$ on $\partial \Omega\setminus\overline{\Gamma}$, we obtain that
	\begin{align}
	\widetilde{t}-\nabla[(\anti[(\id-n\otimes n)\cdot\widetilde{g}])\cdot T]:T\Big|_{\partial \Omega\setminus\overline{\Gamma}}&=0,\notag\\
	(\id-n\otimes n)\cdot\anti[(\id-n\otimes n)\cdot\widetilde{g}]\cdot n\Big|_{\partial \Omega\setminus\overline{\Gamma}}&=0.
	\end{align}
	We also remark that
	\begin{align}
	(\id-n\otimes n)\cdot\anti[(\id-n\otimes n)\cdot\widetilde{g}]\cdot n=\anti[(\id-n\otimes n)\cdot\widetilde{g}]\cdot n.
	\end{align}
	Hence, it remains to prove that from
	\begin{align}
	\widetilde{t}-\nabla[\anti[(\id-n\otimes n)\cdot\widetilde{g}]\cdot T]:T\Big|_{\partial \Omega\setminus\overline{\Gamma}}&=0,\notag\\
	\anti[(\id-n\otimes n)\cdot\widetilde{g}]\cdot n\Big|_{\partial \Omega\setminus\overline{\Gamma}}&=0
	\end{align}
	it follows that $\widetilde{t}\Big|_{\partial \Omega\setminus\overline{\Gamma}}=0$ and $(\id-n\otimes n)\cdot\widetilde{g}\Big|_{\partial \Omega\setminus\overline{\Gamma}}=0$.

	In the suitable local coordinate system $(\tau,\nu,n)$ at the boundary we obtain
	\begin{align}
	&n\otimes n=\diag(0,0,1), \qquad (\id-n\otimes n)=\diag(1,1,0), \qquad \widetilde{g}=\widetilde{g}_\tau\,\tau+\widetilde{g}_\nu\,\nu+\widetilde{g}_n\,n,\notag\\
	&(\id-n\otimes n) \widetilde{g}=\widetilde{g}_\tau\,\tau+\widetilde{g}_\nu\,\nu, \qquad [(\id-n\otimes n)\cdot\widetilde{g}]\times n=\widetilde{g}_\tau\,\nu-\widetilde{g}_\nu\,\tau.
	\end{align}
	Hence
	\begin{align}
	\anti[(\id-n\otimes n)\cdot\widetilde{g}]\cdot n\Big|_{\partial \Omega\setminus\overline{\Gamma}}&=(\widetilde{g}_\tau\,\nu-\widetilde{g}_\nu\,\tau)\Big|_{\partial \Omega\setminus\overline{\Gamma}}=0\notag\\&\Leftrightarrow \quad \widetilde{g}_\tau\Big|_{\partial \Omega\setminus\overline{\Gamma}}=0, \ \widetilde{g}_\nu\Big|_{\partial \Omega\setminus\overline{\Gamma}}=0 \quad \Leftrightarrow
	\quad (\id-n\otimes n)\cdot\widetilde{g}\Big|_{\partial \Omega\setminus\overline{\Gamma}}=0.
	\end{align}
	
	Moreover, we deduce
	\begin{align}
	\{\nabla[\anti[(\id-n\otimes n)\cdot\widetilde{g}]\cdot T]:T\}_i\Big|_{\partial \Omega\setminus\overline{\Gamma}}&=[\anti[(\id-n\otimes n)\cdot\widetilde{g}]_{ij,k}\, T_{jk}\Big|_{\partial \Omega\setminus\overline{\Gamma}}=[\anti[(\id-n\otimes n)\cdot\widetilde{g}]_{ij,k}\,\tau_{k}\tau_{j}\Big|_{\partial \Omega\setminus\overline{\Gamma}}\notag\\
	&=\langle\nabla\{[\anti[(\id-n\otimes n)\cdot\widetilde{g}]_{ij}\},\tau\rangle \,\tau_{j}\Big|_{\partial \Omega\setminus\overline{\Gamma}}=0,
	\end{align}
	since $[\anti[(\id-n\otimes n)\cdot\widetilde{g}]_{ij}\Big|_{\partial \Omega\setminus\overline{\Gamma}}=0$, for all $i,j=1,2,3$. 
	Therefore, it follows that also $\widetilde{t}\Big|_{\partial \Omega\setminus\overline{\Gamma}}=0$.
\end{proof}

\begin{lemma}
	Let $\Omega$ be an open subset of $\mathbb{R}^3$, $\Gamma$ an open subsets of $\partial \Omega$ and  $\widetilde{t}:\partial\Omega\setminus \overline{\Gamma}\rightarrow\mathbb{R}^3$, $\widetilde{g}:\partial\Omega\setminus \overline{\Gamma}\rightarrow\mathbb{R}^3$ two functions. Then, the equality
	\begin{align}
	& \int_{\partial\Omega\setminus \overline{\Gamma}}\langle\widetilde{t},\delta u\rangle\, da
	+\int_{\Gamma}\langle\widetilde{g},(\id-n\otimes n)\cdot{\rm curl}\,\delta u\rangle\, da=0
	\end{align}
	holds for all $\delta u\in  C^2(\overline{\Omega})$, $\delta u \Big|_{\Gamma}=0$ and $(\id-n\otimes n)\cdot\curl \delta u\Big|_{\partial\Omega\setminus \overline{\Gamma}}=0$ if and only if  $\widetilde{t}\Big|_{\partial\Omega\setminus \overline{\Gamma}}=0$ and $(\id-n\otimes n)\, \widetilde{g}\Big|_{\Gamma}=0$.
\end{lemma}
\begin{proof}
	The proof is similar with the  proof of Lemma \ref{lemacurlbun}.

\end{proof}

\subsection{Concluding diagrams}

\begin{figure}[h!]
	\setlength{\unitlength}{1mm}
	\begin{center}
		\begin{picture}(10,115)
		\thicklines
		
		\put(5,85){\oval(165,50)}
		\put(-72,105){\bf\footnotesize{Standard boundary conditions in the indeterminate couple stress model \cite{Mindlin62}}}
		\put(-72,100){\footnotesize{Geometric (essential) boundary conditions \ (3+2) {\bf [only weakly independent]}}}
		\put(-72,95){\footnotesize{$u\big|_{\Gamma}= u^{\rm ext}\in\mathbb{R}^{ 3}, \quad (\id-n\otimes n) \cdot \nabla  u \cdot  n\big|_{\Gamma}=(\id-n\otimes n) \cdot  \widetilde{a}^{\rm ext}\in\mathbb{R}^{3},\quad \text{or}\ \ (\id-n\otimes n) \cdot  \curl u\big|_{\Gamma}=(\id-n\otimes n) \cdot  \widetilde{b}^{\rm ext}$}}
		\put(-72,90){\footnotesize{Mechanical (traction) boundary conditions \ (3+2) }}
		\put(-72,85){\footnotesize{$\left((\sigma-\widetilde{\tau}) \cdot  n
				-\frac{1}{2} n\times \nabla[\langle n,(\sym \widetilde{ {m}}) \cdot  n\rangle]\right)\big|_{\partial \Omega\setminus\overline{\Gamma}}= t^{\rm ext}$, \qquad\qquad\qquad\qquad $\widetilde{\tau}=\Div \widetilde{\mathfrak{m}}=\frac{1}{2}\anti({\rm Div}\,  \widetilde{m})\in \so(3)$}}
		\put(-43.5,80){\footnotesize{$(\id -n\otimes n) \cdot \widetilde{ {m}} \cdot  n\big|_{\partial \Omega\setminus\overline{\Gamma}}=(\id-n\otimes n)\cdot m^{\rm ext}$}}
		\put(78,85){\footnotesize{3 bc}}
		\put(78,80){\footnotesize{2 bc}}
		\put(-72,75){\footnotesize{Boundary virtual work }}
		\put(-72,70){\footnotesize{$-\int_{\partial \Omega}\langle (\sigma-\widetilde{\tau}) \cdot  n, \delta u\rangle \,da
				-\int_{\partial\Omega}\langle \widetilde{ {m}} \cdot  n, \axl (\skw \nabla \delta u) \rangle\, da=
				0\qquad \Leftrightarrow$}}
		\put(-72,65){\footnotesize{$-\int_{\partial \Omega}\langle \Big\{(\sigma-\widetilde{\tau}) \cdot  n -\frac{1}{2} n\times \underbrace{\nabla[\langle n,(\sym \widetilde{ {m}}) \cdot  n\rangle]}_{\text{normal curvature}}\Big\}, \delta u\rangle \,da
				-\int_{\partial\Omega}\langle \widetilde{ {m}}  \cdot  n,\Big\{(\id-n\otimes n) \cdot [\axl (\skw \nabla \delta u)]\Big\}\rangle\,
				da=0$}}
		
		\put(5,52){\Huge $\Updownarrow$}
		
		\put(5,23){\oval(165,50)}
		\put(-72,45){\bf\footnotesize{Standard boundary conditions in the indeterminate couple stress model,  index-format}}
		\put(-72,40){\footnotesize{Geometric (essential) boundary conditions \ (3+2) {\bf [only weakly independent]}}}
		\put(-72,35){\footnotesize{$u_i\big|_{\Gamma}= u^{\rm ext}_{,i}\in\mathbb{R}^{ 3}, \qquad \ \left(\epsilon_{ikl}u_{l,k}-\epsilon_{jkl}u_{l,k}n_jn_i\right)\big|_{\Gamma}= \epsilon_{ikl} u^{\rm ext}_{l,k}-\epsilon_{jkl} u^{\rm ext}_{l,k}n_jn_i,$}}
		\put(-71.2,30){\footnotesize{$\quad \ \ \qquad \quad \qquad \quad \quad \text{or}\ \ \ \left(u_{i,k}n_k-u_{j,k}n_kn_jn_i\right)\big|_{\Gamma}= u^{\rm ext}_{i,k}n_k- u^{\rm ext}_{j,k}n_kn_jn_i$}}
		\put(-72,25){\footnotesize{Mechanical (traction) boundary conditions \ (3+2) }}
		\put(-72,20){\footnotesize{$\left((\sigma_{ij}-\widetilde{\tau}_{ij})\, n_j
				-\frac{1}{2} \epsilon_{ikl}n_k (\widetilde{m}_{ij}n_in_j)_{,l}\right)\big|_{\partial \Omega\setminus\overline{\Gamma}}= t^{\rm ext}_i$, \qquad\qquad\qquad\qquad $\widetilde{\tau}_{ij}=\frac{1}{2}\epsilon_{jik}  \widetilde{m}_{kl,l}\in \so(3)$}}
		\put(-52.5,15){\footnotesize{$\left(\widetilde{ {m}}_{ik}n_k-\widetilde{ {m}}_{jk}n_kn_jn_i\right)\big|_{\partial \Omega\setminus\overline{\Gamma}}= m^{\rm ext}_i- m^{\rm ext}_jn_jn_i$}}
		\put(78,20){\footnotesize{3 bc}}
		\put(78,15){\footnotesize{2 bc}}
		\put(-72,10){\footnotesize{Boundary virtual work }}
		\put(-72,5){\footnotesize{$-\int_{\partial \Omega}\left((\sigma_{ij}-\widetilde{\tau}_{ij})\, n_j
				\right) \delta u_i \,da
				-\int_{\partial\Omega}\langle \widetilde{ {m}} \cdot  n, \axl (\skw \nabla \delta u) \rangle\, da=
				0\qquad \Leftrightarrow$}}
		\put(-72,0){\footnotesize{$-\int_{\partial \Omega}\left((\sigma_{ij}-\widetilde{\tau}_{ij})\, n_j
				-\frac{1}{2} \epsilon_{ikl}n_k (\widetilde{m}_{ij}n_in_j)_{,l}\right) \delta u_i \,da
				-\frac{1}{2}\,\int_{\partial\Omega}\left(\widetilde{ {m}}_{ik}n_k-\widetilde{ {m}}_{jk}n_kn_jn_i\right)\,\left(\epsilon_{ikl}\delta u_{l,k}-\epsilon_{jkl}\delta u_{l,k}n_jn_i\right)\,
				da=0$}}
		
		\end{picture}
	\end{center}
	\caption{The standard boundary conditions in the indeterminate couple stress model which have been employed hitherto by all authors to our knowledge. The virtual displacement is denoted by $\delta u\in C^\infty(\overline{\Omega})$. The number of traction boundary conditions is correct, but the split into independent variations at the boundary is no taken to it's logical end.}\label{limitmodel003}
\end{figure}

\begin{figure}[h!]
	\setlength{\unitlength}{1mm}
	\begin{center}
		\begin{picture}(10,100)
		\thicklines

		\put(5,87){\oval(165,62)}
		\put(-72,115){\bf\footnotesize{Boundary conditions in the indeterminate couple stress model in terms of gradient elasticity}}
		\put(-72,110){\bf\footnotesize{ and third order moment tensors}}
		\put(-72,105){\footnotesize{Geometric (essential) boundary conditions \ (3+2) {\bf [strongly independent]}}}
		\put(-72,100){\footnotesize{$u\big|_{\Gamma}= u^{\rm ext}\in\mathbb{R}^{ 3}, \quad (\id-n\otimes n) \cdot \nabla  u \cdot  n\big|_{\Gamma}=(\id-n\otimes n) \cdot \widetilde{a}
				^{\rm ext}\in\mathbb{R}^{3},\quad \text{or}\ \ (\id-n\otimes n) \cdot \curl u\big|_{\Gamma}=(\id-n\otimes n)    \cdot \widetilde{b}^{\rm ext}$}}
		\put(-72,95){\footnotesize{Mechanical (traction) boundary conditions \ (3+2)}}
		\put(-69,90){\footnotesize{$\left((\sigma-\Div \widetilde{\mathfrak{m}}) \cdot   n-\nabla [(\widetilde{\mathfrak{m}} \cdot  n) \cdot (\id-n\otimes n)]: (\id-n\otimes n)\right)\big|_{\partial \Omega\setminus\overline{\Gamma}}= t^{\rm ext}$,\quad \quad  $\widetilde{\mathfrak{m}}=D_{\nabla\nabla u}[\widetilde{W}_{\mathrm{curv}}(\nabla [\axl (\skw \nabla u)])]$}}
		\put(-24.5,85){\footnotesize{$(\id -n\otimes n) \cdot [\widetilde{\mathfrak{m}} \cdot   n] \cdot  n\big|_{\partial \Omega\setminus\overline{\Gamma}}=(\id-n\otimes n) \cdot  m^{\rm ext}$}}
		\put(-7.2,80){\footnotesize{$ \jump{(\widetilde{\mathfrak{m}} \cdot   n) \cdot   \, \nu}\big|_{\partial \Gamma}= \pi^{\rm ext}$,\ \ \qquad\quad ``edge line force'' on $\partial \Gamma$}}
		\put(78,90){\footnotesize{3 bc}}
		\put(78,85){\footnotesize{2 bc}}
		\put(78,80){\footnotesize{3 bc}}
		\put(-72,75){\footnotesize{Boundary virtual work}}
		\put(-72,70){\footnotesize{$-\int_{\partial \Omega}\langle (\sigma-\Div \widetilde{\mathfrak{m}}) \cdot  n, \delta u\rangle \,da
				-\int_{\partial\Omega}\langle \widetilde{\mathfrak{m}} \cdot   n, \nabla \delta u \rangle\, da=
				0\qquad \Leftrightarrow$}}
		\put(-72,65){\footnotesize{$ -\int_{\partial \Omega}\langle (\sigma-\Div \widetilde{\mathfrak{m}}) \cdot  n-\nabla [(\widetilde{\mathfrak{m}} \cdot   n) \cdot (\id-n\otimes n)]:(\id-n\otimes n), \delta u\rangle \,da-\int_{\partial \Omega}\langle (\id -n\otimes n) \cdot [\widetilde{\mathfrak{m}} \cdot  n] \cdot  n,  \nabla \delta u \cdot  n\rangle da$}}
		\put(-72,60){\footnotesize{$ -
				\int_{\partial \Gamma}\langle \jump{(\widetilde{\mathfrak{m}} \cdot  n)  \cdot  \nu}, \delta u\rangle ds=0$}}

		\put(5,49){{\huge $\Updownarrow$} \  \footnotesize{equivalent}}
		
		\put(5,21){\oval(165,50)}
		\put(-72,40){\footnotesize{\bf Boundary conditions  (3+2) in the indeterminate couple stress model in terms of gradient elasticity,}}
		\put(-72,35){\bf\footnotesize{  third order  moment tensors,  and written in indices}}
		\put(-72,30){\footnotesize{Geometric (essential) boundary conditions \ (3+2) {\bf [strongly independent]}}}
		\put(-72,25){\footnotesize{$u_i\big|_{\Gamma}= u^{\rm ext}_{ i}, \qquad (u_{i,k}n_k-u_{j,k}n_jn_k n_i)\big|_{\Gamma}= u^{\rm ext}_{i,k}n_k- u^{\rm ext}_{j,k}n_in_jn_k  $,}}
		\put(-79,20){\footnotesize{$\hspace{2.25cm} \text{or}\ \ (\epsilon_{ikl}u_{l,k}-\epsilon_{jkl} u_{l,k}n_jn_i)\big|_{\Gamma}=(\epsilon_{ikl} u^{\rm ext}_{l,k}-\epsilon_{jkl}  u^{\rm ext}_{l,k}n_jn_i) $,}}
		\put(-72,15){\footnotesize{Mechanical (traction) boundary conditions \ (3+2)}}
		\put(-72,10){\footnotesize{$\left[\left(\sigma_{ij}-\widetilde{\mathfrak{m}}_{ijk,k}\right)\, n_{j}-\left(\widetilde{\mathfrak{m}}_{ipk}\,n_k-\widetilde{\mathfrak{m}}_{ijk}\, n_{k}n_jn_p\right)_{,h}\left(\delta_{ph}-n_p n_h\right)\right]\big|_{\partial \Omega\setminus\overline{\Gamma}}= t^{\rm ext}_i$,\quad  $\widetilde{\mathfrak{m}}_{ijk}=D_{u_{,ijk}}\widetilde{W}_{\mathrm{curv}}\left(\left(\epsilon_{ijk}u_{k,j}\right)_{,m}\right)$}}
		\put(-25.5,5){\footnotesize{$\left(\widetilde{\mathfrak{m}}_{ijp}\, n_{j}-\widetilde{\mathfrak{m}}_{pjk}\, n_{k}n_{j}n_i\right)\,n_p\big|_{\partial \Omega\setminus\overline{\Gamma}}= m^{\rm ext}_i- m^{\rm ext}_pn_pn_i$}}
		\put(-1.5,0){\footnotesize{$ \jump{\widetilde{\mathfrak{m}}_{pjk}\, n_{k}\, \nu_j}\big|_{\partial \Gamma}= \pi^{\rm ext}_p$,\ \ \qquad\quad ``edge line force'' on $\partial \Gamma$}}
		\put(79,10){\footnotesize{3 bc}}
		\put(79,5){\footnotesize{2 bc}}
		\put(79,0){\footnotesize{3 bc}}
		\end{picture}
	\end{center}
	\caption{The standard strongly independent boundary conditions in the indeterminate couple stress model in terms of a third order couple stress tensor coming from full gradient elasticity.  The virtual displacement is denoted by $\delta u\in C^\infty(\overline{\Omega})$.  The summation convention was used in index notations. }\label{limitmodel03order}
\end{figure}

\newpage

\begin{figure}[h!]
	\setlength{\unitlength}{1mm}
	\begin{center}
		\begin{picture}(10,170)
		\thicklines

		\put(5,159){\oval(165,72)}
		\put(-72,190){\bf \footnotesize{Strongly independent boundary conditions in the indeterminate couple stress model}}
		\put(-72,185){\footnotesize{Geometric (essential) boundary conditions \ (3+2)}}
		\put(-72,180){\footnotesize{$u\big|_{\Gamma}= u^{\rm ext}\in\mathbb{R}^{ 3}, \  (\id-n\otimes n) \cdot  \nabla  u \cdot  n\big|_{\Gamma}=(\id-n\otimes n) \cdot a^{\rm ext}\in\mathbb{R}^{3},\  \text{or}\  (\id-n\otimes n) \cdot \curl u\big|_{\Gamma}=(\id-n\otimes n) \cdot   b^{\rm ext}$}}
		\put(-72,175){\footnotesize{Mechanical (traction) boundary conditions \ (3+2)}}
		\put(-71,170){\footnotesize{$\Big((\sigma-\widetilde{\tau}) \cdot  n
				-\frac{1}{2} n\times \nabla[\langle n,(\sym \widetilde{ {m}}) \cdot  n\rangle]$}}
		\put(-58,165){\footnotesize{$-\frac{1}{2}\nabla[(\anti[(\id-n\otimes n) \cdot  \widetilde{ {m}} \cdot  n]) \cdot (\id-n\otimes n)]: (\id-n\otimes n)\Big)\big|_{\partial \Omega\setminus\overline{\Gamma}}= t^{\rm ext}$,}}
		\put(-31.5,160){\footnotesize{$(\id -n\otimes n) \cdot  \anti[(\id-n\otimes n) \cdot  \widetilde{ {m}} \cdot   n] \cdot  n\big|_{\partial \Omega\setminus\overline{\Gamma}}=(\id-n\otimes n) \cdot  m^{\rm ext}$}}
		\put(3.5,155){\footnotesize{ $\jump{\anti[\widetilde{ {m}}  \cdot   n] \cdot  \nu}\big|_{\partial \Gamma}= \pi^{\rm ext}$,\ \ \qquad\quad ``edge line force'' on $\partial \Gamma$}}
		\put(81.5,165){\footnotesize{3 bc}}
		\put(81.5,160){\footnotesize{2 bc}}
		\put(81.5,155){\footnotesize{3 bc}}
		\put(-72,150){\footnotesize{Boundary virtual work}}
		\put(-72,145){\footnotesize{$-\int_{\partial \Omega}\langle (\sigma-\widetilde{\tau}) \cdot  n, \delta u\rangle \,da
				-\int_{\partial\Omega}\langle \widetilde{ {m}} \cdot   n, \axl (\skw \nabla \delta u) \rangle\, da=
				0\qquad \Leftrightarrow$}}
		\put(-72,140){\footnotesize{$-\int_{\partial \Omega}\langle \Big\{(\sigma-\widetilde{\tau}) \cdot   n -\frac{1}{2} n\times \underbrace{\nabla[\langle n,(\sym \widetilde{ {m}}) \cdot  n\rangle]}_{\text{normal curvature}}\Big\}, \delta u\rangle \,da
				+\int_{\partial\Omega}\langle \widetilde{ {m}}\cdot n,\Big\{(\id-n\otimes n) \cdot  [\axl (\skw \nabla \delta u)]\Big\}\rangle\,
				da=0$}}

		\put(-72,130){\footnotesize{$ -\int_{\partial \Omega}\langle (\sigma-\widetilde{\tau}).\, n
				-\frac{1}{2} n\times \nabla[\langle n,(\sym \widetilde{ {m}}) \cdot  n\rangle]-\frac{1}{2}\nabla[(\anti[(\id-n\otimes n)  \cdot \widetilde{ {m}} \cdot  n]) \cdot (\id-n\otimes n)]: (\id-n\otimes n), \delta u\rangle \,da\notag$}}
		
		\put(-72,125){\footnotesize{$ \qquad-\frac{1}{2}
				\int_{\partial \Omega}\langle (\id -n\otimes n)  \cdot \anti[(\id-n\otimes n) \cdot  \widetilde{ {m}} \cdot  n] \cdot  n,  \nabla \delta u \cdot n \rangle da-\frac{1}{2}
				\int_{\partial\Gamma}\langle \jump{\anti[\widetilde{ {m}} \cdot  n] \cdot \nu}, \delta u\rangle ds\notag =0$}}

		\put(5,115){{\huge $\Updownarrow$} \  \footnotesize{equivalent}}

		\put(5,82){\oval(165,60)}
		\put(-72,105){\bf \footnotesize{Equivalent strongly independent  boundary conditions in the indeterminate couple stress model}}
		\put(-72,100){\footnotesize{Geometric (essential) boundary conditions \ (3+2)}}
		\put(-72,95){\footnotesize{$u\big|_{\Gamma}= u^{\rm ext}\in\mathbb{R}^{ 3}, \  (\id-n\otimes n)  \cdot \nabla  u \cdot  n\big|_{\Gamma}=(\id-n\otimes n) \cdot   a^{\rm ext}\in\mathbb{R}^{3},\  \text{or}\ (\id-n\otimes n) \cdot \curl u\big|_{\Gamma}=(\id-n\otimes n) \cdot  b^{\rm ext}$}}
		\put(-72,90){\footnotesize{Mechanical (traction) boundary conditions \ (3+2)}}
		\put(-72,85){\footnotesize{$\left((\sigma-\widetilde{\tau}) \cdot  n-\frac{1}{2}\nabla[(\anti(\widetilde{ {m}} \cdot  n)) \cdot (\id-n\otimes n)]: (\id-n\otimes n)\right)\big|_{\partial \Omega\setminus\overline{\Gamma}}= t^{\rm ext}$,}}
		\put(-28.5,80){\footnotesize{$(\id -n\otimes n) \cdot \anti[\widetilde{ {m}} \cdot  n] \cdot  n\big|_{\partial \Omega\setminus\overline{\Gamma}}=(\id-n\otimes n) \cdot  m^{\rm ext}$}}
		\put(-32.9,75){\footnotesize{$\hspace{2.23cm} \jump{\anti[\widetilde{ {m}} \cdot  n] \cdot   \, \nu}\big|_{\partial \Gamma}= \pi^{\rm ext}$,\ \ \qquad\quad ``edge line force'' on $\partial \Gamma$}}
		\put(78,85){\footnotesize{3 bc}}
		\put(78,80){\footnotesize{2 bc}}
		\put(78,75){\footnotesize{3 bc}}
		\put(-72,70){\footnotesize{Boundary virtual work}}
		\put(-72,65){\footnotesize{$-\int_{\partial \Omega}\langle (\sigma-\widetilde{\tau}) \cdot  n, \delta u\rangle \,da
				-\int_{\partial\Omega}\langle \widetilde{ {m}}\cdot n, \axl (\skw \nabla \delta u) \rangle\, da=
				0\qquad \Leftrightarrow$}}
		\put(-72,60){\footnotesize{$ -\int_{\partial \Omega}\langle (\sigma-\widetilde{\tau}) \cdot  n-\frac{1}{2}\nabla[(\anti(\widetilde{ {m}} \cdot  n)) \cdot (\id-n\otimes n)]: (\id-n\otimes n), \delta u\rangle \,da$}}
		\put(-72,55){\footnotesize{$ \qquad-\frac{1}{2}
				\int_{\partial \Omega}\langle (\id -n\otimes n)\anti(\widetilde{ {m}} \cdot  n) \cdot  n,  \nabla \delta u \cdot n \rangle da-\frac{1}{2}
				\int_{\partial \Gamma}\langle \jump{\anti[\widetilde{ {m}} \cdot  n] \cdot   \nu}, \delta u\rangle ds=0$}}

		\put(5,45){{\huge $\Updownarrow$} \  \footnotesize{equivalent}}

		\put(5,18){\oval(165,45)}
		\put(-72,35){\bf \footnotesize{Alternative equivalent strongly independent boundary conditions,  index-format}}
		\put(-72,30){\footnotesize{Geometric (essential) boundary conditions \ (3+2)}}
		\put(-72,25){\footnotesize{$u_i\big|_{\Gamma}= u^{\rm ext}_{i}\in\mathbb{R}^{ 3}, \qquad \ \left(\epsilon_{ikl}u_{l,k}-\epsilon_{jkl}u_{l,k}n_jn_i\right)\big|_{\Gamma}= a^{\rm ext}_i$}}
		\put(-71.4,20){\footnotesize{$\ \quad \ \ \qquad \qquad \qquad \quad \text{or}\ \ \left(u_{i,k}n_k-u_{j,k}n_kn_jn_i\right)\big|_{\Gamma}= b^{\rm ext}_{i}$}}
		\put(-72,15){\footnotesize{Mechanical (traction) boundary conditions \ (3+2)}}
		\put(-72,10){\footnotesize{$\left((\sigma_{ij}-\widetilde{\tau}_{ij})\, n_j+\frac{1}{2}(\epsilon_{ihk}\widetilde{m}_{ks}n_s-\epsilon_{ijk}\widetilde{m}_{ks}n_sn_jn_h)_{,p}(\delta_{hp}-n_hn_p)\right)\big|_{\partial \Omega\setminus\overline{\Gamma}}= t^{\rm ext}_i$,}}
		\put(-52,5){\footnotesize{$\hspace{2.23cm}(\epsilon_{ipk}\widetilde{m}_{ks}n_s-\epsilon_{jpk}\widetilde{m}_{ks}n_sn_jn_i)\,n_p\big|_{\partial \Omega\setminus\overline{\Gamma}}= m^{\rm ext}_i- m^{\rm ext}_pn_pn_i,$}}
		\put(-59.5,0){\footnotesize{$ \hspace{5.72cm}\jump{\epsilon_{ipk}\widetilde{m}_{ks}n_s\, \nu_p}\big|_{\partial \Gamma}= \pi^{\rm ext}_i$,\ \ \qquad\quad``edge line force''  on $\partial \Gamma$}}
		\put(78,10){\footnotesize{3 bc}}
		\put(78,5){\footnotesize{2 bc}}
		\put(78,0){\footnotesize{3 bc}}
		\end{picture}
	\end{center}
	\caption{The possible boundary conditions in the indeterminate couple stress model. The equivalence of the geometric boundary condition is clear. The virtual displacement is denoted by $\delta u\in C^\infty(\overline{\Omega})$. }\label{limitmodel00}
\end{figure}

\end{footnotesize}

\end{document}